\documentclass[11pt,letterpaper]{article}
\usepackage[margin=1in]{geometry}
\usepackage[utf8]{inputenc}
\usepackage{amsmath}
\usepackage{amssymb}
\usepackage{amsthm}
\usepackage{dsfont}
\usepackage{color}
\usepackage{amsmath}
\usepackage{amsfonts}
\usepackage{amssymb}
\usepackage{bbm}
\usepackage{lineno}
\usepackage[colorlinks=true,linkcolor=red,citecolor=blue]{hyperref}
\usepackage{relsize}
\usepackage{pst-node}
\usepackage[procnumbered, linesnumbered,algoruled]{algorithm2e}
\usepackage[compact]{titlesec}
\usepackage{tikz}
\usepackage{mdframed}
\usepackage{tcolorbox}
\usepackage[shortlabels]{enumitem}
\tcbuselibrary{theorems}

\newtcbtheorem[]{boxdef}{Definition}{colback=green!5,colframe=green!35!black,fonttitle=\bfseries}{def}

\usetikzlibrary{positioning}
\usetikzlibrary {shapes,matrix}
\usetikzlibrary{arrows}
\tikzset{
	semithick,
	node distance = 2cm,
	dot/.style={circle,fill,inner sep=2pt}
}
\tikzset{
	side by side/.style 2 args={
		line width=2pt,
		#1,
		postaction={
			clip,postaction={draw,#2}
		}
	}
}
\tikzstyle{every state}=[draw = black,thick,fill = white,minimum size = 4mm]
\tikzstyle{selected edge} = [draw,line width=2pt,-,red!50]
\tikzset{
	edge/.style={->,> = latex'}
}
\usepackage{hyperref}
\newcommand{\abs}[1]{{\left|#1 \right|}}

\newcommand{\ceil}[1]{{\left\lceil#1  \right\rceil}}
\renewcommand{\comment}[1]{}
\usepackage[nameinlink]{cleveref}
\usepackage[procnumbered, linesnumbered,algoruled]{algorithm2e}

\newcommand{\cA}{{\mathcal{A}}}

\newcommand{\ord}{{\textnormal{\textsf{ord}}}}
\newcommand{\cB}{{\mathcal{B}}}

\newcommand{\var}{{\textnormal{\textsf{var}}}}

\newcommand{\Adj}{{\textnormal{\textsf{Adj}}}}

\newcommand{\cI}{{\mathcal{I}}}
\newcommand{\cF}{{\mathcal{Q}}}
\newcommand{\cC}{{\mathcal{C}}}

\newcommand{\N}{{\mathbb{N}}}

\newcommand{\OPT}{\textnormal{OPT}}

\newcommand{\eps}{{\varepsilon}}

\newcommand{\SAT}{{{\textnormal{\textsf{SAT}}}}}
\newcommand{\CSP}{{{\textnormal{\textsf{CSP}}}}}
\newcommand{\E}{{\mathbb{E}}}

\newcommand{\floor}[1]{\left\lfloor #1 \right\rfloor}

\DeclareMathOperator*{\Var}{Var}

\newcommand{\MaxPar}{\mathrm{Par}}
\crefname{claim}{claim}{claims}
\crefname{cor}{corollary}{corollaries}

\newcommand{\LP}{\textsf{LP}}
\newcommand{\bounded}{\textsf{2-bounded}}
\newcommand{\unbounded}{\textsf{2-unbounded}}
\newcommand{\varpiup}{\varpi^{\textsf{up}}}
\newcommand{\varpidown}{\varpi^{\textsf{down}}}

\newif\ifapprox
\approxtrue

\newtheorem{thm}{Theorem}[section]

\newtheorem{con}[thm]{Conjecture}
\newtheorem{obs}[thm]{Observation}
\newtheorem{fact}[thm]{Fact}

\newtheorem{cor}[thm]{Corollary}
\newtheorem{lemma}[thm]{Lemma}
\newtheorem{definition}[thm]{Definition}
\newtheorem{theorem}[thm]{Theorem}

\newtheorem{claim}[thm]{Claim}
\Crefname{obs}{Observation}{Observations}

\newtheorem{reduction}[thm]{Reduction}

\begin{document}

\def \II   {{\mathcal I}}
\newcommand{\one}{\mathbbm{1}}
	\def\claimproof{\proof}
\def\endclaimproof{\hfill$\square$\\}

\renewcommand\qedsymbol{$\blacksquare$}
%\usepackage{hyperref}

%\begin{titlepage}
\title{
	Fine Grained Lower Bounds for Multidimensional Knapsack
%Nearly Tight Lower Bounds for Multidimensional Knapsack %via 2-CSP
}

\author{Ilan Doron-Arad\thanks{\texttt{ilan.d.a.s.d@gmail.com}} \and Ariel Kulik\thanks{\texttt{kulik@cs.technion.ac.il}}
%\author[1]{Roy Schwartz\thanks{\texttt{schwartz@cs.technion.ac.il}}}
%\author[1]{Hadas Shachnai\thanks{\texttt{hadas@cs.technion.ac.il}}}
%\affil[1]{Computer Science Department, Technion, Haifa 3200003, Israel}
%\author{Anonymous Submission to FOCS}
%\and
%Ariel Kulik
%\and
\and Pasin Manurangsi\thanks{Google Research, Thailand. Email: \texttt{pasin@google.com}}}
\maketitle

\begin{abstract}
We study the $d$-{\em dimensional knapsack problem}. We are given a set of items, each with a $d$-dimensional cost vector and a profit, along with a $d$-dimensional budget vector. The goal is to select a set of items that do not exceed the budget in all dimensions and maximize the total profit. A polynomial-time approximation scheme (PTAS) with running time $n^{\Theta(d/\eps)}$ has long been known for this problem, where $\eps$ is the error parameter and $n$ is the encoding size. Despite decades of active research, the best running time of a PTAS has remained $O(n^{\lceil d/\eps \rceil - d})$. Unfortunately, existing lower bounds only cover the special case with two dimensions $d = 2$, and do not answer whether there is a $n^{o(\frac{d}{\eps})}$-time PTAS for larger values of $d$. 
The status of exact algorithms is similar: there is a simple $O\left(n \cdot W^d\right)$-time (exact) dynamic programming algorithm, where $W$ is the maximum budget, but there is no  lower bound which  explains the strong exponential dependence on $d$.

In this work, we show that the running times of the best-known PTAS and exact algorithm cannot be improved up to a polylogarithmic factor assuming Gap-ETH. Our techniques are based on a robust reduction from 2-CSP, which {\em embeds} 2-CSP constraints into a desired number of dimensions, exhibiting tight trade-off between $d$ and $\eps$ for most regimes of the parameters.  Informally, we obtain the following main results for $d$-dimensional knapsack.

\begin{itemize}
	\item No $n^{o\left(\frac{d}{\eps}\cdot \frac{1}{ (\log(d/\eps))^2}\right)}$-time $(1-\eps)$-approximation for every $\eps = O \left(\frac{1}{\log d}\right)$. 
	
	\item No $(n+W)^{o\left(\frac{d}{ \log d}\right)}$-time exact algorithm (assuming ETH).   
	
	\item No $n^{o\left(\sqrt{d}\right)}$-time $(1-\eps)$-approximation for constant $\eps$.

	\item 
	$\left((d \cdot \log W)^{O(d^2)} + n^{O(1)}\right)$-time $\Omega\left(\frac{1}{\sqrt{d}}\right)$-approximation and a matching $n^{O(1)}$-time lower~bound.
\end{itemize}

\end{abstract}

\section{Introduction}
\label{sec:introduction}

In the $0/1$-Knapsack problem, we are given a set of items each with a cost and a profit, together with a budget. The goal is to select a subset of items whose total cost does not exceed the budget while maximizing the total profit. Knapsack is one of the most well-studied problems in combinatorial optimization. Indeed, a {\em polynomial-time approximation scheme} (PTAS) for the problem~\cite{Johnson73,Sahni75} was among the first PTASes to have ever been discovered in the field, and is often the first PTAS taught in computer science courses (see e.g.~\cite[Section 3.1]{WS-book} and~\cite[Chapter 35.5]{CLRS}).
While such a PTAS has been known to exist since the 70s~\cite{Johnson73,Sahni75}, the problem remains an active topic of research to this day~\cite{IbarraK75,Lawler79,KellererP99,KellererP04,rhee2015faster,Chan18a,Jin19,DengJM23,Mao23,CLMZ23}, with the best known PTAS that gives a $(1-\eps)$-approximation in time $\tilde{O}(n + 1 / \eps^2)$~\cite{Mao23,CLMZ23}. This running time is essentially tight assuming $(\min, +)$-convolution cannot be solved in truly sub-quadratic time~\cite{CyganMWW19,KunnemannPS17}.

Numerous generalizations of Knapsack have been considered over the years. One of the most natural and well-studied versions is the so-called {\em $d$-dimensional knapsack} problem. Here each cost is a $d$-dimensional vector. The problem remains the same, except that the total cost must not exceed the budget in every coordinate. The problem is defined more formally below.

\setlength{\fboxsep}{2pt} % Adjust the distance between text and box
\noindent\fbox{%
	\parbox{0.98\linewidth}{% Adjust the width of the box
		\begin{definition} %\leavevmode\vspace{-\baselineskip} 
			\label{def:dKP} {\bf { ($d$-Dimensional Knapsack)}} \newline
			{\bf Input:} An instance $\cI = (I,p,c,B)$ consisting of:
			\begin{itemize}
				\item a set of items $I$,
				\item a profit function $p:I \rightarrow \mathbb{N}$,
				%\item a number of dimensions $d \in \mathbb{N}$
				\item a $d$-dimensional cost function $c:I \rightarrow \mathbb{N}^d$,
				\item a $d$-dimensional budget $B \in \mathbb{N}^d$.  
			\end{itemize}
			
			\noindent {\bf Solution:} 	A subset $S \subseteq I$ such that for all $j \in [d]$ it holds that $c_j(S) := \sum_{i \in S} c_j(i) \leq B_j$. %, where $c_j(S),B_j$ are the $j$-th entry of the vectors $c(S), B$, respectively.
			%The {\em profit} of a solution $S$ is~$p(S)$.
			\newline
			
			\noindent  {\bf Objective:} Find a solution of maximum profit, where the profit of $S$ is $p(S) := \sum_{i \in S} p(i)$ 
			
		\end{definition}
	}%
}

\leavevmode\vspace{1pt} 

{\em Multidimensional (vector) knapsack} is the family of $d$-dimensional knapsack problems for all $d \geq 1$. The study of $d$-dimensional knapsack dates back to the 50s~\cite{lorie1955three,markowitz1957solution,weingartner1967methods} and it has been used to model problems arising in many contexts, including resource allocation, delivery system, budgeting, investment, and voting (see the survey \cite{Freville04} for a more comprehensive review). Given its importance, there have been several algorithms and hardness results devised for the problem over the years. We will summarize the main theoretical results below. We note that many heuristics for the problems have been proposed (e.g.,~\cite{alaya2004ant,balev2008dynamic,boyer2009heuristics,ke2010ant,6297286,SabbaC14,AzadRF14,YanCLZ19,Rezoug2019}), although these do not provide theoretical guarantees and thus will be omitted from the subsequent discussions.

On the hardness front, the decision version of the Knapsack problem ($d = 1$) generalizes the Subset Sum problem, which is one of Karp's 21 NP-complete problems~\cite{Karp72}. The $d = 1$ case admits a {\em fully polynomial time approximation scheme} (FPTAS), i.e., a PTAS that runs in time $(n/\eps)^{O(1)}$. In contrast, for $d = 2$, it has also long been known that no FPTAS exists unless $\textnormal{P = NP}$~\cite{KORTE1981415,MagazineC84}. 
This has more recently been strengthened in \cite{KulikS10}, who show that, even when $d = 2$, the problem does not even admit an efficient PTAS (EPTAS), i.e., one that runs in time $f(\epsilon) \cdot n^{O(1)}$ time, assuming $\textnormal{W}[1]\neq \textnormal{FPT}$. This in turn was improved by \cite{JansenLL16} to rule out any PTAS that runs in time $n^{o(1/\eps)}$ even for $d = 2$, assuming ETH. 

%KORTE1981415,MagazineC84,KulikS10,JansenLL16

Meanwhile, on the algorithmic front, a PTAS for $d$-dimensional knapsack with running time $n^{O(d/\eps)}$ was first described by Chandra et al.~\cite{ChandraHW76}\footnote{Strictly speaking, Chandra et al.'s algorithm is for the \emph{unbounded} version where each item can be picked multiple times. However, it can be extended to the bounded case too~\cite{OM80}.}%\pasin{I actually can't access \cite{OM80}, but this is cited in \cite{Freville04}} \ariel{It seems it is only exists in the archives of Waterlo, I update the bib entry. }}
, not long after a PTAS for Knapsack was discovered in~\cite{Johnson73,Sahni75}. The $d$-dimensional knapsack has been a challenging research topic ever since: Over the past nearly 50 years, all improvements on Chandra et al.'s PTAS come just in the constant in the exponent~\cite{FRIEZE1984100,CapraraKPP00}, with the best running time known being $O(n^{\lceil d/\eps \rceil - d})$~\cite{CapraraKPP00}. Unfortunately, this barrier cannot be explained by aforementioned existing lower bounds. For example, they do not rule out $n^{O(d + 1/\eps)}$-time PTAS. This brings us to the main question of this paper:
\begin{center}
{\em Is there a \textnormal{PTAS} for $d$-dimensional knapsack with running time $n^{o(d/\eps)}$?}
\end{center}

 %\pasin{Figure out how to mention \cite{ChekuriK2004} in the setting where $d$ is part of the input.}

\subsection{Our Results}

Our main result is the resolution of the above question in the negative, up to a polylogarithmic factor in the exponent. Unlike previous lower bounds \cite{KulikS10,JansenLL16} which focused on two dimensions $d = 2$, or a non-parametrized number of dimensions \cite{ChekuriK2004}, our lower bound gives a nearly tight {\em trade-off} between $d$ and $\frac{1}{\eps}$.  
%as stated more formally below. 
In the following results, $n$ denotes the encoding size of the instance. 

%\ariel{we should clearly state that $n$ is used as the encoding size of the input (and not the number of items). }
\begin{theorem}\label{thm:main}
	Assuming \textnormal{Gap-ETH}, there exist constants $\zeta, \chi,d_0 > 0$ such that for every integer $d>d_0 $ and every $\eps \in \left(0, \frac{\chi}{\log d}\right)$, there is no $(1 - \eps)$-approximation algorithm for $d$-dimensional knapsack that runs in time $O\left(n^{\frac{d}{\eps} \cdot \frac{\zeta}{(\log(d/\eps))^2}}\right)$.
\end{theorem}
For example, \Cref{thm:main} implies that there is no algorithm which for every $d$ and $\eps<0.1$ returns a $(1-\eps)$-approximate  solution for a $d$-dimensional knapsack instance in time $n^{O \left( d+ \frac{1}{\eps}\right)}$, or $n^{O \left( \frac{d^{1-\delta}}{\eps}\right)}$, or even $n^{O \left( \left(d+ \frac{1}{\eps}\right)^{2-\delta}\right)}$  for any~$\delta\in (0,1)$ (by setting $d = \frac{1}{\eps}$). Additionally, as the lower bound in \Cref{thm:main} considers specific values  of $d$ and $\eps$, the result also rules out algorithms with running time of the form $f(d,\eps) \cdot n^{\frac{d}{\eps} \cdot \frac{\chi}{(\log(d/\eps))^2}}$ for an arbitrary function $f$. (A similar statement is also true for each of the following lower bounds.)   

 %$n^{O \left( d+ \frac{1}{\eps}\right)}$. Similarly, the theorem also rules out the existence of a similar algorithm with running time $n^{O \left( \left(d+ \frac{1}{\eps}\right)^{2-\delta}\right)}$  for any~$\delta\in (0,1)$. 
% or in time $n^{O \left( \frac{d^{1-\delta}}{\eps}\right)}$ for some $\delta\in (0,1)$.

\paragraph{Exact Algorithms} It is fairly easy to attain an exact algorithm for $d$-dimensional knapsack via dynamic programming. Such an algorithm runs in time $O(n\cdot W^d)$, where $W = \max_{j\in [d]} B_j$ is the maximum budget in one of the dimensions of the instance (see, .e.g., \cite{kellerer2004multidimensional}). However, no existing lower bounds can explain the exponential increase in the running time as $d$ increases. Interestingly, our techniques can be used to obtain hard $d$-dimensional knapsack instances with small $W$ (with respect to the number of items), which we use  to lower bound the running time of exact pseudo-polynomial algorithms. 
 The following theorem indicates that the naive dynamic programming cannot be significantly improved.  
\begin{theorem}
	\label{thm:VKexact}
	Assuming \textnormal{ETH},  there exist constants  $\zeta,d_0 >0$, such that  for every integer $d>d_0$ there is  no algorithm that 
	solves $d$-dimensional knapsack exactly in   time $O \left( \big(n+W \big)^{\zeta \cdot \frac{d}{ \log(d)}} \right)$.   
\end{theorem}
%\ariel{Maybe change the name to ``constant-factor hardness''}
\paragraph{Constant-Factor Hardness.} The hardness results above deal with the case where the (in)approximation factor is $(1-\eps)$ for some small $\eps \in (0,1)$. It remains an intriguing question as to whether better approximation algorithms exist for smaller approximation factors. Making progress towards this question, %we show two strong conditional lower bounds. First, 
we show that for any constant approximation ratio, the running time of $n^{\Omega(\sqrt{d})}$ is still required, as stated below. Note that as $\eps$ is a constant here, the upper bound on the running time remains $n^{O(d)}$ \cite{ChandraHW76,FRIEZE1984100,CapraraKPP00} and thus this is not yet tight.

\begin{theorem} \label{thm:sqrt-runningtime-lb}
Assuming \textnormal{Gap-ETH}, for any constant $\rho \in (0,1)$, there exist constants $\delta > 0$ and $d_0 \in \N$ such that the following holds: for any constant $d \geq d_0$, there is no $\rho$-approximation algorithm for $d$-dimensional knapsack that runs in time $O(n^{\delta \sqrt{d}})$.
\end{theorem}

%\ariel{We need to write that there is a $\sqrt{d}$-approximation in case the costs and budgets are $0/1$, and there is a $\sqrt{d}$ lower bound in \cite{ChekuriK2004}. We must clearly state the difference between our lower bound to the one in \cite{ChekuriK2004}}

\paragraph{$\Theta(1/\sqrt{d})$-Factor Hardness\ifapprox~and Approximation\fi.} 
%Next, we consider the other extreme where we want the running time to be {\em fixed parameter tractable (FPT)} with the parameter $d$. That is, the running time is bounded by $f(d) \cdot n^{C}$, where $C$ is independent of $d$ and $f$ is an arbitrary function. %(in other words, $d$ can be considered to be given as part of the input). 
Next, we consider the other extreme where we want the running time to be $O(n^C)$ where $C$ is independent of $d$ (in other words, $d$ can be considered to be given as part of the input).
In the special case where all costs and budgets are boolean (i.e. $B_j = 1$ for all $j \in [d]$)\footnote{In fact, the algorithm works for a more general case; see \Cref{thm:randomized-rounding}.}, there is a polynomial-time $\Omega\left(\frac{1}{\sqrt{d}}\right)$-approximation algorithm \cite{Srinivasan95} and a nearly-tight NP-hardness result of inapproximability factor $\Omega\left(\frac{1}{d^{1/2 - \eps}}\right)$ for any constant $\eps > 0$ \cite{ChekuriK2004}. %in polynomial time in $d$ and in $n$. The upper bound deteriorates as a function of the maximum budget.    
%In this case, 
However, this lower bound only holds when $d = n^{\Omega(1)}$ and does not exclude {\em fixed parameter tractable (FPT)} algorithms\footnote{That is, the running time is bounded by $f(d) \cdot n^{C}$, where $C$ is independent of $d$ and $f$ is an arbitrary function.} with the parameter $d$. In this work, we strengthen this result by showing that $\Omega\left(\frac{1}{\sqrt{d}}\right)$-factor hardness holds even for sufficiently large constants $d$. As mentioned after \Cref{thm:main}, such a lower bound also rules out any FPT (in $d$) algorithm with $\omega\left(\frac{1}{\sqrt{d}}\right)$-approximation ratio.

%$o\left(\frac{1}{\sqrt{d}}\right)$-approximation
%\pasin{I removed the $f(d)$ stuff below. I think it doesn't make sense because $d$ is a constant here. Also, we don't have $f(d)$ in other theorems so it'd be incredibly confusing to have it in just this one theorem.}

\begin{theorem} \label{thm:sqrt-ratio}
Assuming \textnormal{Gap-ETH}, for any constant $C \geq 1$, there exist $\rho > 0$ and $d_0 \in \N$ such that the following holds: for any $d \geq d_0$, there is no $\left(\frac{\rho}{\sqrt{d}}\right)$-approximation algorithm for $d$-dimensional knapsack that runs in $O(n^C)$ time.
\end{theorem}

%$(\rho \sqrt{d})$-approximation

We stress that the lower bound in \Cref{thm:sqrt-ratio} does \emph{not} hold for boolean instances. Indeed, such a case can be easily solved (exactly) in time $O(n \cdot 2^d)$. Thus, our reduction is quite different compared to that of \cite{ChekuriK2004}. In the non-boolean case, the best known upper bound, assuming $d$ is given as part of the input, is an $\Omega\left(\frac{1}{d}\right)$-approximation~\cite{Srinivasan95,CapraraKPP00}.
%
%In this regime, it is tempting to start with the following natural reduction from maximum independent set. Given a graph $G = (V,E)$ construct a $d = |E|$-dimensional knapsack instance with one item for each vertex; the cost of item $v \in V$ is $1$ in dimension $e \in E$ if and only if $e$ is adjacent to $v$, and $0$ otherwise. The budget is $1$ in all dimensions. Clearly, there is a one-to-one correspondence between independent sets and solutions to the knapsack instance. (Un)fortunately, naive enumeration in time $2^{O(d)}$ gives an exact solution to these instances. Hence, this technique cannot be used to give meaningful FPT lower bounds with the parameter $d$. %parametrized $d$ (recall that \Cref{thm:main} allows arbitrarily large coefficient $f(d,\eps)$ depending on $d$ and $\eps$). 
%
%Using different techniques, we show an inapproximability factor of $\Omega(\sqrt{d})$ that allows an arbitrary coefficient $f(d)$ in the running time. We remark that the best known upper bound for this case is an $\Omega\left(\frac{1}{d}\right)$-approximation~\cite{RaghavanT87,CapraraKPP00}\footnote{Note that Srinivasan~\cite{Srinivasan95} gave an improved approximation algorithm for certain regimes of parameters, but overall it does not improve the approximation ratio.}.
%
\ifapprox
To complement this result, we give an $\Omega\left(\frac{1}{\sqrt{d}}\right)$-approximation algorithm that runs in polynomial time as long as $W \leq \exp(O(n^{1/d^2}))$, as stated more precisely below.
\begin{theorem} \label{thm:main-apx}
For every $d \geq 1$, there is a $\left((d \cdot \log W)^{O(d^2)} + n^{O(1)}\right)$-time $\Omega\left(\frac{1}{\sqrt{d}}\right)$-approximation algorithm for $d$-dimensional knapsack.
\end{theorem}
Since we only require $W = O(\log n)$ in the above hardness (\Cref{thm:sqrt-ratio}), this algorithm matches our lower bound for a large regime of parameters. 
\fi

%\ariel{I was not able to understand the paragraph,  and I am not sure what its purpose is. It should either be clarified or removed (after discussion: we agreed to relocate the paragraph to the more technical section and add details there )} 
%Finally, we note that we can also provide an ``approximation-vs-running time tradeoff'' in-between the settings of \Cref{thm:sqrt-runningtime-lb} and \Cref{thm:sqrt-ratio}. However, we choose not to state it since the parameterization becomes cumbersome and, in our opinion, the two cases stated here are the most interesting.

\subsection{Technical Contribution}

Existing lower bounds for $d$-dimensional knapsack can be roughly partitioned into two categories. First, lower bounds for $d = 2$ dimensions  \cite{KORTE1981415,MagazineC84,KulikS10,JansenLL16}. The basic technique in this setting is to reduce Subset Sum to $2$-dimensional knapsack. In broad terms, each number in the Subset Sum instance corresponds to an item. One dimension imposes an upper bound by the target value on the original sum, while the second dimension guarantees that the original sum is higher than the target value using clever cost construction. However, these constructions are only useful for $d = 2$ using an extremely small $\eps$, and cannot be used to reveal the asymptotic hardness of $d$-dimensional knapsack for larger values of $d$. The second category of lower bounds considers the setting where $d$ is arbitrarily large \cite{ChekuriK2004,Srinivasan95}; as mentioned previously, the existing lower bounds break down if one allows coefficients that exponentially depend on $d$. Therefore, to reach our desired results, we need different techniques.

In contrast, we resort to a reduction from {\em $2$-constrained satisfaction problem} with {\em rectangular constraints (R-CSP)} (see the definition in \Cref{sec:preliminaries}).  The problem is a variant of the  {\em 2-constraint satisfaction problem (2-CSP)} in which the input is a graph $G$, alphabets $\Sigma$ and $\Upsilon$, and mappings $\pi_{e,u}, \pi_{e,v}: \Sigma\rightarrow \Upsilon$ for every $e=(u,v)\in E$. Informally, a solution is a {\em partial} assignment $\varphi$ which assigns a letter in $\Sigma$ to {\em some} of the vertices in $G$. The assignment must satisfy the property that $\pi_{e,u} (\varphi(u)) = \pi_{e,v} (\varphi(v))$ for every $e=(u,v)\in E(G)$ such that $\varphi$ assigns values for both $u$ and $v$. The goal is to find a partial assignment which assigns values to a maximum number of vertices. 

Our starting point is a reduction from R-CSP to $d$-dimensional knapsack using small (polynomial) weights. Roughly speaking, each constraint (edge) of the R-CSP instance is expressed by a pair of dimensions, obtaining an important property that every dimension can contain up to $2$ items with non-zero cost in the dimension, yielding a reduction that preserves the approximation guarantee. %enabling to provide constant approximation guarantee. 
This construction is sufficient to obtain the lower bound presented in \Cref{thm:VKexact} in conjunction with existing lower bounds for Max 2-CSP \cite{Marx10,karthik2023conditional}. 
However, this construction is not versatile enough to give a strong lower bound where both $d$ and $\eps$ play a significant role. For instance, it cannot attain a $n^{O\left(d+\frac{1}{\eps}\right)}$ running time lower bound for a PTAS.

Our main technical contribution is a robust reduction from R-CSP to $d$-dimensional knapsack in which $d$ is (essentially) given as a parameter to the reduction. In an intuitive level, this reduction can be seen as a {\em dimension embedding} of an instance with $d$ many dimensions to an instance with $\tilde{d} \ll d$ dimensions. 
Shrinking the number of dimensions increases almost at the same rate the {\em depth} of each dimension - the maximum number possible of items with non-zero cost in the dimension that can be taken to a solution.  As the depth increases, it requires having an increasingly smaller error parameter $\eps$ to preserve the approximation guarantee of the reduction.  This trade-off between $d$ and $\eps$ is the key insight which eventually gives the very fine bound of \Cref{thm:main}. 

We remark that as a component of the reduction, we write a strengthened version of the Max 2-CSP lower bound of \cite{Marx10,karthik2023conditional} yielding inapproximability up to a certain factor based on Gap-ETH. (See \Cref{thm:CSP}.)
 Although our proof is completely based on the constructions of \cite{Marx10,karthik2023conditional}, the inapproximability result itself may be of an independent interest, since the work of \cite{Marx10} has been a classic technique in proving (exact) parametrized lower bounds (e.g., \cite{marx2022optimal,jansen2013bin,curticapean2017homomorphisms,curticapean2015parameterizing,bonnet2020parameterized,bringmann2015hitting,pilipczuk2018directed,chitnis2021parameterized,lokshtanov2020parameterized,bonnet2017graph}).  
 
\ifapprox
Finally, our approximation algorithm \Cref{thm:main-apx} is based on a simple discretization of the cost together with a bruteforce algorithm. As usual, the cost is discretized geometrically based on some scaling factor. The main difference compared to previous work is that we also discretize the \emph{complement} of the costs (i.e. $B_j - c(i)_j$); our main structural lemma shows that such a discretization preserves the cost of the optimal solution up to $\Omega\left(\frac{1}{\sqrt{d}}\right)$-factor.
\fi

 %previously written only fo exact algorithms. 

\paragraph{Organization:}
We start with some preliminary definitions in \Cref{sec:preliminaries}. In \Cref{sec:simple} we present our simple reduction from R-CSP to $d$-dimensional knapsack and use it to prove \Cref{thm:VKexact}. Then, in \Cref{sec:RtoVK} we give our second and more elaborate reduction from R-CSP to VK. In \Cref{sec:mainResults} we prove our remaining main lower bounds (\Cref{thm:main}, \Cref{thm:sqrt-runningtime-lb}, and \Cref{thm:sqrt-ratio}). \ifapprox\Cref{sec:algorithm} describes the approximation algorithm presented in \Cref{thm:main-apx}. \fi
We conclude in a discussion in \Cref{sec:discussion}. We defer the proofs on 2-CSP to \Cref{sec:SATtoRCSP}.

\section{Preliminaries}
\label{sec:preliminaries}

%\ariel{I redefined the size of a partial assignment for R-CSP as well as $\MaxPar$}
%\ariel{ it seems we want to keep the preliminaries. We should probably define rectangular CSP and not 2CSP, and also state the hardness result for rectangular CSP here (and refer the reader to the proof in the appendix)}

\paragraph{Basic definitions} For an instance $\cI$ of a maximization problem $\Pi$, let $\OPT(\cI)$ be the optimum value of $\cI$ and let $|\cI|$ denote the encoding size of $\cI$. For a set $S$, a function $f:S \rightarrow \mathbb{N}$, and a subset $X \subseteq S$ we use the abbreviation $f(X) = \sum_{x \in X} f(x)$. For some $\alpha \in (0,1)$, an $\alpha$-{\em approximate} solution for $\cI$ is a solution for $\cI$ of value at least $\alpha \cdot \OPT(\cI)$. %We give a hardness result for \dkp~parametrized by the number of the dimensions, i.e., $d$.   We formally define the problem below. 
For all $t \in \mathbb{N}$ let $[t] = \{1,2,\ldots,t\}$. 

\paragraph{$d$-dimensional knapsack}
For short, for any $d \geq 1$ we use $d$-VK to denote an instance of $d$-dimensional (vector) knapsack (see the formal definition of $d$-VK in \Cref{def:dKP}). We also use VK to denote the collection of $d$-VK instances over all $d \geq 1$. For some $d \in \N$ and a $d$-VK instance $\cI = (I,p,c,B)$, let $W(\cI) = \max_{j \in [d]} B_j$ be the maximum budget of the instance. For any $d \in \N$ and $\tilde{d} \in [d]$, we may assume w.l.o.g. that every $\tilde{d}$-dimensional knapsack instance is also a $d$-dimensional knapsack instance by adding extra $d-\tilde{d}$ dimensions with zero budgets and zero costs for all items.  

\paragraph{Graph Theory definitions} For every graph $G$ let $V(G)$ and $E(G)$ denote the set of vertices and edges of $G$, respectively. %Throughout the paper directed edges are denoted by $(u,v)$ and undirected edges by $\{u,v\}$. 
We assume that all graphs $G$ considered in this paper are simple directed graphs with no antisymmetric edges, such that %there is a total order $<$ over the vertex set and 
edges $(u,v) \in E(G)$ are always oriented according to $u \prec v$, where $\prec$ denotes that $u$ is smaller than $v$ according to the lexicographic order. This assumption is purely technical as our reduction requires a fixed order of the endpoints of the edges. For a vertex $v \in V(G)$ we use $\Adj_G(v)$ to denote the set of adjacent edges to $v$ in $G$. For some $r \in \N$, a graph $G$ is $r$-{\em regular} if the degree of each vertex is exactly $r$.

%in an undirected graph. %to order vertices within a constraints. 
%We give below the formal definitions of the problems considered in this paper. 

\paragraph{2-CSP with rectangular constraints} We prove the hardness of VK via a reduction from 2-CSP with {\em rectangular constraints (R-CSP)}  defined formally below. 

\begin{mdframed} \vspace{-0.1in}
	\begin{definition} \leavevmode\vspace{-\baselineskip}
		{\bf 2-CSP with rectangular constraints (R-CSP)} 
		
		\begin{itemize}
			\item {\bf Input:} $\Pi = \left(G,\Sigma,\Upsilon, \{\pi_{e, u}, \pi_{e, v}\}_{e = (u,v)\in E(G)}\right)$
			consisting of
			\begin{itemize}
				\item A constraint graph $G$.
				\item Alphabet sets $\Sigma, \Upsilon$, where $\Upsilon = \{1,\ldots,m\}$ for some $m \in \mathbb{N}_{>0}$.  
				\item For each edge $e = (u, v) \in E(G)$, two functions $\pi_{e, u}: \Sigma \to \Upsilon$ and $\pi_{e, v}: \Sigma \to \Upsilon$.
				%\item A perpendicular symbol $\perp$ such that $\perp \notin \Sigma$.  
			\end{itemize}
			
			\item {\bf Satisfied edges:} For $e = (u, v) \in E(G)$, we say that that $e$ is {\em satisfied} by $(\sigma_u, \sigma_v) \in \Sigma \times \Sigma$ if and only if $\pi_{e, u}(\sigma_u) = \pi_{e, v}(\sigma_v)$. 
			
			\item {\bf Partial Assignment:} A \emph{partial assignment} is a function $\varphi: V(G) \to \Sigma \cup \{\perp\}$. 
			
				\item {\bf Consistency:} 
			A partial assignment $\varphi$ is \emph{consistent} if for all $e = (u, v) \in E(G)$ such that $\varphi(u) \neq \perp$ and  $\varphi(v) \ne \perp$ it holds that $e$ is satisfied by $(\varphi(u), \varphi(v))$. 
			
				\item {\bf Size:} The \emph{size} of the partial assignment $\varphi$ is defined as $\abs{\varphi} := |\{v \in V(G) \mid \varphi(v) \ne \perp\}|$.
				
			\item {\bf Objective:} Find a consistent partial assignment of maximum size. We define  $\MaxPar(\Pi)$ to be the maximum size of a consistent partial assignment for $\Pi$. 
			%	\item  {\bf Parameter:} $k = |E(G)|$.  
		\end{itemize}
	\end{definition}
\end{mdframed}

Using techniques from \cite{Marx10,karthik2023conditional}, we can show the following hardness result for R-CSP. We note that the objective in R-CSP as defined here is different than the usual ``Max 2-CSP'' objective (i.e. maximizing the number of edges satisfied). A similar result can also be shown for this objective. We defer the full details to \Cref{sec:SATtoRCSP}.

\begin{theorem}
	\label{thm:RCSP}
	Assuming \textnormal{Gap-ETH},  there exist constants $\alpha,\beta \in (0,1)$ and $k_0 \in \N$, such that, for any constant $k \geq k_0$, there is no algorithm that given an \textnormal{R-CSP} instance $\Pi$ with a 3-regular constraint graph $H$ such that $|V(H)| \leq k$, runs in time $O\left(|\Pi|^{\alpha \cdot \frac{k}{\log(k)}}\right)$ and distinguish between:
	\begin{itemize}
		\item {\bf (Completeness)} $\MaxPar(\Pi) = |V(H)|$.
		\item {\bf (Soundness)} $\MaxPar(\Pi) < \left(1 - \frac{\beta}{\log(k)}\right) \cdot |V(H)|$.
	\end{itemize}
\end{theorem}  

We also use the following result, which is analogous to \Cref{thm:RCSP} but uses the weaker assumption of ETH rather than Gap-ETH. 
%thm:RCSP
The proofs is analogous to \Cref{thm:RCSP}, which follows from the work of \cite{karthik2023conditional} combined with our reduction from Gap-ETH to R-CSP (see \Cref{sec:SATtoRCSP}). 

%\ilan{Pasin: Can we really say that the following result indeed follows from \cite{karthik2023conditional} combined with our reduction? or do we need something else?}

\begin{theorem}
	\label{thm:RCSPeth}
	Assuming \textnormal{ETH},  there exist constants $\alpha \in (0,1)$ and $k_0 \in \N$, such that, for any constant $k \geq k_0$, there is no algorithm that takes in an \textnormal{R-CSP} instance $\Pi$ with a 3-regular constraint graph $H$ such that $|V(H)| \leq k$ variables, runs in time $O\left(|\Pi|^{\alpha \cdot \frac{k}{\log(k)}}\right)$ and decides if $\MaxPar(\Pi) = |V(H)|$. %distinguish between:
	%\begin{itemize}
	%\item {\bf (Completeness)} $\MaxPar(\Pi) = |V(H)|$.
	%\item {\bf (Soundness)} $\MaxPar(\Pi) < \left(1 - \frac{\beta}{\log(k)}\right) \cdot |V(H)|$.
	%\end{itemize}
\end{theorem}

\subsubsection*{3-SAT, ETH, and Gap-ETH:} 
 
Our lower bounds are conditioned on Gap-ETH and ETH. Before we state these, recall the standard 3-SAT problem. % using convenient notations for the technical sections.  %Given a set of  variables $V$, we use  $\overline{V} = \{\overline{v}~|~v \in V\}$ to denote the set of {\em negations} of $V$, considered to be disjoint from $V$ (we also use $v = \overline{\overline{v}}$ for all $v \in V$). 
%Given a set of variables $V$, a {\em clause} over $V$ is a %logical formula of the form 
%tuple $c = (v_1,v_2,v_3,t_1,t_2,t_3)$ where $v_1,v_2,v_3 \in V$ and $t_1,t_2,t_3 \in \{0,1\}$ are the target values of $v_1,v_2,v_3$, respectively. For all $i \in [3]$ we use the notation $c_i = v_i$ and $\tar(c_i) = t_i$. %x \lor y \lor z$ where there are variables $v_x,v_y,v_z \in V$ such that for all $a \in \{x,y,z\}$ it holds that $a = v_a$ or $a = \overline{v}_a$. We use the notation $c[1] = x, c[2] = y$, and $c[3] = z$, as well as $\overline{c}[1] = \overline{x}, \overline{c}[2] = \overline{y}$, and $\overline{c}[3] = \overline{z}$. 
%We also say that $v_1, v_2$, and $v_3$ {\em appear} in $c$.  

%(we use $v = \overline{\overline{v}}$ for all $v \in V$). 

%\ariel{This does not define 3-SAT: the main issue is that the term ``clause'' is undefined, as well as ``satisfy''. Also, there is no restriction for the clauses to be with $3$ variables. Resolve this in the meeting with Ilan. What happens to the previous definition we used in the proofs for the csp problems?}
\begin{definition}
	\label{def:SAT}
	{\bf {\bf 3-SAT$(D)$}} 
		\begin{itemize}
		\item {\bf Input:} $\phi = (V,C)$, where $V$ is a set of variables and $C$ is a set of disjunction clauses of three variables or negation of variables in $V$, such that each variable appears in at most $D$ clauses.
		\item {\bf Assignment:} A function $s:V \rightarrow \{0,1\}$. %where for all $v \in V$ it holds that $s(v) = 1-s \left(\overline{v}\right)$. 
		
		%\item {\bf Satisfied Clause:} A clause $c \in C$ is {\em satisfied} by an assignment $s$ if there is $i \in [3]$ satisfying $s(c_i) = \tar(c_i)$. %(i) $c[i] \in V$ and $s \left(c[i]\right) = 1$, or (ii) $c[i] \in \overline{V}$ and $s \left(c[i]\right) = 0$.%one of the following holds.
%		\begin{itemize}
%			\item There is $i \in [3]$ such that $c[i] \in V$ and $s \left(c[i]\right) = 1$.
%				\item There is $i \in [3]$ such that $c[i] \in \overline{V}$ and $s \left(c[i]\right) = 0$.
%		\end{itemize}
		%$s \left(c[1]\right)+s \left(c[2]\right)+s \left(c[3]\right) \geq 1$.
		\item {\bf Objective:} Find an assignment satisfying a maximum number of clauses. Let $\SAT(\phi)$ be the maximum number of clauses in $C$ satisfied by  a single assignment to $\phi$. 
	%	\item  {\bf Parameter:} $k = |E(H)|$.  
	\end{itemize}
	
	\comment{
	
	A \textnormal{3-SAT} formula is a pair $\phi = (V,C)$ where $V$ is a set of variables and $C$ is a set of clauses over $V$. %contains triplets from $V$ or their negation. %That is, let $\overline{V} = \{\overline{v}~|~v \in V\}$ be the set of negations of variables, disjoint from $V$ (we use $v = \overline{\overline{v}}$ for all $v \in V$). 
	%For every $c \in C$ it holds that %$c = x \lor y \lor z$ where there are variables $v_x,v_y,v_z \in V$ such that for all $a \in \{x,y,z\}$ it holds that $a = v_a$ or $a = \overline{v}_a$; we denote $c[1] = x, c[2] = y$, and $c[3] = z$, and in addition $\overline{c}[1] = \overline{x}, \overline{c}[2] = \overline{y}$, and $\overline{c}[3] = \overline{z}$. 
	An {\em assignment} for $\phi$ is a function $s:V \rightarrow \{0,1\}$. A clause $c \in C$ is {\em satisfied} by an assignment $s$ if %(i) there is $v \in V$ such that 
	for all $c \in C$ it holds that $s \left(c[1]\right)+s \left(c[2]\right)+s \left(c[3]\right) \geq 1$. %$s(v) = 1$ and $$, or (ii)   there is $v \in V$ such that $s(v) = 0$ and $\overline{v} \in c$. 
	Let $\SAT(\phi)$ be the maximum number of clauses in $C$ satisfied by  a single assignment to $\phi$. 
}
\end{definition}

We first state ETH \cite{impagliazzo2001complexity}.

\begin{con}
	\label{ETH}
	{\bf  Exponential-Time Hypothesis (ETH) } 
	There is a constant $\xi > 0$ such that there is no algorithm
	that given a \textnormal{3-SAT} formula $\phi $ on $n$ variables and $m$ clauses %, where $m \leq \gamma \cdot n$, distinguishes
	%between 
	can decide if $\SAT(\phi) = m$ %and $\SAT(\phi)< 0.9 \cdot m$ 
	in time $O \left(2^{\xi \cdot n} \right)$. 
\end{con}

%We use the following result of \cite{karthik2023conditional}. 

%\begin{theorem}
%	\label{thm:marx}
%	Assuming \textnormal{ETH}, there is a constant $\nu \in (0,1)$ such that there is no algorithm that given a \textnormal{2-CSP} instance $\Gamma = (H,\Sigma,X)$, where $H$ is $3$-regular and $k = \abs{E(H)}$, and decides if $\CSP(\Gamma) = k$ in time $O \left(f(k) \cdot \abs{\Gamma}^{ \frac{\nu \cdot k}{\log(k)}}\right)$, where $f$ is any computable function. 
%\end{theorem}

The Gap Exponential Time Hypothesis (Gap-ETH)~\cite{dinur2016mildly,manurangsi2016birthday} is stated below. Note that the bounded degree assumption is w.l.o.g. (see e.g. \cite[Footnote 5]{manurangsi2016birthday}).

\begin{con}[{\bf Gap Exponential-Time Hypothesis (Gap-ETH)}]
\label{GapETH}
There exist constants $D \in \N$, and $\delta, \eps \in (0,1)$ such that there is no algorithm that is given a \textnormal{3-SAT}$(D)$ instance $\phi $ on $n$ variables and $m$ clauses distinguishes between $\SAT(\phi) = m$ and $\SAT(\phi)< (1 - \eps) \cdot m$ in time $O \left(2^{\delta \cdot n} \right)$.  
\end{con}

\section{A Simple Reduction from R-CSP to $d$-Dimensional Knapsack}
\label{sec:simple}

%The reduction is stated and proved below.
In this section, we give our first reduction from R-CSP to $d$-VK and prove \Cref{thm:VKexact}. Then, in the next section we give a more involved reduction between the two mentioned problems in order to prove our second main result (\Cref{thm:main}). %We give the easier reduction as part of the exposition and since it is lossless - completely preserves the approximation hardness of~R-CSP.  

Given an R-CSP instance
$\Pi = \left(G,\Sigma,\Upsilon, \{\pi_{e, u}, \pi_{e, v}\}_{e = (u,v)\in E(G)},\perp \right)$, the reduction creates a $d$-VK instance $\cI(\Pi) = (I,p,w,B)$ where $d = |V(G)| + 2|E(G)|$; namely, there is a dimension for every vertex and every endpoint of an edge. The items will be pairs $(v,\sigma) \in V(G) \times \Sigma$ of a vertex and a symbol from the alphabet. The vertex-dimensions guarantee that we can take at most one {\em copy} of each vertex - corresponding to an assignment of at most one symbol to the vertex.  The edge-dimensions guarantee consistency. For each edge $(u,v) \in E(G)$, we have two dimensions $(e,u)$ and $(e,v)$. We define the costs of items that are copies of $u,v$ accordingly so that any two items $(u,\sigma_u)$ and $(v,\sigma_v)$ satisfy together the budget constraints in both $(e,u)$ and $(e,v)$ if and only if $\pi_{e,u}(\sigma_u) = \pi_{e,v}(\sigma_v)$. The reduction is stated and proven as follows.  
%The cost of item $(u,\sigma)$ in dimension $(e,u)$

\begin{lemma} 
	\label{thm:red-large-gap}
	There is a reduction \textnormal{\textsf{simple}},  %\textnormal{\textsf{R-CSP $\rightarrow $ VK}} 
	which given an \textnormal{R-CSP} instance 
	\\$\Pi = \left(G,\Sigma,\Upsilon, \{\pi_{e, u}, \pi_{e, v}\}_{e = (u,v)\in E(G)}\right)$, returns in time $O(|\Pi|^3)$ an instance $\cI(\Pi) = (I,p,w,B)$ of $d$-\textnormal{VK} such that the following holds.
	\begin{enumerate}
		\item There is a solution for $\cI(\Pi)$ of profit $q$ iff there is a consistent partial assignment for $\Pi$ of size~$q$.  
		
		\item $d = |V(G)| + 2 \cdot |E(G)|$. % \leq |V(G)|^2$,
		
		\item $\abs{\cI(\Pi)} \leq O(|\Pi|^3)$. 
		
		\item $W(\cI(\Pi)) \leq \abs{\Pi}$. 
	\end{enumerate}
\end{lemma}

\begin{proof}
	%Let $\Pi = \left(G,\Sigma,\Upsilon, \{\pi_{e, u}, \pi_{e, v}\}_{e = (u,v)\in E(G)},\perp \right)$ be an R-CSP instance. 
	We define the instance $\cI(\Pi) = (I,p,w,B)$ as follows:
	\begin{itemize}
		\item Let $I = V(G) \times \Sigma$.
		\item  Define $p(i) = 1$ for all $i \in I$. %and let the profit be 1 for every item.
		\item Let $d = |V(G)| + 2 \cdot |E(G)|$. For notational convenience, we assume that each dimension is associated with an element in $V(G) \cup \{(e, v) \mid v \in V(G), e \in \Adj_G(v)\}$.
		\item Let $m = \abs{\Upsilon}$ and let the budget $B_j$ be $m$ for all dimensions $j$. 
		\item Finally, let the cost be as follows:
		\begin{itemize}
			\item For all $v \in V(G)$ and $\sigma \in \Sigma$, let $w_v(v, \sigma) = m$
			\item For all $e = (u, v) \in E(G )$ and $\sigma \in \Sigma$, let
			\begin{align*}
				w_{(e, u)}(v, \sigma) &= m-\pi_{(e, v)}(\sigma), \\
				w_{(e, v)}(v, \sigma) &=  \pi_{(e, v)}(\sigma), \\
				w_{(e, u)}(u, \sigma) &=  \pi_{(e, u)}(\sigma), \\
				w_{(e, v)}(u, \sigma) &= m-\pi_{(e, u)}(\sigma).
			\end{align*}
			\item Cost of all items in all other coordinates are set to zero.
		\end{itemize}
	\end{itemize}
	
	It is clear that the reduction runs in $O(|\Pi|^3)$ time and the encoding size is $\abs{\cI(\Pi)} \leq O(|\Pi|^3)$.  Also, by definition it holds that $W(\cI(\Pi)) = m \leq \abs{\Pi}$. %Consider the following example. %of the construction.
	We next prove the completeness and soundness. %An example of the construction is given at the end of this section. 
	
	\paragraph{(Completeness)} Let $\varphi: V(G) \to \Sigma \cup \{\perp\}$ be a consistent partial assignment for $\Pi$. Define a solution $S = \{(v, \varphi(v)) \mid \varphi(v) \ne \perp\}$. Clearly, $p(S) = |S| = |\varphi|$. To prove the feasibility of $S$, note that for all $v \in V(G)$ it holds that $w_v(S) \leq m = B_v$. For every $e = (u,v) \in E(G)$ one of the following cases holds.
	
	\begin{itemize}
		\item  $\varphi(u) = \perp$ or $\varphi(v) = \perp$. Then, it can be easily verified that $w_{(e,u)}(S) \leq m = B_{(e,u)}$ and $w_{(e,v)}(S) \leq m = B_{(e,v)}$. 
		
		\item 	$\varphi(u) \neq \perp$ and $\varphi(v) \neq \perp$. then
		$$w_{(e,u)}(S) = \pi_{(e,u)}(\varphi(u))+m-\pi_{(e,v)}(\varphi(v)) = m = B_{(e,u)},$$ where the second equality follows from the consistency of $\varphi$. Similarly, it holds that $$w_{(e,v)}(S) = \pi_{(e,v)}(\varphi(v))+m-\pi_{(e,u)}(\varphi(u)) = m = B_{(e,v)}.$$
	\end{itemize}
	%Otherwise, if 
 Thus, $S$ is a feasible solution for $\cI(\Pi)$.  %of profit $\abs{\varphi}$. 

	%\ariel{The completeness proof is technically incorrect: in the above $\varphi(v)$ may be $\perp$ and the expression becomes undefined. I am not sure whether we want to resolve this issue. }	
	
	%It is also simple to verify that $S$ is a valid solution of $\cI(\Pi)$.
	
	\paragraph{(Soundness)} Let $S \subseteq I$ be a solution of $\cI(\Pi)$. For every $v \in V(G)$ let $\Sigma_{v, S} := \{\sigma \in \Sigma \mid (v, \sigma) \in S\}$. Since $S$ is a feasible solution, for every $v \in V(G)$ it holds that $w_v(S) = \abs{\Sigma_{v, S}} \cdot m \leq m$. Thus, $\abs{\Sigma_{v, S}} \leq 1$. 
	Construct an assignment $\varphi: V(G) \to \Sigma \cup \{\perp\}$ such that for all $v \in V(G)$ set $\varphi(v) = \perp$ if  $\Sigma_{v, S} = \emptyset$, and otherwise set $\varphi(v)$ as the unique element in $\Sigma_{v, S}$. %be the only element in $\Sigma_{v, S}$
%	\begin{itemize}
%		\item For each $v \in V(G)$:
%		\begin{itemize}
%			\item Consider the set $\Sigma_{v, S} := \{\sigma \in \Sigma \mid (v, \sigma) \in S\}$. 
%			\item Observe that the coordinate $v$ implies that $|\Sigma_{v, S}| \leq 1$.
%			\item If $\Sigma_{v, S} = \emptyset$, let $\varphi(v) = \perp$. Otherwise, let $\varphi(v)$ be the only element in $\Sigma_{v, S}$.
%		\end{itemize}
%	\end{itemize}
	Since $\abs{\Sigma_{v,S}} \leq 1$ for all $v \in V(G)$, it follows that $|\varphi| = |S| = p(S)$. To see that $\varphi$ is a consistent partial assignment, suppose for the sake of contradiction that %this is not the case, i.e. 
	there exists $e = (u, v) \in E(G)$ such that $\varphi(u), \varphi(v) \ne \perp$ and $\pi_{(e, u)}(\varphi(u)) \ne \pi_{(e, v)}(\varphi(v))$. Suppose w.l.o.g. that $\pi_{(e, u)}(\varphi(u)) > \pi_{(e, v)}(\varphi(v))$. Since $S$ contains both $(v, \varphi(v))$ and $(u, \varphi(u))$ we have
	\begin{align*}
		w_{(e,u)}(S) \geq w_{(e, u)}(u, \varphi(u)) + w_{(e, u)}(v, \varphi(v)) = \pi_{(e, u)}(\varphi(u)) + m-\pi_{(e, v)}(\varphi(v)) > m = B_{(e,u)}. 
	\end{align*}
	 It follows that $S$ is not a valid solution to $\cI(\Pi)$, a contradiction.
\end{proof}

%We conclude with a simple example of the construction.
%\paragraph{Example} Assume that $V(G) = \{u,v\}$ and $E(G) = \{(u,v)\}$. Moreover, let $\Sigma = \{a,b\}$ and $\Upsilon = \{1,2\}$, $\pi_{(e, v)}(a) = \pi_{(e, u)}(a)= 1$, and $\pi_{(e, v)}(b) = \pi_{(e, u)}(b)= 2$. Thus, in words, an assignment $\varphi$ to the instance $\Pi$ can be consistent only by setting $\varphi(u) = \varphi(v)$. We show how the reduced VK instance $\cI(\Pi)$ adheres to this property. the instance $\cI(\Pi)$ has four dimensions $u,v,(e,u)$, $(e,v)$, and four items $(u,a), (u,b), (v,a), (v,b)$. The dimension $u$ (resp. $v$) restricts a solution to take at most one item from $(u,a), (u,b)$ (resp. $(v,a), (v,b)$). Additionally, if a solution takes two items $(u,x), (v,y)$, the dimension $(e,u)$ guarantees that $w_{(e,u)}((u,x))+w_{(e,u)}((v,y)) = \pi_{(e, u)}((u,x))+m-\pi_{(e, u)}((v,y)) = x+m-y \leq m$. Conversely, using the dimension $(e,u)$ we have $m-\pi_{(e, v)}((u,x))+\pi_{(e, u)}((v,y)) = m-x+y \leq m$. Overall, it implies that $x = y$.  

The above reduction, together with \Cref{thm:RCSPeth}, suffices to prove \Cref{thm:VKexact}. 

%We now prove  \Cref{thm:VKexact} using the above reduction. %We remark that it suffices to use the simpler reduction from \Cref{thm:red-large-gap}, however, we use the more general reduction to demonstrate its versatility. 

\comment{
	\begin{theorem}
		\label{thm:VKexact}
		Assuming \textnormal{ETH},  there exist a constant  $q \in (0,1)$, such that 
		there is no algorithm that 
		optimally solves \textnormal{VK} in time $O \left(f(d) \cdot \big(\left|\cI\right|+W(\cI) \big)^{q \cdot \frac{d}{ \log(d)}} \right)$%, 
		where $d$ is the number of dimensions in the instance $\cI$, and 
		where $f$ is any computable function.   
	\end{theorem} 
	Assuming \textnormal{ETH}, there is a constant $\nu \in (0,1)$ such that there is no algorithm that given a \textnormal{2-CSP} instance $\Gamma = (H,\Sigma,X)$, where $H$ is $3$-regular and $k = \abs{E(H)}$, and decides if $\CSP(\Gamma) = k$ in time $O \left(f(k) \cdot \abs{\Gamma}^{ \frac{\nu \cdot k}{\log(k)}}\right)$, where $f$ is any computable function. 
}

\subsubsection*{Proof of \Cref{thm:VKexact}}
Assume that \textnormal{ETH} holds. 
Let $\alpha \in (0,1)$ and $k_0 \in \N$ be the promised constants by \Cref{thm:RCSPeth}. 
%Define a constant $b = \frac{\chi}{6}$. 
%Let $C\geq1$ be a constant such that for any R-CSP instance $\Pi$ it holds that the reduction \textnormal{\textsf{simple}} (described in \Cref{thm:red-large-gap}) runs in time $O \left(\abs{\Pi}^C\right)$. 
%Let $\nu \in (0,1)$ be the promised constant by \Cref{thm:marx}. Let $C\geq1$ be a constant such that for any 2-CSP instance $\Gamma'$ it holds that the reduction \textnormal{2-CSP $\rightarrow$ VK} (described in \Cref{thm:redMain}) runs in time $O \left(\abs{\Gamma'}^C\right)$; there is such a constant by \Cref{thm:redMain}. 
Define a parameter $\zeta =   \frac{\alpha}{1000}$. 
%Let $d_0 = $. 
Let $d_0 \in \N$ such that the following holds.
\begin{enumerate}
	\item $\alpha \cdot \frac{\floor{\frac{d_0}{10}}}{\log \left(\floor{\frac{d_0}{10}}\right)} \geq 3$.
	
	\item $d_0 \geq \max \{10 \cdot k_0,40\}$.

	%	\item $\frac{d_0}{15} \in \N$. 
	%	\item $\floor{\frac{d_0 \cdot \beta \cdot \zeta}{10 \cdot \log(d_0)}} \geq \max \{k_0,C,6\}$.
	
	%\item $d_0 \geq \max \{k_0,6\}$.
	
	%\item $ \frac{\log \left(\log(d_0)\right)}{\log(d_0)} \leq \chi$.
	%\frac{\beta}{40}$.  
\end{enumerate}
%Clearly, there is such $d_0$ since $\zeta,\chi,\alpha,\beta,C$, and $k_0$ are constants, $\lim_{x \rightarrow \infty} \frac{x}{\log x} = \infty$, and $\lim_{x \rightarrow \infty} \frac{\log \log x}{\log x} = 0$. 

\comment{

	Assuming \textnormal{ETH},  there exist constants  $\chi,d_0 >0$, such that  for every integer $d>d_0$ there is  no algorithm that 
	solves $d$-dimensional knapsack exactly in   time $O \left( \big(n+W \big)^{\zeta \cdot \frac{d}{ \log(d)}} \right)$.   
	
}

%In addition, let $\eps = \frac{\chi}{}$. 
Assume towards a contradiction that  there is $d > d_0$ and an algorithm $\cA$ that given a $d$-dimensional knapsack instance returns an optimal solution in time $O \left( \big(n+W \big)^{\zeta \cdot \frac{d}{ \log(d)}} \right)$, where $n$ and $W$ are the encoding size and maximum number in the instance, respectively. %, where $d$ is the number of dimensions in the instance $\cI$, and 
Let $k = \floor{\frac{d}{10}}$. 
We define the following algorithm $\cB$ that decides if an R-CSP instance $\Pi$ on $3$-regular constraint graph $H$ where $\abs{V(H)} \leq k$ satisfies $\MaxPar(\Pi) = \abs{V(H)}$ or $\MaxPar(\Pi)<\abs{V(H)}$. %Let $\Gamma = (H,\Sigma,X)$ be a 2-CSP instance with $k =  \abs{E(H)}$ constraints such that $H$ is $3$-regular. Define $\cB$ on input $\Gamma$ as follows. 
Let %$\Gamma = (H,\Sigma,X)$ 
$$\Pi = \left(H,\Sigma,\Upsilon, \{\pi_{e, u}, \pi_{e, v}\}_{e = (u,v)\in E(G)} \right)$$
be an R-CSP instance with $\abs{V(H)} \leq k$ vertices such that $H$ is $3$-regular. Define $\cB$ on input $\Pi$ by:  

\begin{enumerate}
	%\item Compute using $\textnormal{2-CSP}  \rightarrow  \textnormal{R-CSP}$ the R-CSP instance $$\Pi(\Gamma) = \left(G,\Sigma_{\Pi},\mathcal{X}, \{\pi_{e, u}, \pi_{e, v}\}_{e = (u,v)\in E(G)}, \right).$$
	%\item Compute the VK instance $\cI(\Pi \left(\Gamma\right)) = (I,p,c,B)$ using $\textnormal{R-CSP}  \rightarrow  \textnormal{VK}$.  
	
	\item Compute the VK instance $\cI(\Pi) = (I,p,w,B)$ by the reduction \textnormal{\textsf{simple}} described in \Cref{thm:red-large-gap}. %with parameter $F = 1$.  
	
	\item Execute $\cA$ on $\cI \left(\Pi\right)$. Let $S$ be the returned solution. 
	
	\item If $p(S) = \abs{V(H)}$: return that  $\MaxPar(\Pi) = \abs{V(H)}$.
	
	\item If $p(S) <  \abs{V(H)}$: return that $\MaxPar(\Pi) < \abs{V(H)}$. 
	%	\item Compute $\rightarrow $ R-CSP}} that, given a \textnormal{2-CSP} $$\Pi(\Gamma) = \left(G,\Sigma_{\Pi},\mathcal{X}, \{\pi_{e, u}, \pi_{e, v}\}_{e = (u,v)\in E(G)},\perp \right)$$
\end{enumerate}   

First, by  \Cref{thm:red-large-gap} observe that the number of dimensions of  $\cI \left(\Pi\right)$ is 
\begin{equation}
	\label{eq:XfeasR}
 |V(H)| + 2 \cdot |E(H)| \leq k+ 3 \cdot 2 \cdot k  \leq 7 \cdot k \leq d. 
\end{equation}
The first inequality holds since $\abs{V(H)} \leq k$ and since $H$ is $3$-regular.
\comment{
\begin{equation}
\label{eq:XfeasR}
3 \cdot \ceil{ \frac{\abs{V(H)}+\abs{E(H)}}{F}} = 3 \cdot \abs{V(H)}+3 \cdot \abs{E(H)} \leq 3 \cdot k+ 3 \cdot 2 \cdot k  = 9 \cdot k \leq d. 
\end{equation}
The inequality holds since $\abs{V(H)} \leq k$ and since $H$ is $3$-regular. The last inequality follows from the selection of $k$. By \eqref{eq:XfeasR} it holds that the number of dimensions of $\cI \left(\Pi\right)$ is at most $d$. 
}
%For simplicity, let $n = \abs{\cI(\Pi)}$ be the encoding size and 
Let $W = W \left(\cI(\Pi)\right)$ be the maximum weight of the instance. By  \Cref{thm:red-large-gap}  it holds that $W \leq \abs{\Pi}$. 
\comment{
 By \Cref{thm:red-large-gap} it holds that 
\begin{equation}
\label{eq:Wval}
W \leq \left(3 \cdot F \cdot \abs{\Pi}\right)^{6 \cdot F} = \left(3 \cdot \abs{\Pi}\right)^{6} 
%O \left(\left(3 \cdot F \cdot \abs{\Pi}\right)^{6 \cdot F}\right).
%2 \cdot F \left(3 \cdot F^2 \cdot \abs{\Upsilon}\right)^{2 \cdot F} = 2 \cdot \left(3 \cdot \abs{\Upsilon}\right)^{2} \leq 18 \cdot \abs{\Pi}^2.  
\end{equation}
}

By \Cref{thm:red-large-gap}, there is a constant $C \geq 3$ such that 
$\abs{\cI(\Pi)} \leq C \cdot \abs{\Pi}^3$ if $\abs{\Pi} > C$. If $\abs{\Pi} \leq C$, then the running time of $\cB$ on the input $\Pi$ is bounded by a constant. Otherwise, assume for the following that $\abs{\Pi} > C$ thus $\abs{\cI(\Pi)} \leq \abs{\Pi}^4$. Therefore,
\begin{equation}
\label{eq:NW}
\abs{\cI(\Pi)}+W \leq \abs{\Pi}^4+\abs{\Pi} \leq \abs{\Pi}^5.%\left(3 \cdot \abs{\Pi}\right)^{6}  \leq 3^7 \cdot\abs{\Pi}^6 %\abs{\Pi}^5 \cdot \abs{\Pi}^5 
%\leq \abs{\Pi}^{13}. 
\end{equation} The first inequality holds since $\abs{\cI(\Pi)} \leq \abs{\Pi}^4$. %The second inequality holds since we assume $\abs{\Pi} \geq E \geq 3$. 
In addition, we have
\begin{equation}
\label{eq:zetaBoundk}
5 \cdot \zeta \cdot \frac{d}{\log (d)} = \frac{     5 \cdot \zeta \cdot   \left( 10 \cdot \left(\frac{d}{10}-1\right)+10\right)             }{\log (d)} \leq 	5 \cdot \zeta \cdot \frac{ \left(10 \cdot k+10\right)}{\log (k)} \leq  5 \cdot \zeta \cdot \frac{ 20 \cdot k}{\log (k)} \leq \frac{\alpha \cdot k}{\log (k)}. 
\end{equation} The first inequality holds since $k = \floor{\frac{d}{10}}$; thus, $k \leq d$ and $k \geq \frac{d}{10}-1$. The second inequality holds since $k \geq \frac{d}{10}-1 \geq \frac{d_0}{10}- \geq 1$. The last inequality follows from the selection of $\zeta$. Then, by the running time guarantee of $\cA$ and since the running time of computing $\cI(\Pi)$ can be bounded by $O \left(\abs{\Pi}^3\right)$, executing $\cB$ on $\Pi$ can be done in time 

\begin{equation*}
\label{eq:TIMEF}
\begin{aligned}
O \left(\big(\abs{\cI(\Pi)}+W \big)^{\zeta \cdot \frac{d}{ \log(d)}} +\abs{\Pi}^{3}\right) 
\leq
O \left( \left|\Pi\right|^{5 \cdot \zeta \cdot \frac{d}{ \log \left(d\right)}} +\abs{\Pi}^{3}\right) 
\leq
O \left( \abs{\Pi}^{ \frac{\alpha \cdot k}{	 \log \left(k \right)}
} +\abs{\Pi}^{3}\right)  \leq 
O \left( \abs{\Pi}^{ \frac{\alpha \cdot k}{	 \log \left(k \right)}
}
\right) 
\end{aligned}
\end{equation*} The first inequality follows from \eqref{eq:NW}. The second inequality uses \eqref{eq:zetaBoundk}. The last inequality holds since $\frac{\alpha \cdot k}{	 \log \left(k \right)} \geq 3$ by the assumption on $d_0$ using that $d \geq d_0$.

It remains to prove the two directions of the reduction.
First, assume that $\MaxPar(\Pi) = \abs{V(H)}$. Thus, there is a consistent partial assignment for $\Pi$ of size $|V(H)|$. Then, by  %\Cref{thm:MainCSPtoG-CSP} there is a consistent partial assignment to $\Pi(\Gamma)$ of size $k = |E(H)|$. Thus, by 
\Cref{thm:red-large-gap} there is a solution for $\cI  \left(\Pi\right)$ of profit $\abs{V(H)}$. Therefore, since $\cA$ returns an optimal solution for $\cI   \left(\Pi\right)$, the returned solution $S$ by $\cA$ has profit $\abs{V(H)}$. 
Thus, $\cB$ correctly decides that $\MaxPar(\Pi) = \abs{V(H)}$. Conversely, assume that $\MaxPar(\Pi) < \abs{V(H)}$. Thus, every consistent partial assignment for $\Pi$ has size strictly less than  $\abs{V(H)}$. Therefore, by %\Cref{thm:MainCSPtoG-CSP} there is no consistent partial assignment for $\Pi(\Gamma)$ of size at least $k-\frac{\chi}{6 \cdot \log(k)} \cdot k$. Thus, by 
\Cref{thm:red-large-gap} there is no solution for $\cI  \left(\Pi\right)$ of profit $\abs{V(H)}$. Hence, the returned solution $S$ by $\cA$ has profit strictly less than $\abs{V(H)}$. It follows that $\cB$ returns that $\MaxPar(\Pi) < \abs{V(H)}$ as required. By the above, $\cB$ is decides correctly if $\MaxPar(\Pi) = \abs{V(H)}$ in time $O \left(\abs{\Pi}^{\alpha \cdot \frac{k}{\log (k)}}\right)$. This is a contradiction to \Cref{thm:RCSPeth} and the proof follows.  \qed

\section{A reduction from R-CSP to VK with Varying Dimensions} 
\label{sec:RtoVK}

%\pasin{We need to figure out whether we want to us VK or $d$-VK; currently, the $d$ is not clear at all from the definition.}

%"In this section, we outline our key reduction from \textnormal{R-CSP} to \textnormal{VK} (\Cref{thm:redA}). This reduction is pivotal for our main theorems (\Cref{thm:main} and \Cref{thm:VKexact}), detailed later. The necessity for this second reduction arises because the previous approach (\Cref{sec:simple}) limits solutions to at most $|V(G)|$ items, where $G$ denotes the constraint graph of the R-CSP instance. 
%Given our lower bound applies to $3$-regular graphs, this implies $d = O(|V(G)|)$. 

In this section, we give our main reduction from \textnormal{R-CSP} to \textnormal{VK}. As we will see in the next section, this implies our main result (\Cref{thm:main}).  
%One may justifiably wonder why this second reduction is needed, namely, why the reduction described in \Cref{sec:simple} cannot be used to prove \Cref{thm:main}. The reason is that a solution to the reduction from the previous section may contain at most $|V(G)|$ items, where $G$ is the constraint graph of the given R-CSP instance. 
%The necessity for this second reduction arises because the previous reduction (\Cref{sec:simple}) limits solutions to at most $|V(G)|$ items, where $G$ denotes the constraint graph of the R-CSP instance.
%As our lower bound for R-CSP applies to $3$-regular graphs, the number of dimensions is $d = O(\abs{V(G)})$. %Thus, in time $\abs{\Pi}^{O(d)}$ we can find using exhaustive search the optimal solution for $\cI(\Pi)$, which is significantly weaker then the running time lower bound of \Cref{thm:main} for small error parameter $\eps$.
%Consequently, exhaustive search within time $\abs{\Pi}^{O(d)}$ yields an optimal solution for $\cI(\Pi)$, with significantly faster running time than the lower bound of \Cref{thm:main} for small error parameter $\eps \ll 1$. For example, ruling out  an $n^{O(d+\frac{1}{\eps})}$-time PTAS using the previous reduction does not seem possible.  
We observe the result cannot be attain from \Cref{thm:red-large-gap} since the VK instances generates by \Cref{thm:red-large-gap} {\em can} be solved  exactly in time $\abs{\cI(\Pi)}^{d}$  using an exhaustive enumeration in time $\abs{\Pi}^{\abs{E(G)}}$ on the original R-CSP instance. For example, this implies that the reduction cannot be used to rule out approximation schemes for VK which run in  time $n^{d+\frac{1}{\eps}}$.

The reduction considers an additional integer parameter $F$ besides the R-CSP instance $\Pi$. This parameter is used to control the three key aspects of the reduced $d$-VK instance $\cI(\Pi,F)$. First, the number of dimensions is roughly $d \approx \frac{\abs{V(G)}}{F}$, where $G$ is the constraint graph of $\Pi$; we therefore coin this reduction as the {\em dimension embedding reduction}. Second, the approximation guarantee loses a factor of $2 \cdot F$; thus, there is an almost tight trade-off between the number of dimensions and the approximation guarantee.
The main properties of the reduction are given in \Cref{thm:redA}.   

\begin{lemma} 
	\label{thm:redA}
	\textnormal{\textsf{Dimension Embedding Reduction:}}
	There is a reduction \textnormal{\textsf{R-CSP $\rightarrow $ VK}} that, given an \textnormal{R-CSP} instance %with rectangular constraint instance 
	%$\Pi = (V, E, \Sigma, \{\pi_{e, u}, \pi_{e, v}\}_{e = (u,v)\in E})$ 
	%$\Pi = \left(G,\Sigma,\Upsilon, \{\pi_{e, u}, \pi_{e, v}\}_{e = (u,v)\in E(G)}\right),$
	$\Pi$ whose constraint graph $G$ is $3$-regular and $F \in \left[\abs{V(G)}\right]$
	returns in time $|\Pi|^{O(1)}$ an instance $\cI(\Pi,F) = (I,p,c,B)$ of $d$-\textnormal{VK} such that the following holds. %the following holds.
	\begin{enumerate}
		\item The number of dimensions is $d = 2 \cdot \ceil{\frac{|V(G)|+|E(G)|}{F}}$. %$3 \cdot \frac{|V(G)|+|E(G)|}{F} \leq d \leq 3 \cdot \frac{|V(G)|+|E(G)|}{F}+3$. %, where $G$ is the constraint graph of $\Pi$. 
		\item 
		$\abs{\cI(\Pi)} = O \left(\abs{\Pi}^4  \right)$. 
		%$\abs{\cI(\Pi)} = O \left(\log (|\Upsilon| \cdot F) \cdot  \abs{V(G)} \cdot \abs{ \Sigma} \cdot \left(|V(G)|+|E(G)|\right)  \right)$. %\abs{V(G)} \cdot \abs{\Sigma} \cdot \left(2 \cdot |E(G)|+|V(G)|\right) \cdot |\Upsilon|$.
	%	\item $W\left(\cI(\Pi,F)\right) \leq \left(3 \cdot F \cdot \abs{\Pi}\right)^{6 \cdot F}$.%2 \cdot F \cdot \left(3 \cdot F^2 \cdot |\Upsilon|\right)^{2 \cdot F}$. %\pasin{I don't think we've defined the notation $W(\cI)$ before?}

		\item 	{\bf (Completeness)} If $\MaxPar(\Pi) = |V(G)|$, then there is a solution for $\cI(\Pi, F)$ with profit \\ $\abs{V(G)}+2 \cdot \abs{E(G)}$. 
		
		\item 	{\bf (Soundness)} For every $q \in \mathbb{N}$, if there is a solution for $\cI(\Pi,F)$ of profit at least \\ $\abs{V(G)}+2 \cdot \abs{E(G)}-q$, then $\MaxPar(\Pi) \geq \abs{V(G)} - 2 \cdot q \cdot  F$.

		%if there is a consistent partial assignment to $\Pi$ of size $\abs{V(G)}$ then there is a solution for $\cI(\Pi,F)$ of profit $\abs{V(G)}+2 \cdot \abs{E(G)}$.    
		
		%$q \in \mathbb{N}$, if there is a solution for $\cI(\Pi,F)$ of profit $\abs{V(G)} +2 \cdot \abs{E(G)} - q$, then there is a consistent partial assignment to $\Pi$ of size $\abs{V(G)} - 2 \cdot q \cdot  F$.  
	\end{enumerate}
\end{lemma}

The above reduction considers an R-CSP instance $\Pi = \left(G, \Sigma,\Upsilon, \{\pi_{e, u}, \pi_{e, v}\}_{e = (u,v)\in E(G)} \right)$ and some integer parameter $F \in \left[\abs{V(G)}\right]$. %Our aim is to construct 
%We define the output \textnormal{VK} 
The reduction outputs a $d$-VK instance $\cI(\Pi,F) = (I,p,c,B)$ that in a high level constructed as follows. We let $D = V(G) \cup E(G)$ be the set of {\em constraints}, resembling the use of the vertices and edges as constraints in the reduction of \Cref{sec:simple}. We define an arbitrary partition $D_1,\ldots, D_r$ of $D$ with $F$ constraints in each set (with possibly fewer constraints in the last set). For each $\ell \in [r]$, we make an {\em embedding} of $D_{\ell}$ into only two dimensions, thus having overall $d = 2 \cdot r$ dimensions. 

The items of the constructed instance are $I = V(G) \times \Sigma$ - as in the previous reduction. For each~$\ell\in [r]$ the reduction defines two dimensions $(\ell,1)$ and $(\ell,2)$, which can be viewed as $(2\cdot F+1)$-digit numbers on a basis of a very large number $\mathcal{Q}$. Each of the constraints in $D_{\ell}$ is encoded into one of the first $F$ digits in both dimension. The encoding of each constraint into its digit resembles the encoding used in \Cref{thm:red-large-gap}. The highest digit is only used in $(\ell,2)$, and its goal is to bound the total number of selected items which participate in one of the constraints in $D_{\ell}$. The remaining digits are not used.  The reduction is formally stated as follows.

\begin{definition}
	\label{def:DKPInstance}
	Let $\Pi = \left(G, \Sigma,\Upsilon, \{\pi_{e, u}, \pi_{e, v}\}_{e = (u,v)\in E(G)} \right)$ be an \textnormal{R-CSP} instance, where $G$ is $3$-regular, and let $F \in \left[\abs{V(G)}\right]$. We define the output $d$-\textnormal{VK} instance $\cI(\Pi,F) = (I,p,c,B)$ as follows. 
	\begin{itemize}

		\item The items are pairs of a vertex and a symbol in $\Sigma$, i.e., $I = \{(v,\sigma) ~|~v \in V(G), \sigma \in \Sigma\}$.

%		\item Define $p:I \rightarrow \mathbb{N}$ by $p(i) = 1$ for all $i \in I$. 
		%\item Define $d = 2 \cdot |E(G)|+|V(G)|$ and define $m = |\Upsilon|$.  %and let $m = |E(G)|$. %and 
		
		%		\item Define $D = V(G) \cup \left(E(G) \times \{1,2\} \right)$. 
		
		\item Define $D = V(G) \cup E(G)$. 
		
		%\item Let $\prec$ denote the lexicographic 
		
		\item Let $D_1,\ldots, D_{r}$ be an arbitrary partition of $D$ such that $|D_{\ell}| \leq F$ for all $\ell \in [r]$.
		
			\item Define the number of dimensions as $d = 2 \cdot r$.   
		%\begin{itemize}
		%	\item For all $j \in [r-1]$ it holds that $|D_j| = F$ and $|D_r| \leq F$. 
		%\end{itemize} 
		
		\item For all $\ell \in [r]$ and $v \in V(G)$ define 
		$$J(\ell,v) = \begin{cases}
			|\Adj_G(v) \cap D_{\ell}|+1, & \textnormal{if } v \in D_{\ell} \\
		|\Adj_G(v) \cap D_{\ell}|, & \textnormal{else }
		\end{cases}$$
		Intuitively, $J(\ell,v)$ can be interpreted as  the number of constraints in $D_{\ell}$ in which $v$ participate. 
		%$J(\ell,v) = \Adj_G(v) \cap D_{\ell}+$
		
		%$V_{\ell} = \left(V(G) \cap D_{\ell}\right) \cup \left\{ v \in V(G) \mid \exists e \in E(G) \cap D_{\ell} \textnormal{ s.t. } e \in \Adj_G(v) \right\}$.
		
		\item For all $\ell \in [r]$ define $N_{\ell}  = \sum_{v \in V(G)} J(\ell,v)$ and $I_{\ell} = \left\{   (v,\sigma) \in I \mid J(\ell,v)>0 \right\}$.  
		
			\item Define $p:I \rightarrow \mathbb{N}$ by $p((v,\sigma)) = \sum_{\ell \in [r]} J(\ell,v)$ %\abs{\left\{  \ell \in [r] \mid  v \in V_{\ell} \right\}}$
			 for all $(v,\sigma) \in I$.  
		
		\item Let $m = |\Upsilon|$.  
		
		\item  Define $\cF = 3 \cdot F^2 \cdot m \cdot \abs{V(G)} \cdot \abs{\Sigma}$. 
		\item Define $M = \cF^{2 \cdot F}$. %$M = \sum_{j \in D_{\ell}} m \cdot \cF^{\ord_{\ell}(j)}$. %and let $M = 2 \cdot F \cdot M$. 
		%and associate the $d$ dimensions with elements of $D$. 
		
				\item For all $\ell \in [r]$ let $\ord_{\ell}: D_{\ell} \rightarrow  \left[\abs{D_{\ell}}\right]$ be an arbitrary bijection.

		\item Define $w:I \rightarrow \mathbb{N}^D$ such that for all $(v,\sigma) \in I$ %, $\sigma \in \Sigma$, 
		and $j \in D$: 
		$$w_j((v,\sigma)) = \begin{cases}
			m, & \textnormal{ if } j = v, \\
			\pi_{e, v}(\sigma), & \textnormal{ if } j = e, \textnormal{where } e = (v,u) \textnormal{ and } e \in E(G), \\
			m -	\pi_{e, v}(\sigma), & \textnormal{ if }  j = e, \textnormal{where } e = (u,v) \textnormal{ and } e \in E(G), \\
			%	m-\pi_{e, v}(\sigma), &  j = (e,1), \textnormal{where } e = (v,u) \textnormal{ and } e \in E(G)\\
			%\pi_{e, v}(\sigma), &  j = (e,2), \textnormal{where } e = (v,u) \textnormal{ and } e \in E(G)\\
			0, & \textnormal{otherwise.} %i = (e = (v, u),2) \in E(G)\\
		\end{cases}$$
		
		\item Define $c:I \rightarrow \mathbb{N}^{[r] \times \{1,2\}}$ such that for all $i = (v ,\sigma) \in I$ and $(\ell,t) \in [r] \times \{1,2\}$ define
		$$c_{(\ell,t)}(i) = \begin{cases}
			\sum_{j \in D_{\ell}} w_j(i) \cdot \cF^{\ord_{\ell}(j)}, & \textnormal{ if } i \in I_{\ell} \textnormal{ and } t = 1, \\
		M \cdot J(\ell,v)-\sum_{j \in D_{\ell}} w_j(i) \cdot \cF^{\ord_{\ell}(j)}, & \textnormal{ if } i \in I_{\ell} \textnormal{ and } t = 2, \\
			0, & \textnormal{ if } i \notin I_{\ell}.\\
		\end{cases}$$
		
		%as follows. 
		\item Define the budget $B \in \mathbb{N}^{[r] \times \{1,2\}}$ for all $(\ell,t) \in [r] \times \{1,2\}$ by
		$$B_{\ell,t} = \begin{cases}
\sum_{j \in D_{\ell}} m \cdot \cF^{\ord_{\ell}(j)}, & \textnormal{ if } t = 1, \\
		M \cdot N_{\ell} - \sum_{j \in D_{\ell}} m \cdot \cF^{\ord_{\ell}(j)} & \textnormal{ if } t = 2.
		\end{cases}$$
		%$B_{(\ell,t)} = N_{\ell} \cdot M +M$ for all $(\ell,t) \in [r] \times \{1,2\}$.   
	\end{itemize}
\end{definition}

Henceforth, for every R-CSP instance $\Pi$ with a graph  $G$ and $F \in \left[\abs{V(G)}\right]$, we use $\cI\left(\Pi,F\right)$ to denote the VK instance described in \Cref{def:DKPInstance}. To simplify the presentation, when $\Pi, F$ are clear from the context, we may discard $\Pi, F$ from the notations. We give an illustration of the construction in \Cref{fig:1}. 

\begin{figure}[h]	
	\begin{center}
		\begin{tikzpicture}[ultra thick,scale=1.1, every node/.style={scale=1}]

			\fill[gray, opacity=0.4] (-3,0) rectangle (2,1);
			%	\node at (0.5,-0.5) {$\x$};
			%		\foreach \x in {1,2,3,4,5,6,7,8,9}{
				%			\draw (\x,0)--(\x,1);
				%			\node at (\x-0.5,-0.5) {$\x$};
				%		}
			%		\node at (10.5,-0.5) {$\ldots$};
			%		\node at (11.5,-0.5) {$Q()b$};}

										\fill[gray, opacity=0.25] (-3,3) rectangle (0,4);

			\fill[gray, opacity=0.25] (3,3) rectangle(4.2,4);
			
			\node at (-4,3.5) {$\bf \textcolor{blue}{c_{(\ell,1)}(i)}$};
			\draw (-3,3) rectangle (9,4);
			\draw (0,3)--(0,4);
			%\draw (2,3)--(2,4);
			\draw (3,3)--(3,4);
			
			\draw (6,3)--(6,4);
		
						\node at (-1.5,3.5) {$w_v(i) = m$};
			
					\node at (1.5,3.5) {$w_{j_2}(i) = 0$};
			
								\node at (4.5,3.5) {$w_e(i) = \pi_{(e,u)}(\sigma)$};

				\node at (7.5,3.5) {$w_{j_4}(i) = 0$};
			\node at (-1.5,2.5) {$\ord_{\ell}(v) = 1$};
				\node at (1.5,2.5) {$\ord_{\ell}(j_2) = 2$};
					\node at (4.5,2.5) {$\ord_{\ell}(e) = 3$};
				\node at (7.5,2.5) {$\ord_{\ell}(j_4) = 4$};

								\node at (9,-0.5) {$4 \cdot M$}; 
								
									\node at (6,-0.5) {$3 \cdot M$}; 
									
													\node at (3,-0.5) {$2 \cdot M$}; 
									
									\node at (0,-0.5) {$M$};

			\node at (-4,0.5) {$\bf \textcolor{red}{c_{(\ell,2)}(i)}$};
			
				%\draw (-3,3) rectangle (9,4);
			\draw (-3,0) rectangle (9,1);
			
				\draw (0,0)--(0,1);
			%\draw (2,3)--(2,4);
			\draw (3,0)--(3,1);
			
			\draw (6,0)--(6,1);
			
	\vspace{-1pt}
		\end{tikzpicture}
	\end{center}
	\caption{\label{fig:1} An illustration of the reduction \Cref{def:DKPInstance}. The figure shows the cost of item $i = (v,\sigma) \in I$ in dimensions $(\ell,1)$ and $(\ell,2)$ for some $\ell \in [r]$. The constraints in $D_{\ell}$ are $D_{\ell} = \{j_1,j_2,j_3,j_4\}$ where $j_1 = v$, $j_3 = e = (u,v)$ which is adjacent to $v$, and $j_2,j_4$ are constraints not involving $i$. Thus, $J(\ell,v) = 2$. The constraints are ordered by $j_1,j_2,j_3, j_4$ so that $\ord_{\ell}(j_1) = 1$, $\ord_{\ell}(j_2) = 2$, etc. The cost of $i$ in dimension $(\ell,1)$ is $w_v(i)+w_{e}(i) \cdot \cF^{3}$. Considering this cost as a base-$\cF$ number, the first digit is $m$ and the third digit is $\pi_{(e,u)}(\sigma)$. This is illustrated  as the gray area in the figure upper rectangle, depicting the $4$-digit number in base-$\cF$.  The cost of $i$ in dimension $(\ell,2)$ is $2 \cdot M-c_{(\ell,1)}(i)$ (since $J(\ell,v) = 2$) depicted as the gray area in the lower rectangle. Note that the two rectangles are not in their true proportions.  
	}
\end{figure}
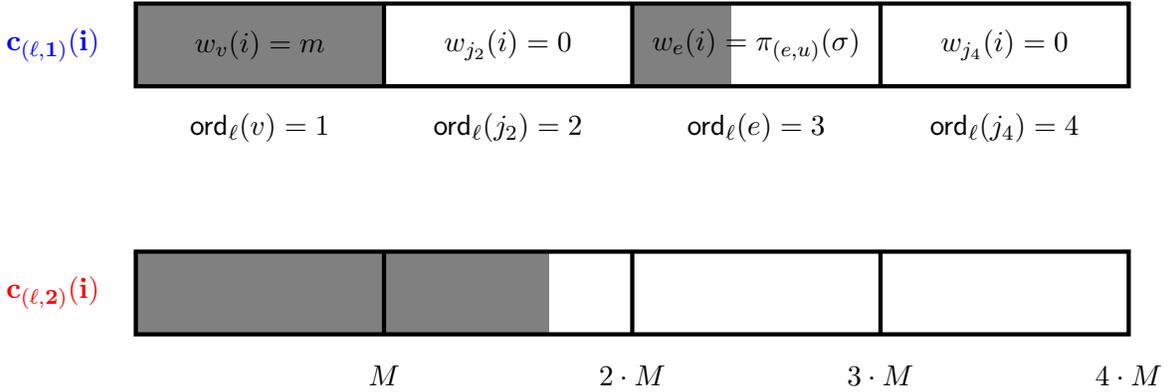

%Before we proceed with the analysis, consider the following example of the construction. 
%\paragraph{Example} Consider an example for an R-CSP instance $\Pi$ with the following properties. A graph $G  = (V,E)$ with two vertices $V = \{u,v\}$ and one edge $E = \{e = (u,v)\}$. Let $\Sigma = \{a,b\}$ and let $\Upsilon = \{1,2\}$. Moreover, assume that $\pi_{e,x} (a) = 1$ and $\pi_{e,x} (b) = 2$ for all $x \in V$. Thus, the edge $e$ is satisfied by an assignment $\varphi:V \rightarrow \Sigma$ if and only $\varphi(u) = \pi_{e,u} (\varphi(u)) = \pi_{e,v} (\varphi(v)) = \varphi(v)$. Now, consider the reduced instance $\cI(\Pi,F)$. The items are all pairs of a vertex and a possible assignment to the vertex: $I = \{u,v\} \times \{a,b\}$, and there are $d = 2 \cdot 1+2 = 4$ budget constraints. We are given a budget constraint $x$ for every vertex $x \in V$; this constraints guarantee that a solution for $\cI(\Pi,F)$ chooses at most one {\em representative} item from $(x,a),(x,b)$, which is analogous to an assignment for the vertex $x$. In addition, there are two constraints $(e,1)$ and $(e,2)$. The first constraint $(e,1)$ guarantees that if a solution for $\cI(\Pi,F)$ chooses a representative item for both vertices $u,v$, then the value of $ \pi_{e,u} (\varphi(u))$ is at least as large as $ \pi_{e,v} (\varphi(v))$; conversely, the constraint $(e,2)$ restricts $ \pi_{e,u} (\varphi(u))$ to be smaller or equal than $ \pi_{e,v} (\varphi(v))$, yielding that the assignment, e.g., $\varphi(u) = a, \varphi(v) = a$ satisfies the edge $e$. 

Clearly, given an R-CSP instance $\Pi$ with a constraint graph  $G$ and $F \in \left[\abs{V(G)}\right]$, the instance $\cI(\Pi,F)$ can be constructed in time $|\Pi|^{O(1)}$. %We start by proving the properties described in \Cref{thm:redA}. 
%To simplify the presentation, for any \textnormal{R-CSP} instance 
%$\Pi$ with a graph $G$ and $F \in \left[\abs{V(G)}\right]$, when $\Pi$ is clear from the context, we refer in this section to all values and objects defined in \Cref{def:DKPInstance} w.r.t. $\Pi$ and $F$ without a repeated deceleration.
%Intuitively, solutions for $\cI$ of high profit have larger 
We state the following observation (which follows directly from definitions) regarding the costs and profits. This observation will be helpful throughout this section. 

\begin{obs} \label{obs:basic}
For any $S \subseteq I$ and $\ell \in [r]$, we have
\begin{enumerate}[(i)]
\item $c_{\ell, 1}(S) = \sum_{j \in D_{\ell}} w_j(S) \cdot \cF^{\ord_{\ell}(j)}$.
\item $c_{\ell, 2}(S) = M \cdot \sum_{(v, \sigma) \in S} J(\ell, v) - \sum_{j \in D_{\ell}} w_j(S) \cdot \cF^{\ord_{\ell}(j)}$.
\item $p(S)	= \sum_{\ell \in [r]~} \sum_{(v,\sigma) \in S}  J(\ell,v)$.
\end{enumerate}
\end{obs}

Next, we prove some basic properties of the reduction.

	\begin{lemma}
	\label{claim:NF}
		For any \textnormal{R-CSP} instance 
	$\Pi = \left(G, \Sigma,\Upsilon, \{\pi_{e, u}, \pi_{e, v}\}_{e = (u,v)\in E(G)},\perp \right)$ and $F \in \left[\abs{V(G)}\right]$ the following holds. 
	
	\begin{enumerate}
		\item $\sum_{\ell \in [r]} N_{\ell} = 	\abs{V(G)}+ 2 \cdot \abs{E(G)}$. 
		
		\item 	For all $\ell \in [r]$ it holds that $N_{\ell} \leq 2 \cdot F$. 
		
	%	\item 	$3 \cdot \frac{|V(G)|+|E(G)|}{F} \leq d \leq 3 \cdot \frac{|V(G)|+|E(G)|}{F}+3$. 
	
	\item  $d = 2 \cdot \ceil{\frac{|V(G)|+|E(G)|}{F}}$. 
		
%		\item $W\left(\cI(\Pi,F)\right) = \left(3 \cdot F \cdot \abs{\Pi}\right)^{6 \cdot F}$.%2 \cdot F \cdot \left(3 \cdot F^2 \cdot m\right)^{2 \cdot F}$. 
		
		\item $\abs{\cI(\Pi,F)} \leq O \left( \abs{\Pi}^4\right)$.
	\end{enumerate}
\end{lemma}

\begin{proof}
	The first property of the lemma follows from 
	 \begin{equation*}
		\label{eq:profitTwo}
		\begin{aligned}
			\sum_{\ell \in [r]} N_{\ell}= {} &  \sum_{\ell \in [r]} \sum_{v \in V(G)} J(\ell,v) =  \sum_{v \in V(G)}  \sum_{\ell \in [r]} J(\ell,v) =  \sum_{v \in V(G)} \left(\abs{\Adj_G(v)}+\abs{\{v\}}\right) = \abs{V(G)}+ 2 \cdot \abs{E(G)}. 
		\end{aligned}
	\end{equation*} %The third equality uses the definition of $J$ and holds since for each $j \in V(G) \cup E(G)$ there is exactly one $\ell \in [r]$ such that $j \in D_{\ell}$. The last equality holds since for each edge $e = (u,v) \in E(G)$ there are exactly two endpoints $u$ and $v$; thus, $e$ belongs to $\Adj_G(v)$ and $\Adj_G(u)$ but does not belong to $\Adj_G(x)$ for any $x \in V(G) \setminus \{u,v\}$. 
	
	For the second property of the lemma, %by the construction in \Cref{def:DKPInstance}, 
	for all $\ell \in [r]$, it holds that 
	\begin{equation}
			\label{eq:profitTwo'}
		\begin{aligned}
		N_{\ell} = \sum_{v \in V(G)} J(\ell,v) 
		=
		 \left|V(G) \cap D_{\ell}\right|+ 2 \cdot \left|E(G) \cap D_{\ell}\right| \leq 2 \cdot \abs{D_{\ell}} \leq 2 \cdot F.
%		 
%		 \abs{D_{\ell} \cap V(G)} +\abs{\bigcup_{(u,v) \in D_{\ell} \cap E(G)} \{u,v\}} \leq \sum_{v \in D_{\ell} \cap V(G)} 1+\sum_{(u,v) \in D_{\ell} \cap E(G)} 2 \leq \sum_{j \in D_{\ell}} 2 = 2 \cdot \abs{D_{\ell}} \leq 2 \cdot F.
		\end{aligned}
	\end{equation}
	
		The third property clearly holds in \Cref{def:DKPInstance}. By  \Cref{def:DKPInstance}, the largest number in the instance $\cI(\Pi,F)$ is bounded by $M \cdot N_{\ell}$. Thus,
		$$W\left(\cI(\Pi,F)\right) \leq M \cdot N_{\ell} \overset{\eqref{eq:profitTwo'}}{\leq} 2 \cdot F \cdot M = 2 \cdot F \cdot \cF^{2 \cdot F} = 2 \cdot F \cdot \left(3 \cdot F^2 \cdot m \cdot \abs{V(G)} \cdot \abs{\Sigma}\right)^{2 \cdot F} \leq \left(3 \cdot F \cdot \abs{\Pi}\right)^{6 \cdot F}.$$
		%The second equality follows from \eqref{eq:profitTwo'}. 
		
%		Finally, we prove the last property of the lemma. By \Cref{def:DKPInstance} the number of items of $\cI(\Pi,F)$ can be expressed as $|I| = \abs{V(G)} \cdot \abs{ \Sigma}$. Moreover, by the above the number of dimensions is bounded by $d \leq 3 \cdot \frac{|V(G)|+|E(G)|}{F}+3$ and the largest number encoded in the instance is bounded by $W\left(\cI(\Pi,F)\right) \leq O \left(\left(3 \cdot F \cdot \abs{\Pi}\right)^{3 \cdot F}\right)$. Thus, encoding each number in the instance can be done in space $O \left(\log \left( \left(3 \cdot F \cdot \abs{\Pi}\right)^{3 \cdot F} \right)    \right) = O \left(F \cdot \log (\abs{\Pi}) \right)$.
		%Overall
	Therefore, the reduced instance $\cI(\Pi,F)$ can be encoded in space 
		\begin{equation*}
			\begin{aligned}
		O\left(|I| \cdot d \cdot \log W\left(\cI(\Pi,F)\right) \right) =	O \left(\abs{V(G)} \cdot \abs{ \Sigma} \cdot  \frac{|V(G)|+|E(G)|}{F} \cdot F \cdot \log \abs{\Pi} \right) \leq{} & O \left( \abs{\Pi}^4\right),
			\end{aligned} \qedhere
		\end{equation*}
		which shows the last property. 
		%
		%\comment{By \Cref{def:DKPInstance} the number of items of $\cI(\Pi,F)$ can be expressed as $|I| = \abs{V(G)} \cdot \abs{ \Sigma}$. Moreover, by \Cref{claim:NF}, the number of dimensions is bounded by $$d \leq 3 \cdot \frac{|V(G)|+|E(G)|}{F}+3.$$ 
		%	Third, by \Cref{claim:NF}, the largest number encoded in the instance is bounded by $$W\left(\cI(\Pi,F)\right) \leq 2 \cdot F \cdot \left(3 \cdot F^2 \cdot m\right)^{2 \cdot F}.$$ Thus, encoding each number in the instance can be done in space $$O \left(\log \left( 2 \cdot F \cdot \left(3 \cdot F^2 \cdot m\right)^{2 \cdot F} \right)    \right) = O \left(F \cdot \log (m \cdot F) \right).$$
		%	Overall, the reduced instance $\cI(\Pi,F)$ can be encoded in space 
		%	$$O \left(       F \cdot \log (m \cdot F) \cdot  \abs{V(G)} \cdot \abs{ \Sigma} \cdot  \frac{|V(G)|+|E(G)|}{F}               \right) = O \left(\log (|\Upsilon| \cdot F) \cdot  \abs{V(G)} \cdot \abs{ \Sigma} \cdot \left(|V(G)|+|E(G)|\right)  \right).$$
		%	By the above the proof follows. }
		%By the above, we prove all properties of the lemma. 
\end{proof}

%
%\begin{proof}
%	Let $S$ be a solution for $\cI(\Pi,F)$ and let $j \in D$ be a dimension of $\cI(\Pi,F)$. Consider the following~cases.
%		\begin{enumerate}
%		\item $j = v$ for some $v \in V(G)$. Then, clearly $\left|\left\{i \in S \mid c_j(i)>0\right\}\right| \leq 1$ since either $c_j(i) = 0$ or $c_j(i) = m$, where $m = B_j$. 
%		\item $j = (e,\ell)$ for some $e = (u, v) \in E(G)$ and $\ell \in \{1,2\}$. Assume towards a contradiction that $\left|\left\{i \in S \mid c_j(i)>0\right\}\right| > 2$. Observe that by \Cref{def:DKPInstance}, for all $i \in S$ such that $c_j(i)>0$ there is $\sigma \in \Sigma$ such that $i = (u,\sigma)$ or $i = (v,\sigma)$. Thus,
%		 there are $i_1,i_2 \in S$, $\sigma_1,\sigma_2$, where $\sigma_1 \neq \sigma_2$, and $w \in \{u,v\}$, such that $i_1 = (w,\sigma_1)$, $i_2 = (w,\sigma_2)$. Therefore,
%		 \begin{equation}
%		 	\label{eq:violateW}
%		 	c_w(S) \geq c_w(i_1)+c_w(i_2) = 2 \cdot m > m = B_w.
%		 \end{equation} The first inequality holds since $i_1,i_2 \in S$ and $i_1 \neq i_2$. The first equality holds since $i_1 = (w,\sigma_1)$, $i_2 = (w,\sigma_2)$. By \eqref{eq:violateW} we reach a contradiction to the fact that $S$ is a solution for $\cI(\Pi,F)$. Thus, $\left|\left\{i \in S \mid c_j(i)>0\right\}\right| \leq 2$.
%	\end{enumerate} By the above, $\left|\left\{i \in S \mid c_j(i)>0\right\}\right| \leq 2$ in all cases. It follows that $\cI(\Pi,F)$ is $2$-flat by \Cref{def:flat}. 
%\end{proof}

We give below the completeness and soundness of the reduction. We start with the former, which is (relatively) straightforward. %Intuitively, given a consistent partial assignment to $\Pi$ we can construct a solution for $\cI$

	\begin{lemma}
	\label{lem:VKdir1}
	{\bf (Completeness)}	For any \textnormal{R-CSP} instance 
	$\Pi = \left(G, \Sigma,\Upsilon, \{\pi_{e, u}, \pi_{e, v}\}_{e = (u,v)\in E(G)}\right)$ 
and $F \in  \left[\abs{V(G)}\right]$, if $\MaxPar(\Pi) = \abs{V(G)}$ then there is a solution for $\cI(\Pi,F)$ of profit $\abs{V(G)}+2 \cdot \abs{E(G)}$.    
\end{lemma}

\begin{proof}
	%Let $\Pi = \left(G, \Sigma,\Upsilon, \{\pi_{e, u}, \pi_{e, v}\}_{e = (u,v)\in E(G)},\perp \right)$ be an R-CSP instance, let $F \in  \left[\abs{V(G)}\right]$, and 
	Let $\varphi:V(G) \rightarrow \left(\Sigma \cup \{\perp\}\right)$ be a  consistent partial assignment to $\Pi$ of size $\abs{V(G)}$. %In addition, let $\cI(\Pi,F) = (I,p,c,B)$ be the reduced instance defined in \Cref{def:DKPInstance}. %and let $m = |\Upsilon|$. 
	Let $S = \{(v,\varphi(v))~|~v \in V(G)\}.$ %Observe that for every $v \in V(G)$ it holds that $(v,\varphi(v)) \in I = V(G) \times \Sigma$ since $\varphi$ is a consistent partial assignment of size $\abs{V(G)}$.
	% Thus, it holds that 
%	 \begin{equation}
%		\label{eq:profitOne}
%		\begin{aligned}
%			p(S) ={} & \sum_{i \in S} p(i) 
%			\\={} &  \sum_{(v,\sigma) \in S} J(\ell,v)
%			\\={} &
%			\sum_{(v,\varphi(v)) \in S} J(\ell,v)
%			\\={} &
%			\sum_{v \in V(G)} J(\ell,v)
%			\\={} &
%		\abs{V(G)}+2 \cdot \abs{E(G)}
%		%	\\={} &
%			%\sum_{\ell \in [r]} N_{\ell}. 
%			%  	 \\={} &
%			%	  \sum_{\ell \in [r]} N_{\ell}. 
%		\end{aligned}
%	\end{equation} 
	From the third item of \Cref{obs:basic} and the first item of \Cref{claim:NF}, the profit is equal to $$p(S) = \sum_{\ell \in [r]~} \sum_{v \in V(G)} J(\ell,v) = \sum_{\ell \in [r]~} N_{\ell} = \abs{V(G)}+2 \cdot \abs{E(G)}.$$

	 To conclude, it remains to prove that $S$ is a solution for $\cI(\Pi,F)$, %which requires to % by proving 
	 i.e. that all budget constraints are satisfied. 
	 To see this, first notice that for any $v \in V(G)$, we have $w_v(S) = w_v((v, \varphi(v))) = m$. Moreover, for every $e = (u, v) \in E(G)$, we have $$w_e(S) = w_e(u,\varphi(u))+w_e(v,\varphi(v)) = \pi_{e, u}\left(\varphi(u)\right)+m- \pi_{e, v}\left(\varphi(v)\right) = m,$$ where the last equality is from consistency of $\varphi$. In other words, we have
	 \begin{align} \label{eq:completeness-w-tight}
	 w_j(S) = m & &\forall j \in D.
	 \end{align}

	 Let $(\ell,t) \in [r] \times \{1,2\}$ be a budget constraint. Consider the two cases depending on the value of $t$:
	 \begin{itemize}
	 \item $t = 1$. In this case, \eqref{eq:completeness-w-tight} and the first item of \Cref{obs:basic} immediately yield $c_{(\ell, t)}(S) = B_{\ell, t}$.
	 \item $t = 2$. In this case, by \eqref{eq:completeness-w-tight}, the second item of \Cref{obs:basic} and the definition of $N_\ell$, we have
	 \begin{align*}
	 c_{(\ell, t)}(S) = M \cdot \sum_{v \in V(G)} J(\ell, v) - \sum_{j \in D_{\ell}} m \cdot \cF^{\ord_{\ell}(j)} = M \cdot N_\ell - \sum_{j \in D_{\ell}} m \cdot \cF^{\ord_{\ell}(j)} = B_{\ell, t}.
	 \end{align*}
	 \end{itemize}

	Thus, %\eqref{eq:Case:t=1} and \eqref{eq:Case:t=2} 
	$S$ is a feasible solution for $\cI(\Pi,F)$ and this completes the proof. 
\end{proof}

For the soundness, we prove that given a solution to the reduced instance $\cI(\Pi,F)$, we can construct a consistent partial assignment for $\Pi$ with roughly the same size/profit. Before that, we give a couple of useful properties over the weights $w$ of any solution of $\cI(\Pi, F)$: %For some R-CSP instance $\Pi$ and a number $F \in \left[\abs{V(G)}\right]$,  let $T = \left\{ \ell \in D_{\ell} \mid \abs{S \cap I_{\ell}} = N_{\ell} \right\}$ be the set of {\em tight constraints} of the reduced instance $\cI(\Pi,F)$. The tight constraint play a vital role in the analysis. 

	\begin{lemma}
	\label{lem:VKw}
	 %with rectangular constraint instance 
	%$\Pi = (V, E, \Sigma, \{\pi_{e, u}, \pi_{e, v}\}_{e = (u,v)\in E})$ 
	%Let $\Pi$ be an \textnormal{R-CSP} instance, let $F \in \left[\abs{V(G)}\right]$, and 
	Let $S$ be any solution for $\cI(\Pi,F)$. Then, for any $\ell \in [r]$, the following holds:
	\begin{enumerate}[(i)]
	\item $\sum_{(v,\sigma) \in S} J(\ell,v) \leq  N_{\ell}$. \label{item:count-constraint}
	\item If $\sum_{(v,\sigma) \in S} J(\ell,v) =  N_{\ell}$, then $w_j(S) = m$ for all $j \in D_\ell$.% \pasin{Did we define $w_j(S)$ (as the appropriate sum) before?} 
	\label{item:count-exact-tight}
	\end{enumerate}
	%In addition, let $\ell \in [r]$ such that $\sum_{(v,\sigma) \in S} J(\ell,v) =  N_{\ell}$. Then, for all $j \in D_{\ell}$ it holds that $w_j(S) = m$.  %= \left(G, \{\Sigma_v\}_{v \in V(G)}, \{\Sigma_e\}_{e \in E(G)}, \{\pi_{e, u}, \pi_{e, v}\}_{e = (u,v)\in E}) \right)$
	%and $q \in \mathbb{N}$, if there is a solution for $\cI(\Pi,F)$ of profit $q$, then there is a consistent partial assignment to $\Pi$ of size $q-10 \cdot F$.  
\end{lemma}

To prove this, we use the following fact that an integer can be uniquely represented in base-$N$.
	
	\begin{fact}
		\label{obs:number}
		Let $N \in \mathbb{N}$ and $a_1,\ldots,a_n,A \in \{0, \dots, N-1\}$ be such that $\sum_{i \in [n]} a_i \cdot N^i = \sum_{i \in [n]} A \cdot N^i$. Then, for all $i \in [n]$ it holds that $a_i = A$. 
	\end{fact}

\begin{proof}[Proof of \Cref{lem:VKw}]
	%Let $\Pi = \left(G, \Sigma,\Upsilon, \{\pi_{e, u}, \pi_{e, v}\}_{e = (u,v)\in E(G)} \right)$ be an R-CSP instance, let $F \in \left[\abs{V(G)}\right]$, and let $\cI(\Pi,F) = (I,p,c,B)$ be the reduced instance defined in \Cref{def:DKPInstance}. In addition, let $S$ be a  solution for $\cI(\Pi,F)$. %, and let $m = |\Upsilon|$. 
Since $S$ is a solution for $\cI(\Pi, F)$,
\begin{align}
M \cdot N_{\ell} = B_{\ell, 1} + B_{\ell, 2} \geq c_{\ell, 1}(S) + c_{\ell, 2}(S) = M \cdot \sum_{(v, \sigma) \in S} J(\ell, v), \label{eq:count-constraint-temp}
\end{align}
where the last equality follows from the first two items of \Cref{obs:basic}. Thus, \Cref{item:count-constraint} holds. 
For \Cref{item:count-exact-tight}, note that for \eqref{eq:count-constraint-temp} to be an equality, we must have  $B_{\ell, 1} = c_{\ell, 1}(S)$ otherwise violating the feasibility of $S$. This means that
\begin{align}
\sum_{j \in D_{\ell}} m \cdot \cF^{\ord_{\ell}(j)} = B_{\ell, 1} = c_{\ell, 1}(S) = \sum_{j \in D_{\ell}} w_i(S) \cdot \cF^{\ord_{\ell}(j)}. \label{eq:base-num-equal}
\end{align}
Observe that $w_j(S) \leq |I| \cdot m < Q$ by our setting of parameters.
%\begin{equation}
%		\label{eq:obsS2}
%w_j(S) \leq |S| \cdot \max_{i \in S} w_j(i) \leq |S| \cdot m \leq |I| \cdot m = 	m \cdot \abs{V(G)} \cdot \abs{\Sigma} < 3 \cdot F^2 \cdot m \cdot \abs{V(G)} \cdot \abs{\Sigma} = \cF.
%	\end{equation} Observe that $S \neq \emptyset$ using \Cref{claim:VK-1}; thus, the first inequality holds. By \eqref{eq:obsS} and \eqref{eq:obsS2} it holds that $w_j(S) \in =[\cF-1]$ for all $j \in D_{\ell}$ such that $\sum_{j \in D_{\ell}} w_j(S) \cdot \cF^{\ord_{\ell}(j)} = \sum_{j \in D_{\ell}} m \cdot \cF^{\ord_{\ell}(j)}$. Hence, the proof follows from \Cref{obs:number}.
From this, \eqref{eq:base-num-equal} and \Cref{obs:number}, we can conclude that $w_j(S) = m$ for all $j \in D_\ell$ as desired.
\end{proof}

Using \Cref{lem:VKw}, we can give the second direction of the reduction. %Intuitively,  \Cref{lem:VKw} guarantees that for a solution $S$ of $\cI$ and $\ell \in [r]$ such that $\sum_{(v,\sigma) \in S} J(\ell,v) =  N_{\ell}$, we can construct an assignment $\varphi$

	\begin{lemma}
	\label{lem:VKdir2}
		{\bf (Soundness)} For any \textnormal{R-CSP} instance %with rectangular constraint instance 
	%$\Pi = (V, E, \Sigma, \{\pi_{e, u}, \pi_{e, v}\}_{e = (u,v)\in E})$ 
	$\Pi$, $F \in \left[\abs{V(G)}\right]$, %= \left(G, \{\Sigma_v\}_{v \in V(G)}, \{\Sigma_e\}_{e \in E(G)}, \{\pi_{e, u}, \pi_{e, v}\}_{e = (u,v)\in E}) \right)$
	and $q \in \mathbb{N}$, if there is a solution for $\cI(\Pi,F)$ of profit at least $\abs{V(G)} +2 \cdot \abs{E(G)} - q$, then $\MaxPar(\Pi) \geq \abs{V(G)} - 2 \cdot q \cdot  F$.  
\end{lemma}

\begin{proof}
		%Let $\Pi = \left(G, \Sigma,\Upsilon, \{\pi_{e, u}, \pi_{e, v}\}_{e = (u,v)\in E(G)},\perp \right)$ be an R-CSP instance, let $F \in \left[\abs{V(G)}\right]$, and let $\cI(\Pi,F) = (I,p,c,B)$ be the reduced instance defined in \Cref{def:DKPInstance}. In addition, let $q \in \mathbb{N}$ and let $S$ be a  solution for $\cI(\Pi,F)$ of profit at least $\abs{V(G)} +2 \cdot \abs{E(G)}-q$. 
		For every $j \in D$, let $\ell(j) \in [r]$ be such that $j$ belongs to $D_{\ell(j)}$.
		%For all $\ell \in [r]$ and $j \in D_{\ell}$, with a slight abuse of notation define $\ell(j) = \ell$.
		In addition, %let $T = \left\{ \ell \in D_{\ell} \mid \abs{S \cap I_{\ell}} = N_{\ell} \right\}$ be the set of {\em tight constraints} of $S$.  
		let $$T = \left\{ \ell \in [r] ~\bigg|~ \sum_{(v,\sigma) \in S} J(\ell,v) = N_{\ell} \right\}$$ be the set of {\em tight constraints} of $S$.  
		Define 
		\begin{equation}
			\label{eq:Vstar}
			V^* = %\left\{   v \in V(G) ~\bigg|~   \abs{S \cap I_{\ell(v)}} = N_{\ell(v)} \right\} \cap 
			\left\{   v \in V(G) ~\bigg|~   \ell(x) \in T ~\forall x \in \left(\Adj_G(v) \cup \{v\} \right) \right\}.
		\end{equation} % words, $V^*$ contains all vertices $v$ such that (i) the set $D_{\ell(v)}$ 
		The set $V^*$ satisfies the following crucial property.
		
			\begin{claim}
			\label{claim:Vstar}
			For every $v \in V^*$ there is exactly one $\sigma_v \in \Sigma$ such that $(v,\sigma_v) \in S$.
		\end{claim}
		\begin{claimproof}
			%Let $v \in V^*$. 
			By \eqref{eq:Vstar}, $\ell(v)$ must be {\em tight}, i.e., $\sum_{(u,\sigma) \in S} J(\ell(v),u) = N_{\ell(v)}$.  Thus, by \Cref{item:count-exact-tight} of \Cref{lem:VKw},
			$$m = w_v(S) = m \cdot \left| \left\{\sigma \in \Sigma \mid (v,\sigma) \in S \right\}\right|,$$
			%$$\left| \left\{ (v,\sigma) \in S \right\}\right| \cdot m = \sum_{(v,\sigma) \in S} w_v((v,\sigma)) = \sum_{i \in S} w_v(i) = w_v(S) = m$$
			%The first and second equalities follow by the definition of $w$. The last equality follows from by \Cref{item:count-exact-tight} of \Cref{lem:VKw} as $\sum_{(u,\sigma) \in S} J(\ell(v),u) = N_{\ell(v)}$. Thus, by the above we conclude that $\left| \left\{ (v,\sigma) \in S \right\}\right| = 1$, 
			which implies the statement of the claim. 
		\end{claimproof}
		
		For all $v \in V^*$ denote by $\sigma_v \in \Sigma$ the unique symbol satisfying that $(v,\sigma_v) \in S$; by \Cref{claim:Vstar} it holds that $\sigma_v$ is well defined for every $v \in V^*$. Define $\varphi: V(G) \rightarrow \Sigma \cup \{\perp\}$ by
		$$\varphi(v) = \begin{cases}
			\sigma_v, & \textnormal{if } v \in V^* \\
			\perp, & \textnormal{else }
		\end{cases}$$
		for all $v \in V(G)$. We first prove the consistency of $\varphi$.
		
			\begin{claim}
			\label{claim:Consistency}
			$\varphi$ is a consistent partial assignment for $\Pi$.
		\end{claim}
		\begin{claimproof}
			Let $e = (u,v) \in E(G)$ such that $\varphi(u) \neq \perp$ and $\varphi(v) \neq \perp$. %; we need to show that $\pi_{e,u}(\varphi(u)) = \pi_{e,v}(\varphi(v))$. 
			Since $\varphi(u) \neq \perp$ and $\varphi(v) \neq \perp$ it follows from the definition of $\varphi$ that $u,v \in V^*$. Therefore, by \eqref{eq:Vstar} it holds that $\ell(e)\in T$ implying that $\sum_{(x,\sigma) \in S} J(\ell(e),x) = N_{\ell(e)}$. 
			Thus, by \Cref{item:count-exact-tight} of \Cref{lem:VKw}, we have
			\begin{align*}
			m = w_e(S) = w_e((u,\sigma_u))+ w_e((v,\sigma_v)) = \pi_{e,u}(\varphi(u))+m-	\pi_{e,v}(\varphi(v)).
			\end{align*}
			It follows that $\pi_{e,u}(\varphi(u)) = \pi_{e,v}(\varphi(v))$. Thus, $\varphi$ is a consistent partial assignment for $\Pi$. 
			%Thus, %by \Cref{lem:VKw} we have that 
			%\begin{equation}
			%	\label{eq:VstarConst}
			%	\begin{aligned}
			 %\pi_{e,u}(\varphi(u))+m-	\pi_{e,v}(\varphi(v)) ={} & \pi_{e,u}(\sigma_u)+m-	\pi_{e,v}(\sigma_v) 
			 %\\={} &
			 % w_e((u,\sigma_u))+ w_e((v,\sigma_v)) 
			 %\\={} &
			 %  \sum_{(u,\sigma) \in S} w_e((u,\sigma))+\sum_{(v,\sigma) \in S} w_e((v,\sigma)) 
			%\\={} &
			%    \sum_{i \in S} w_e(i) 
			%\\={} &
			 %    w_e(S) 
			 %  			\\={} &
			 %     m. 
			%	\end{aligned}
			%\end{equation}
			%The first equality follows from the definition of $\varphi$. The second and fourth equalities follow by the definition of $w$. The third equality holds since $\sigma_v$ and $\sigma_u$ are unique by \Cref{claim:Vstar}. 
			%The last equality follows from by \Cref{item:count-exact-tight} of \Cref{lem:VKw} as $\sum_{(x,\sigma) \in S} J(\ell(e),x) = N_{\ell(e)}$. By \eqref{eq:VstarConst} it follows that $\pi_{e,u}(\varphi(u)) = \pi_{e,v}(\varphi(v))$ as required. We conclude that $\varphi$ is a a consistent partial assignment for $\Pi$. 
		\end{claimproof}
		
		%By \Cref{claim:Consistency} 
		It remains to give a lower bound on the size of $\varphi$. To do so, we will need the following bound on the number of non-tight indices.   
		
			\begin{claim}
			\label{claim:nonTight}
			$\left|[r]\setminus T\right| \leq q$. 
		\end{claim}
		\begin{claimproof}
			First, from the third item of \Cref{obs:basic}, we have
			\begin{align*}
	p(S) =
	\sum_{\ell \in [r]~} \sum_{(v,\sigma) \in S}  J(\ell,v)
		={} 
		\sum_{\ell \in T} \sum_{(v,\sigma) \in S}  J(\ell,v)+\sum_{\ell \in [r]\setminus T} \sum_{(v,\sigma) \in S}  J(\ell,v)
		={} 
		\sum_{\ell \in T} N_{\ell}+\sum_{\ell \in [r]\setminus T} \sum_{(v,\sigma) \in S}  J(\ell,v)
		\end{align*}
		Now, for every $\ell \in [r]\setminus T$, it holds that $\sum_{(v,\sigma) \in S}  J(\ell,v) \neq N_{\ell}$; \Cref{item:count-constraint} of \Cref{lem:VKw} and the fact that $J$  assigns only integral value then implies that $\sum_{(v,\sigma) \in S}  J(\ell,v) \leq N_{\ell}-1$. Plugging this into the above, we get
		\begin{align*}
		p(S) \leq \sum_{\ell \in T} N_{\ell}+\sum_{\ell \in [r]\setminus T} \left(N_{\ell}-1\right) =  \sum_{\ell \in [r]} N_{\ell}-\abs{[r]\setminus T} = \abs{V(G)}+2 \cdot \abs{E(G)}-\abs{[r]\setminus T},
		\end{align*}
		where the last equality follows from \Cref{claim:NF}.
		Finally, the claim follows since we assume that $p(S) \geq \abs{V(G)}+2 \cdot \abs{E(G)}- q$.
		\end{claimproof}
		
		To conclude the soundness proof, observe that
		\begin{equation}
			\label{eq:Vstarcomp}
			\begin{aligned}
			\abs{V(G) \setminus V^*} ={} & \abs{ \left\{ v \in V(G) ~\bigg|~ \exists x \in \left(\Adj_G(v) \cup \{v\}\right) \textnormal{ s.t. } \ell(x) \in [r] \setminus T \right\}}
			\\\leq {} &
			\sum_{\ell \in [r]  \setminus T~} \sum_{(v,\sigma) \in S}  J(\ell,v)%\abs{V_{\ell}} 
					\\\leq {} &
			\sum_{\ell \in [r]  \setminus T} N_{\ell}
%				\\\leq {} &
%			\sum_{\ell \in [r] \textnormal{ s.t. } \abs{S \cap I_{\ell}} < N_{\ell}} \abs{V_{\ell}} 
%				\\= {} &
%			\sum_{\ell \in [r] \textnormal{ s.t. } \abs{S \cap I_{\ell}} < N_{\ell}} N_{\ell} 
					\\\leq {} &
			\sum_{ \ell \in [r]  \setminus T} 2 \cdot F
				\\= {} &
		\abs{[r] \setminus T} \cdot 2 \cdot F
			\\\leq {} &
		2 \cdot q  \cdot F. 
			\end{aligned}
		\end{equation} The first inequality holds since for each $v \in V(G)$ such that  there is $x \in \left(\Adj_G(v) \cup \{v\}\right)$ satisfying $\ell(x) \in [r] \setminus T$, it also holds that $J(\ell(x),v) \geq 1$. The second inequality is due to \Cref{item:count-constraint} of \Cref{lem:VKw}. The third inequality follows from \Cref{claim:NF}.  The last inequality uses \Cref{claim:nonTight}. Hence, 
		\begin{equation}
			\label{eq:sizeVARPHI}
			\abs{\varphi} = \abs{V^*} = \abs{V(G)}-\abs{V(G) \setminus V^*} \overset{\eqref{eq:Vstarcomp}}{\geq} \abs{V(G)}-2 \cdot q  \cdot F.  
		\end{equation} By \Cref{claim:Consistency} and the above inequality, the soundness proof follows. 
\end{proof}

The above lemmas give the statement of the reduction.

\subsubsection*{Proof of \Cref{thm:redA}:} The proof follows from \Cref{claim:NF}, %\Cref{lem:VK:W}, 
\Cref{lem:VKdir1}, and \Cref{lem:VKdir2}. \qed%,  and \Cref{lem:VK:I}. \qed

\section{Proofs of the Remaining Lower Bounds}
\label{sec:mainResults}

In this section, we prove our remaining lower bounds. We start by proving \Cref{thm:main} relying on the reduction presented in \Cref{sec:RtoVK}. Then, we give the hardness results for larger approximation ratio (\Cref{thm:sqrt-runningtime-lb} and \Cref{thm:sqrt-ratio}). These proofs are easier and based on the simpler reduction presented in \Cref{sec:simple}.

\subsubsection*{Proof of \Cref{thm:main}:}

Assume that \textnormal{Gap-ETH} holds. Let $\alpha,\beta \in (0,1)$ and $k_0 \in \N$ be the promised constants by \Cref{thm:RCSP}. 
%Define a constant $b = \frac{\chi}{6}$. 
Let $C\geq1$ be a constant such that for any R-CSP instance $\Pi$ it holds that the reduction \textnormal{\textsf{{R-CSP $\rightarrow$ VK}}} (described in \Cref{thm:redA}) runs in time $O \left(\abs{\Pi}^C\right)$. %; there is such a constant by \Cref{thm:redA}. 
Define constants $\zeta = \frac{\alpha \cdot \beta}{10,000}$ and $\chi = \frac{\beta}{4,000}$.
Let $d_0 \in \N$ such that the following holds $\alpha \cdot \frac{d_0}{\log (d_0)} \geq \max \{k_0,6,C\}$, $\frac{1}{\log (d_0)} \leq 00.1$, and $\frac{\log \left(\log(d_0)\right)}{\log(d_0)} \leq \chi$.

% Clearly, there is such $d_0$ since $\zeta,\chi,\alpha,\beta,C$, and $k_0$ are constants, $\lim_{x \rightarrow \infty} \frac{x}{\log x} = \infty$, and $\lim_{x \rightarrow \infty} \frac{\log \log x}{\log x} = 0$. 
 %a second constant $a = \frac{\zeta \cdot b}{42}+C TBA$.  
%In addition, let $\eps = \frac{\chi}{}$. 
%Assume towards a contradiction that  there is an algorithm $\cA$ that given a \textnormal{VK} instance $\cI$ with $d$ dimensions and $\eps \leq \frac{\chi}{6 \cdot \log(d)}$ returns a $(1-\eps)$-approximate solution for $\cI$ in time $O \left(f(d,\eps) \cdot \big(\left|\cI\right|+W(\cI) \big)^{a \cdot \frac{d}{ \eps \cdot \log^2(d)}} \right)$, %, where $d$ is the number of dimensions in the instance $\cI$, and 
%where $f$ is some computable function. 
Assume towards a contradiction that %for all constants $\zeta, \chi,d_0 > 0$ 
there are an integer $d>d_0 $, an $\eps \in \left(0, \frac{\chi}{\log d}\right)$, and an algorithm $\cA$ that returns a $(1 - \eps)$-approximate solution for every $d$-dimensional knapsack instance in time $O\left(n^{\frac{d}{\eps} \cdot \frac{\zeta}{(\log(d/\eps))^2}}\right)$, where $n$ is the encoding size of the instance.

Define $k =  \floor{\frac{d \cdot \beta}{1000 \cdot \eps \cdot \log \left(\frac{d}{\eps}\right)}}$. 
Observe that 
\begin{equation}
	\label{eq:ChilogD}
	\eps \cdot \log \left(\frac{d}{\eps}\right) = \eps \cdot \log \left(\frac{1}{\eps}\right)+\eps \cdot \log (d) \leq \frac{1}{\log(d)} \cdot \log \left(\frac{1}{\frac{1}{\log(d)}}\right)+ \frac{\chi}{\log (d)} \cdot \log (d)\leq 2 \cdot \chi. 
\end{equation}
%\pasin{I don't think the reason you gave below is correct because the inequality is on the wrong side to get $\log\left(\frac{1}{\eps}\right) \leq \log \log d$.}
The first inequality holds since the function $x \cdot \log \left(\frac{1}{x}\right)$ is increasing in the domain $x \in [0,0.01]$; thus, since $0<\eps \leq \frac{\chi}{\log (d)} \leq  \frac{1}{\log (d)} \leq  \frac{1}{\log (d_0)} \leq 0.01$ , it follows that $	\eps \cdot \log \left(\frac{1}{\eps}\right) \leq \frac{1}{\log(d)} \cdot \log \left(\frac{1}{\frac{1}{\log(d)}}\right)$. Moreover, the second expression follows easily since $\eps \leq \frac{\chi}{\log (d)}$. The second inequality follows since $d \geq d_0$ and since $\frac{\log \left(\log(d_0)\right)}{\log(d_0)} \leq \chi$. 
 %and since the function $\log (x)$ is increasing in the domain $(0,\infty)$. 
Thus, 
%The first inequality holds since $\eps \leq \frac{\chi}{\log (d)} \leq  \frac{1}{\log (d)}$; thus, $\eps \cdot \log \left(\frac{1}{\eps}\right) \leq \log \log (d)$ using the monotonicity of $$. The fourth inequality follows from the fact that $d \geq d_0$ and $\frac{\log \left(\log(d_0)\right)}{\log(d_0)} \leq \chi$ using the monotonicity of $ \frac{\log (x)}{\log \log (x)}$. 
\begin{equation}
	\label{eq:Feps2}
	\begin{aligned}
		k =
		\floor{\frac{d \cdot \beta}{1000 \cdot \eps \cdot \log \left(\frac{d}{\eps}\right)}} 
		\geq
		\frac{d \cdot \beta}{1000 \cdot \eps \cdot \log \left(\frac{d}{\eps}\right)}-1
	%	\\\geq{} &
	%	\frac{d \cdot \beta}{20 \cdot \eps \cdot \log \left(\frac{d}{\eps}\right)} -1
	%	\\\geq{} &
%		\frac{d \cdot \beta}{10 \cdot \left(\eps \cdot \log \left(\frac{1}{\eps}\right)+\eps \cdot \log(d)\right)}-1
	%	\\\geq{} &
%		\frac{d \cdot \beta}{10 \cdot \left(\frac{\log \left(\log(d)\right)}{\log(d)}+\frac{\chi}{\log(d)} \cdot \log(d)\right)} -1
		\geq
		\frac{d \cdot \beta}{1000 \cdot 2 \cdot \chi} -1
		\geq
		2 \cdot d -1
			\geq
		d. 
	\end{aligned}
\end{equation}
%Since $\eps \leq \frac{1}{\log d}$ and $d \geq d_0$ it follows that $k \geq k_0$. 
 The second inequality holds by \eqref{eq:ChilogD}. The third inequality follows from the selection of $\chi$. Therefore, by \eqref{eq:Feps2} it follows that $k \geq d \geq d_0 \geq k_0$. 

We define the following %$\frac{\chi/6}{\log}$-
algorithm $\cB$ for R-CSP on $3$-regular graphs with at most $k$ constraints. Namely, given an R-CSP instance $\Pi$ with at most $k$ variables (vertices) with a $3$-regular constraint graph $H$, Algorithm $\cB$ decides if $\MaxPar(\Pi) = \abs{V(H)}$ or $\MaxPar(\Pi) < \left(1-\frac{\beta}{\log(k)}\right) \cdot \abs{V(H)}$. Let %$\Gamma = (H,\Sigma,X)$ 
$$\Pi = \left(H,\Sigma,\Upsilon, \{\pi_{e, u}, \pi_{e, v}\}_{e = (u,v)\in E(G)}\right)$$
be an R-CSP instance with $\abs{V(H)} \leq k$ vertices such that $H$ is $3$-regular. Define $\cB$ on input $\Pi$ by:  

\begin{enumerate}
%	\item Compute using $\textnormal{2-CSP}  \rightarrow  \textnormal{VK}$ the VK instance 

\item Define $F = \ceil{\frac{24 \cdot k}{d}}$. 

	\item Compute the VK instance $\cI(\Pi,F) = (I, p,c,B)$ by \textnormal{\textsf{2-CSP $\rightarrow$ VK}}. 
	
	\item Execute $\cA$ on instance $\cI(\Pi,F)$. %with parameter $\eps = \frac{b}{\log(d)}$. 
	Let $S$ be the returned solution. 
	
	\item If $p(S) \geq (1-\eps) \cdot \left(\abs{V(H)}+2 \cdot  \abs{E(H)}\right)$: return that  $\MaxPar(\Pi) = \abs{V(H)}$.
	
	\item If $p(S) < (1-\eps) \cdot \left(\abs{V(H)}+2 \cdot  \abs{E(H)}\right)$: return that $\MaxPar(\Pi) < \left(1-\frac{\beta}{\log(k)}\right) \cdot \abs{V(H)}$. 
%	\item Compute $\rightarrow $ R-CSP}} that, given a \textnormal{2-CSP} $$\Pi(\Gamma) = \left(G,\Sigma_{\Pi},\mathcal{X}, \{\pi_{e, u}, \pi_{e, v}\}_{e = (u,v)\in E(G)}\right)$$
\end{enumerate}

For the correctness, we first argue that $\cI(\Pi,F)$ is of a smaller (or equal) dimension than $d$. 
\begin{equation}
	\label{eq:smallerDim}
	\begin{aligned}
	3 \cdot \ceil{    \frac{\abs{V(H)}+\abs{E(H)}}{F} } \leq  3 \cdot \ceil{    \frac{k+3 \cdot k}{F} } \leq 3 \cdot \ceil{    \frac{4 \cdot k}{ \frac{24 \cdot k}{d} } } = 3 \cdot \ceil{\frac{d}{6}} \leq 3 \cdot \left( \frac{d}{6}+1\right) \leq 3 \cdot 2 \cdot \frac{d}{6} = d. 
	\end{aligned}
\end{equation}

The first inequality holds since $\abs{V(H)} \leq k$ and since $H$ is $3$-regular. The second inequality follows from the selection of $F$. The last inequality holds since $d \geq d_0 \geq 6$. Therefore, by \eqref{eq:smallerDim} and \Cref{thm:redA} it holds that the number of dimensions of $\cI(\Pi,F)$ is at most $d$; hence, $\cI(\Pi,F)$ is well defined (recall that instances of a smaller dimension are considered to be also of dimension $d$).  
%Define a function $g: \mathbb{N} \rightarrow \mathbb{N}$ by $g(t) = f(d,\eps)$. 
%We use the following claim on the running tim
	 Consider the following inequality for the running time analysis. 
	 \begin{equation}
	 	\label{eq:dtok}
	 	5 \cdot \zeta \cdot \frac{d}{ \eps \cdot \log^2 \left(\frac{d}{\eps}\right)} = 5 \cdot \zeta \cdot \left(\frac{d \cdot \beta}{\eps \cdot \log \left(\frac{d}{\eps}\right) \cdot 1000}\right)
	 	\cdot \frac{1000}{\beta \cdot \log \left(\frac{d}{\eps}\right) } \leq  5 \cdot \zeta \cdot 2 \cdot k
	 	\cdot \frac{1000}{\beta \cdot \log \left(\frac{d}{\eps}\right) } \leq \frac{\alpha \cdot k}{	 \log \left(k \right)}. 
	 \end{equation}
	% Thus, by \eqref{eq:Iclaim}, \eqref{eq:Wclaim}, and 	\item Compute the VK instance $\cI(\Gamma) = (I,p,c,B)$ by \textnormal{\textsf{2-CSP $\rightarrow$ VK}}. 
	
	The first inequality holds since $	2 \cdot k \geq \frac{d \cdot \beta}{1000 \cdot \eps \cdot \log \left(\frac{d}{\eps}\right)} $ as $k =  \floor{\frac{d \cdot \beta}{1000 \cdot \eps \cdot \log \left(\frac{d}{\eps}\right)}}$ and $k \geq d_0 \geq 6$. The last inequality follows from the selection of $\zeta$ and since $k \leq \frac{d}{\eps}$ (using the monotonicity of the function $\log (x)$). 
	By \Cref{thm:redA}, there is a constant $E>0$ such that 
	$\left|\cI(\Pi,F)\right| \leq E \cdot \abs{\Pi}^4$. If $\abs{\Pi} < E$, then the running time of $\cB$ on the input $\Pi$ is bounded by a constant. Otherwise, assume for the following that $\abs{\Pi} \geq E$ thus $\left|\cI(\Pi,F)\right| \leq \abs{\Pi}^5$. Then, by the running time guarantee of $\cA$ and since the running time of computing $\cI(\Pi,F)$ can be bounded by $O \left(\abs{\Pi}^C\right)$, executing $\cB$ on $\Pi$ can be done in time 

	 \begin{equation*}
	\label{eq:TIMEF}
	\begin{aligned}
		O \left( \left|\cI(\Pi,F)\right|^{\zeta \cdot \frac{d}{ \eps \cdot \log^2 \left(\frac{d}{\eps}\right)}} +\abs{\Pi}^{C}\right) 
		\leq
			O \left( \left|\Pi\right|^{5 \cdot \zeta \cdot \frac{d}{ \eps \cdot \log^2 \left(\frac{d}{\eps}\right)}} +\abs{\Pi}^{C}\right) 
		\leq
			O \left( \abs{\Pi}^{ \frac{\alpha \cdot k}{	 \log \left(k \right)}
			} +\abs{\Pi}^{C}\right)  \leq 
		  	 O \left( \abs{\Pi}^{ \frac{\alpha \cdot k}{	 \log \left(k \right)}
		  	 }
		  	 \right) 
	\end{aligned}
\end{equation*} 

	% The first inequality follows from \Cref{thm:redMain} and since $a \geq C$. The second inequality holds since we assume that $\abs{\Gamma} \geq E$. The third and fourth inequalities hold since %$d = \abs{V(H)}+2 \cdot \abs{E(H)}$ by  \Cref{thm:MainCSPtoG-CSP} and \Cref{thm:redA}; thus, since $H$ is $3$-regular it holds that 
	% $k \leq d \leq 7 \cdot k$ also from \Cref{thm:redMain}. The last equality follows by the selection of $a$. Observe that $\eps$ is determined as a function of $d$, which is a function of $k$. Thus, there is a computable function $g: \mathbb{N} \rightarrow \mathbb{N} $
%such that $f(d,\eps) \leq g(k)$, implying by \eqref{eq:TIMEF} that the running time of computing $\cB$ on $\Gamma$ can be bounded by $ O \left(g(k) \cdot \abs{\Gamma}^{\zeta \cdot \frac{k}{\log(k)}}\right)$. 
%\end{claimproof}

The first inequality holds since $\left|\cI(\Pi,F)\right| \leq \abs{\Pi}^5$. The second inequality follows from \eqref{eq:dtok}. The third inequality holds since $\alpha \cdot \frac{k}{\log(k)} \geq \alpha \cdot \frac{d_0}{\log (d_0)} \geq C$ by the assumption on $d_0$. %$\floor{\frac{d_0 \cdot \beta \cdot \zeta}{10 \cdot \log^2(d_0)}} \geq C$; thus,  
It remains to prove the two directions of the reduction. Observe that
\begin{equation}
	\label{eq:Feps}
	\begin{aligned}
		10 \cdot F \cdot \eps = {} & 
		10 \cdot \ceil{\frac{24 \cdot k}{d}} \cdot \eps 
		%	 \\={} &
		% 2 \cdot \ceil{\frac{k}{d}} \cdot \eps 
	\leq
		10 \cdot 2 \cdot 24 \cdot \frac{k}{d} \cdot \eps 
\leq
		500 \cdot
		\frac{	 \frac{d \cdot \beta}{1000 \cdot \eps \cdot \log \left(\frac{d}{\eps}\right)}}{d}
		\cdot \eps
	\\={} & 
		500 \cdot \frac{\beta}{1000 \cdot \eps \cdot \log \left(\frac{d}{\eps}\right)} \cdot \eps 
	<
		\frac{\beta}{\log \left(\frac{d}{\eps}\right)} 
	\leq
		\frac{\beta}{\log \left(k\right)}
	\end{aligned}
\end{equation}
%\begin{equation}
%	\label{eq:Feps}
%	\begin{aligned}
%		2 \cdot F \cdot \eps ={} &
%		 2 \cdot \ceil{\frac{24 \cdot k}{d}} \cdot \eps 
%	%	 \\={} &
%		 % 2 \cdot \ceil{\frac{k}{d}} \cdot \eps 
%		  \\\leq{} &
%		   2 \cdot 2 \cdot 24 \cdot \frac{k}{d} \cdot \eps 
%		\\\leq{} & 
%		100 \cdot
%	    \frac{	 \frac{d \cdot \beta}{200 \cdot \eps \cdot \log \left(\frac{d}{\eps}\right)}}{d}
%		 \cdot \eps
%		 \\={} &
%		  100 \cdot \frac{\beta}{200 \cdot \eps \cdot \log \left(\frac{d}{\eps}\right)} \cdot \eps 
%		  \\<{} &
%		   \frac{\beta}{\log \left(\frac{d}{\eps}\right)} 
%		   \\\leq{} &
%		 \frac{\beta}{\log \left(k\right)}
%	\end{aligned}
%\end{equation}
%\floor{\frac{d \cdot \beta}{10 \cdot \eps \cdot \log \left(\frac{d}{\eps}\right)}}
The first inequality holds since $k \geq d$ using \eqref{eq:Feps2}. The last inequality holds since $k \leq \frac{d}{\eps}$. 

For the first direction, assume that $\MaxPar(\Pi) = |V(H)|$. Thus, there is a consistent partial assignment for $\Pi$ of size $|V(H)|$. Then, by  %\Cref{thm:MainCSPtoG-CSP} there is a consistent partial assignment to $\Pi(\Gamma)$ of size $k = |E(H)|$. Thus, by 
\Cref{thm:redA} there is a solution for $\cI  \left(\Pi,F\right)$ of profit $\left(\abs{V(H)}+2 \cdot  \abs{E(H)}\right)$. %and by \Cref{obs:dimension} there is also a solution for $\cI_d  \left(\Pi,F\right)$ of profit $\left(\abs{V(H)}+2 \cdot  \abs{E(H)}\right)$. 
Therefore, since $\cA$ returns a $(1-\eps)$-approximate solution for $\cI  \left(\Pi,F\right)$, the returned solution $S$ by $\cA$ has profit at least $(1-\eps) \cdot \left(\abs{V(H)}+2 \cdot  \abs{E(H)}\right)$.
%\begin{equation*}
%	\label{eq:profitK}
%	(1-\eps) \cdot \left(\abs{V(H)}+2 \cdot  \abs{E(H)}\right). %= \left(\abs{V(H)}+2 \cdot  \abs{E(H)}\right)- \eps \cdot \left(\abs{V(H)}+2 \cdot  \abs{E(H)}\right) = \left(\abs{V(H)}+2 \cdot  \abs{E(H)}\right)- \frac{\chi}{6 \cdot \log(k)}
%\end{equation*}
 Thus, $\cB$ correctly decides that $\MaxPar(\Pi) = |V(H)|$. 
 
 Conversely, assume that $\MaxPar(\Pi) < \left(1-\frac{\beta}{\log(k)}\right) \cdot |V(H)|$. Thus, every consistent partial assignment for $\Pi$ is of size strictly less than  $\left(1-\frac{\beta}{\log(k)}\right) \cdot |V(H)|$. 
 Therefore, by the above in conjunction with %\Cref{thm:MainCSPtoG-CSP} there is no consistent partial assignment for $\Pi(\Gamma)$ of size at least $k-\frac{\chi}{6 \cdot \log(k)} \cdot k$. Thus, by 
\Cref{thm:redA}, there is no solution for $\cI  \left(\Pi,F\right)$ of profit at least 
\begin{equation}
	\label{eq:Profit<t}
	\begin{aligned}
		& \left(\abs{V(H)}+2 \cdot  \abs{E(H)}\right) - \frac{\beta}{\log(k) \cdot 2 \cdot F} \cdot \abs{V(H)}
		\\={} & 
		 \left(\abs{V(H)}+2 \cdot  \abs{E(H)}\right) - \frac{\beta}{\log(k) \cdot 2 \cdot F \cdot 5} \cdot 5 \cdot \abs{V(H)}
		 	\\\leq{} & 
		 \left(\abs{V(H)}+2 \cdot  \abs{E(H)}\right) - \frac{\beta}{\log(k) \cdot 2 \cdot F \cdot 5} \cdot 	\left(\abs{V(H)}+2 \cdot  \abs{E(H)}\right) 
		 	\\<{} & 
		 \left(\abs{V(H)}+2 \cdot  \abs{E(H)}\right) \cdot (1-\eps). 
	\end{aligned}
\end{equation} 
The first inequality holds since $H$ is $3$-regular; thus, it holds that $\abs{E(H)} \leq 2 \cdot \abs{V(H)}$ implying $\abs{V(H)}+2 \cdot  \abs{E(H)} \leq 5 \cdot \abs{V(H)}$. The second inequality follows from \eqref{eq:Feps}.
 Thus, by \eqref{eq:Profit<t} there is no solution for $\cI  \left(\Pi,F\right)$ of profit at least
$\left(\abs{V(H)}+2 \cdot  \abs{E(H)}\right) \cdot \left(1-\eps \right)$. %Finally, by \Cref{obs:dimension} there is no solution for $\cI_d  \left(\Pi,F\right)$ of profit at least
%$\left(\abs{V(H)}+2 \cdot  \abs{E(H)}\right) \cdot \left(1-\eps \right)$. 
Hence, the returned solution $S$ by $\cA$ has profit strictly less than $\left(\abs{V(H)}+2 \cdot  \abs{E(H)}\right) \cdot \left(1-\eps \right)$. It follows that $\cB$ returns that $\MaxPar(\Pi) < \left(1-\frac{\beta}{\log(k)}\right) \cdot \abs{V(H)}$ as required. By the above, $\cB$ correctly decides if $\MaxPar(\Pi) = \abs{V(H)}$ or $\MaxPar(\Pi) < \left(1-\frac{\beta}{\log(k)}\right) \cdot \abs{V(H)}$ in time $	 O \left( \abs{\Pi}^{ \frac{\alpha \cdot k}{	 \log \left(k \right)}
}
\right) $. 
This is a contradiction to \Cref{thm:RCSP} and the proof follows. \qed %\qed contradicts \Cref{thm:RCSP}. \qed %is a $\frac{\chi/6}{\log}$-algorithm for 2-CSP that runs in time $O \left(g(k) \cdot \abs{\Gamma}^{\zeta \cdot \frac{k}{\log(k)}}\right)$. This is a contradiction to \Cref{thm:CSP} and the proof follows.  \qed

\hspace{1pt}

\subsection{Proofs based on the reduction from \Cref{sec:simple}}

We need an initial hardness for R-CSP with a larger factor than in \Cref{thm:RCSP}. Specifically, we use the following result from \cite{chalermsook2017gap}, which shows that R-CSP is ``inherently enumerative''; putting it differently, the result states that, to distinguish the two cases, the best algorithm is essentially to straightforwardly enumerate all possible assignments to $r$ variables (which runs in time $|\Gamma|^{O(r)}$).

\begin{theorem}[\cite{chalermsook2017gap}] \label{thm:inherent-enumerate}
	Assuming \textnormal{Gap-ETH}, there exist constants $\zeta > 0$ and $r_0 \in \N$, such that, for any constants $k \geq r \geq r_0$, there is no algorithm that takes in an \textnormal{R-CSP} instance $\Pi$ with a constraint graph $H$ with $|V(H)| = k$, runs in time $O\left(|\Pi|^{\zeta \cdot r}\right)$ and distinguish between:
	\begin{itemize}
		\item {\bf (Completeness)} $\MaxPar(\Pi) = k$, and,
		\item {\bf (Soundness)} $\MaxPar(\Pi) < r$.
	\end{itemize}
\end{theorem}

Since the exact statement in \cite{chalermsook2017gap} is slightly different than ours, we provide a proof of \Cref{thm:inherent-enumerate} in \Cref{sec:inherent-enumerate} for completeness. \Cref{thm:sqrt-runningtime-lb} and \Cref{thm:sqrt-ratio}  now follow directly from our reduction (\Cref{thm:red-large-gap}) and \Cref{thm:inherent-enumerate} by setting appropriate parameters.

%\begin{theorem} \label{thm:sqrt-runningtime-lb}
%Assuming \textnormal{Gap-ETH}, for any constant $\rho \in (0,1)$, there exist constants $\delta > 0$ and $d_0 \in \N$ such that the following holds: for any constant $d \geq d_0$, there is no $\rho$-approximation algorithm for $d$-dimensional knapsack that runs in $O(n^{\delta \sqrt{d}})$ time.
%\end{theorem}

\begin{proof}[Proof of \Cref{thm:sqrt-runningtime-lb}]
	Assume that Gap-ETH holds. Let $\zeta, r_0$ be the promised constants from \Cref{thm:inherent-enumerate}. Since $r_0$ can be chosen to be arbitrarily large, assume that $\zeta \cdot r_0 \geq 3$. Suppose that there is $\rho \in (0,1)$ and a $\rho$-approximation algorithm $\cA$ for $d$-dimensional knapsack that runs in $O\left(n^{\delta \sqrt{d}}\right)$ time, where $n$ is the encoding size of the instance, $\delta = \frac{\zeta \cdot \rho}{20}$, $d_0 = \ceil{\left( \frac{8 \cdot r_0}{\rho}\right)^2}$, and $d \geq d_0$. We can use $\cA$ to distinguish the two cases in \Cref{thm:inherent-enumerate} for $k =  \ceil{\sqrt{d}} $ and $r = \lfloor k \cdot \rho \rfloor$ as follows. Let $\Pi$ be an R-CSP instance with $k = \abs{V(H)}$ vertices, where $H$ is the constraint graph of $\Pi$. Define an algorithm $\cB$ on input $\Pi$ as follows.
	\begin{enumerate}
		\item Compute $\cI(\Pi) = (I,p,w,B)$ using the reduction from \Cref{thm:red-large-gap}. 
		
		\item  Execute $\cA$ on $\cI(\Pi)$. Let $S$ be the returned solution. 
		
		\item Return that $\MaxPar(\Pi) = k$ if and only if $p(S) \geq r$. 
	\end{enumerate} 
	%First apply the reduction from \Cref{thm:red-large-gap}, and then run $\cA$ on the resulting instance $\cI(\Pi)$. 
	Observe that $\delta \cdot \sqrt{d} \leq \frac{\zeta \cdot \rho}{20} \cdot k \leq \frac{\zeta \cdot r }{5}$ and recall that $\abs{\cI(\Pi)} \leq O\left(\abs{\Pi}^3\right)$ and that $\cI(\Pi)$ can be computed from $\Pi$ in time $O\left(\abs{\Pi}^3\right)$ by \Cref{thm:red-large-gap}. Therefore, algorithm $\cB$ runs in time $O\left(|\Pi|^{\zeta \cdot r}\right)$. 
	If $\MaxPar(\Pi) = k$ then by \Cref{thm:red-large-gap} there is a solution for $\cI(\Pi)$ of profit $k$. Therefore, since $\cA$ is a $\rho$-approximation algorithm it follows that the profit of $S$ is at least $\rho \cdot k \geq r$; conversely, if $\MaxPar(\Pi) < r$, then by \Cref{thm:red-large-gap} the profit of $S$ is strictly less than $r$.
	%Furthermore, notice that the gap between two cases is $k/r \geq \rho$.
	%Therefore, since $\cA$ is a $\rho$-approximation algorithm it follows that 
	Hence, $\cB$ correctly decides if $\MaxPar(\Pi) = k$ or $\MaxPar(\Pi) < r$ in time $O\left(|\Pi|^{\zeta \cdot r}\right)$. This violates Gap-ETH by \Cref{thm:inherent-enumerate}. 
	%Thus, the algorithm correctly distinguishes the two cases, which violates Gap-ETH.
\end{proof}

The proof of the next theorem is similar to the above proof with different selection of parameters.  

\begin{proof}[Proof of \Cref{thm:sqrt-ratio}]
	Assume that Gap-ETH holds. Let $\zeta, r_0$ be the promised constants from \Cref{thm:inherent-enumerate}. Assume towards a contradiction that there is $C \geq 1$ and $d \geq d_0$ such that there is a $\left(\frac{\rho}{\sqrt{d}}\right)$-approximation algorithm $\cA$ for $d$-dimensional knapsack that runs in $O(n^C)$ time, where $n$ is the encoding size of the instance, $r = \max\{r_0, 3C/\zeta\}, \rho = 2 \cdot r $, and $d_0 = 4 \cdot \rho^2$. We use $\cA$ to construct an algorithm $\cB$ that distinguishes the two cases in \Cref{thm:inherent-enumerate} with value $k = \ceil{\sqrt{d}}$. Let $\Pi$ be an R-CSP instance with $k = \abs{V(H)}$ vertices, where $H$ is the constraint graph of $\Pi$. Define algorithm $\cB$ on input $\Pi$ as follows.
	\begin{enumerate}
		\item Compute $\cI(\Pi) = (I,p,w,B)$ using the reduction from \Cref{thm:red-large-gap}. 
		
		\item  Execute $\cA$ on $\cI(\Pi)$. Let $S$ be the returned solution. 
		
		\item Return that $\MaxPar(\Pi) = k$ if and only if $p(S) \geq 2 \cdot r$.  
	\end{enumerate} 
	
	%by first applying the reduction on with $k = \lfloor \sqrt{d} \rfloor$, and then run $\cA$ on the resulting instance $\cI(\Pi)$. 
	Observe that $3 \cdot C =  \frac{3 \cdot \zeta \cdot C}{\zeta} \leq \zeta \cdot r$. In addition, recall that $\abs{\cI(\Pi)} \leq O\left(\abs{\Pi}^3\right)$ and that $\cI(\Pi)$ can be computed from $\Pi$ in time $O\left(\abs{\Pi}^3\right)$ by \Cref{thm:red-large-gap}. Therefore, algorithm $\cB$ runs in time $O\left(|\Pi|^{\zeta \cdot r}\right)$. If $\MaxPar(\Pi) = k$ then by \Cref{thm:red-large-gap} there is a solution for $\cI(\Pi)$ of profit $k$. Thus, since $\cA$ is a $\left(\frac{\rho}{\sqrt{d}}\right)$-approximation algorithm it follows that the profit of $S$ is at least $\frac{k \cdot \rho}{\sqrt{d}} \geq \rho  = 2 \cdot r$; conversely, if $\MaxPar(\Pi) < r$, then by \Cref{thm:red-large-gap} the profit of $S$ is strictly less than $r$.
	%Therefore, since $\cA$ is a $\left(\frac{\rho}{\sqrt{d}}\right)$-approximation algorithm it follows that 
	Hence, in both cases $\cB$ decides correctly if $\MaxPar(\Pi) = k$ or $\MaxPar(\Pi) < r$ in time $O\left(|\Pi|^{\zeta \cdot r}\right)$. This violates Gap-ETH by \Cref{thm:inherent-enumerate}. 
	%	 It is simple to verify that the algorithm runs in time $O\left(|\Gamma|^{\zeta \cdot r}\right)$. Furthermore, notice that the gap between two cases is $k/r \geq \rho \sqrt{d}$. Thus, the algorithm correctly distinguishes the two cases. By \Cref{thm:inherent-enumerate}, this violates Gap-ETH.
\end{proof}

\ifapprox
\section{$\Omega\left(\frac{1}{\sqrt{d}}\right)$-Approximation Algorithm}
\label{sec:algorithm}
In this section, we provide our $\Omega\left(\frac{1}{\sqrt{d}}\right)$-approximation algorithm and prove \Cref{thm:main-apx}.
To do this, let us first introduce a concept of bounded/unbounded instances.
For $\tau > 1$, we say that an instance is \emph{$\tau$-bounded} if we have $c(i)_j \leq B_j / \tau$ for all $i \in I$ and $j \in [d]$. Furthermore, an instance is \emph{$\tau$-unbounded} if, for every $i \in I$, there exists $j \in [d]$ such that $c(i)_j > B_j / \tau$. We note that (due to the quantifiers) there are instances that are neither $\tau$-bounded nor $\tau$-unbounded. Nevertheless, we will show a (simple) reduction that splits the instance into bounded and unbounded parts.

Bounded instances have been known to be easier to approximate.
When $\tau$ is at least $2$, a $\Omega\left(\frac{1}{\sqrt{d}}\right)$-approximation algorithm is already known for $\tau$-bounded instances from previous work, based on the randomized rounding framework of Raghavan and Thompson~\cite{RaghavanT87}. A derandomized (and slightly improved) version of this is given in \cite{Srinivasan95}, which we state below.

\begin{theorem}[\cite{Srinivasan95}] \label{thm:randomized-rounding}
There is a polynomial-time $\Omega\left(\frac{1}{\sqrt{d}}\right)$-approximation algorithm $\cA_{\LP}$ for $d$-dimensional knapsack on 2-bounded instances.
\end{theorem}

Our main contribution in this section is to give a $\Omega\left(\frac{1}{\sqrt{d}}\right)$-approximation algorithm on \emph{$2$-unbounded instances}, which runs in $\left((d \cdot \log W)^{O(d^2)} + n^{O(1)}\right)$ time.

\begin{lemma} \label{lem:unbounded-fpt-approx}
There is a $\left((d \cdot \log W)^{O(d^2)} + n^{O(1)}\right)$-time $\Omega\left(\frac{1}{\sqrt{d}}\right)$-approximation algorithm $\cA_{\unbounded}$ for $d$-dimensional knapsack on 2-unbounded instances.
\end{lemma}

Using the above two results, it is now easy to obtain \Cref{thm:main-apx}. Throughout this section, recall that $\OPT(\cI)$ denotes the optimum profit among all solutions for the instance $\cI$.

\begin{proof}[Proof of \Cref{thm:main-apx}]
On input $\cI = (I,p,c,B)$, the algorithm works as follows:
\begin{itemize}
\item Partition $I$ into $I_{\bounded} \cup I_{\unbounded}$ where
$$I_{\bounded} = \{i \in I \mid \forall j \in [d], c(i)_j \leq B_j/2\},$$
and
$$I_{\unbounded} = \{i \in I \mid \exists j \in [d], c(i)_j > B_j/2\}.$$
\item Run $\cA_{\LP}$ from \Cref{thm:randomized-rounding} on $\cI_{\bounded} = (I_{\bounded},p,d,c,B)$ to get a solution $S_{\bounded}$
\item Run $\cA_{\unbounded}$ from \Cref{lem:unbounded-fpt-approx} on $\cI_{\unbounded} = (I_{\unbounded},p,d,c,B)$ to get a solution $S_{\unbounded}$.
\item Output the best solution $S$ among the two.
\end{itemize}

The running time claim is obvious from \Cref{thm:randomized-rounding} and \Cref{lem:unbounded-fpt-approx}. As for the approximation guarantee, note that
\begin{align*}
p(S) &= \max\{p(S_{\bounded}), p(S_{\unbounded})\} \\
&\geq \frac{1}{2} \cdot \left(p(S_{\bounded}) + p(S_{\unbounded})\right) \\
&\geq \Omega(1/\sqrt{d}) \cdot \left(\OPT(\cI_{\bounded}) + \OPT(\cI_{\unbounded}))\right) \\
&\geq \Omega(1/\sqrt{d}) \cdot \OPT(\cI),
\end{align*}
where the second inequality is from the approximation guarantees in \Cref{thm:randomized-rounding} and \Cref{lem:unbounded-fpt-approx}. This completes our proof.
\end{proof}

\subsection{Algorithm for 2-Unbounded Instances: Proof of \Cref{lem:unbounded-fpt-approx}}

We now give our algorithm for 2-unbounded instances and prove its guarantees (\Cref{lem:unbounded-fpt-approx}).
On an input 2-unbounded $d$-dimensional knapsack instance $\cI = (I,p,c,B)$, the algorithm works as follows.
\begin{itemize}
\item Let $\gamma := 1 + \frac{0.1}{d}$.
\item Let $\varpiup, \varpidown: \{0, \dots, B\} \to [0, B]$ be the following discretization functions\footnote{Since we only use the discretization to apply the reduction rule in the next step, it suffices to just represent the exponent ${\left\lfloor \log_\gamma(x) \right\rfloor}$ or ${\left\lceil \log_\gamma(x) \right\rceil}$ which is an integer instead of the real number $\gamma^{\left\lfloor \log_\gamma(x) \right\rfloor}$ or $\gamma^{\left\lceil \log_\gamma(x) \right\rceil}$.}:
\begin{align*}
\varpidown(x) :=
\begin{cases}
0 & \text{ if } x = 0 \\
\gamma^{\left\lfloor \log_\gamma(x) \right\rfloor} & \text{ otherwise,}
\end{cases}
& & \text{ and } & & \varpiup(x) :=
\begin{cases}
0 & \text{ if } x = 0 \\
\gamma^{\left\lceil \log_\gamma(x) \right\rceil} & \text{ otherwise.}
\end{cases}
\end{align*}
Then, let $\digamma: \{0, \dots, B\}^d \to [0, B]^d$ be the following discretization function:
\begin{align*}
\digamma(x) :=
\min\{\varpiup(x_j), B_j - \varpidown(B_j - x_j)\} 
& &\forall j \in [d].
\end{align*}
\item Apply the following reduction rule until it cannot be applied: If there exists two items $i, i' \in I$ such that $\digamma(c_i) = \digamma(c'_i)$, remove the item with smaller profit.
\item Use the bruteforce algorithm, where we enumerate all subsets of size at most $d$, to find the best solution among the remaining items.
\end{itemize}

Let us first note that the possible different values of $\digamma(c)$ is at most $(\log_\gamma(B))^d \leq O(d \cdot \log W)^d$. Thus, after applying the reduction rule, there are at most $O(d \cdot \log W)^d$ items left. As a result, the bruteforce algorithm in the last step takes at most $(d \cdot \log W)^{O(d^2)}$ time. Thus, in total, the running time is at most $(d \cdot \log W)^{O(d^2)} + n^{O(1)}$ as desired.

The next lemma is the main ingredient for our approximation guarantee proof. It shows that our discretization scheme decreases the optimum by a factor of at most $O(\sqrt{d})$.

\begin{lemma} \label{lem:discretize}
For any 2-unbounded instance $\cI$, there is a solution $S$ such that the following holds:
\begin{enumerate}[(i)]
\item $|S| \leq d$. \label{item:solution-size-bound-apx}
\item $p(S) \geq \Omega\left(\OPT(\cI) / \sqrt{d}\right)$. \label{item:profit-apx} 
\item $\sum_{i \in S} \digamma(c(i))_j \leq B_j$ for all $j \in [d]$. \label{item:discretize-apx}
\end{enumerate}
\end{lemma}

Before we prove \Cref{lem:discretize}, let us see how to finish the approximation guarantee. Let $S = \{i_1, \dots, i_m\}$ for $m \leq d$ be the solution as guaranteed in \Cref{lem:discretize}, and let $I'$ denote the set of items that remains after the reduction rule. From the reduction rule, there must exist (distinct) items $i'_1, \dots, i'_m \in I$ such that $p(i'_q) \geq p(i_q)$ and $\digamma(c(i_q)) = \digamma(p(i'_q))$ for all $q \in [m]$. The latter implies that
\begin{align*}
c(i'_1)_j + \cdots + c(i'_m)_j \leq \digamma(c(i'_1))_j + \cdots + \digamma(c(i'_m))_j = \sum_{i \in S} \digamma(c(i))_j \leq B_j & &\forall j \in [d],
\end{align*}
where the last inequality is due to the third item of \Cref{lem:discretize}. This means that $\{i'_1, \dots, i'_m\}$ is a feasible solution. Since we use bruteforce in the last step and $m \leq d$, we must output a solution of profit at least
\begin{align*}
p(\{i'_1, \dots, i'_m\}) \geq p(S) \geq \Omega\left(\OPT(\cI) / \sqrt{d}\right)
\end{align*}
where the last inequality is due to the second item of \Cref{lem:discretize}. Thus, the algorithm achieves $\Omega\left(\frac{1}{\sqrt{d}}\right)$-approximation as claimed, which completes the proof of \Cref{lem:unbounded-fpt-approx}.  \qed

Finally, we prove \Cref{lem:discretize}.

\begin{proof}[Proof of \Cref{lem:discretize}]
We assume w.l.o.g. that each item's cost is within the budget (otherwise we can simply discard it). Let $p_{\max} = \max_{i \in I} p(i)$. Consider two cases based on whether $\OPT(\cI) \leq 10\sqrt{d} \cdot p_{\max}$. If $\OPT(\cI) \leq 100\sqrt{d} \cdot p_{\max}$, then the solution $S$ that simply picks just the item with maximum profit already satisfies all the three constraints.

The remainder of the proof is dedicated to the case $\OPT(\cI) > 10\sqrt{d} \cdot p_{\max}$. In this case, let $S^*$ denote the optimum solution of $\cI$. Note that, since this is a 2-unbounded instance, we have\footnote{This is because in each coordinate $j \in [d]$, there can be at most one item $i \in S^*$ such that $c(i)_j > B_j / 2$.} $|S^*| \leq d$. Let $S$ denote a random subset of $S^*$ where each element is included independently with probability $\theta := 0.5/\sqrt{d}$. We will show that with positive probability $S$ satisfies all the three properties. First, note that $|S| \leq |S^*| \leq d$ always. Thus, it suffices to consider the remaining two items. In the following, we will show that
\begin{align} \label{eq:profit-violation-prob}
\Pr\left[p(S) < \frac{\OPT(\cI)}{4\sqrt{d}}\right] \leq 0.4 
\end{align}
and, for all $j \in [d]$,
\begin{align} \label{eq:constraint-violation-prob}
\Pr\left[\sum_{i \in S} \digamma(c(i))_j > B_j\right] \leq \frac{0.25}{d}
\end{align}
Taking the union bound over these ($d + 1$) events then implies that $S$ satisfies all the three items with positive probability as desired.

For \eqref{eq:profit-violation-prob}, note that $\E[p(S)] = \theta \cdot \OPT(\cI) = \frac{\OPT(\cI)}{2 \sqrt{d}}$. Meanwhile, $$\Var(p(S)) = \sum_{i \in S} \theta(1 - \theta) \cdot p(i)^2  \leq \theta \cdot p_{\max} \cdot \OPT(\cI) \leq \frac{(\OPT(\cI))^2}{200d},$$
where the last inequality follows from the assumption of this case. Applying the Chebyshev's inequality then yields \eqref{eq:profit-violation-prob}.

Next, we will prove \eqref{eq:constraint-violation-prob}. To do this, consider a fixed $j \in [d]$ and let $i_1, \dots, i_m$ be the elements of $S^*$ sorted in descending order by $c(i)_j$, tie broken arbitrarily. We first prove the following claim.

\begin{claim}
If $\sum_{i \in S} \digamma(c(i))_j > B_j$, then both $i_1, i_2$ must belong to $S$.
\end{claim}

\begin{claimproof}
Suppose contrapostively that either $i_1$ or $i_2$ do not belong to $S$. Then, we have
\begin{align*}
\sum_{i \in S} \digamma(c(i))_j &\leq \digamma(c(i_1))_j + \digamma(c(i_3))_j + \cdots + \digamma(c(i_m))_j.
\end{align*}
Now, by our choice of $\digamma$, we have $\digamma(x)_j \leq \gamma \cdot x$ and $\digamma(x)_j \leq B_j - \frac{1}{\gamma} \cdot (B_j - x)$ for all $x \in \{0, \dots, B_j\}$. Plugging this into the above, we get
\begin{align*}
\sum_{i \in S} \digamma(c(i))_j 
&\leq B_j - \frac{1}{\gamma} \cdot (B_j - c(i_1)_j) + \gamma \cdot (c(i_3)_j + \cdots + c(i_m)_j) \\
&\leq B_j - \frac{1}{\gamma} \cdot (c(i_2)_j + \cdots c(i_m)_j) + \gamma \cdot (c(i_3)_j + \cdots + c(i_m)_j) \\
&\leq B_j - \frac{1 + \frac{1}{m-2}}{\gamma} \cdot (c(i_3)_j + \cdots c(i_m)_j) + \gamma \cdot (c(i_3)_j + \cdots + c(i_m)_j) \\
&\leq B_j - \frac{1 + \frac{1}{d}}{\gamma} \cdot (c(i_3)_j + \cdots c(i_m)_j) + \gamma \cdot (c(i_3)_j + \cdots + c(i_m)_j) & \leq B_j.
\end{align*}
where the second inequality follows from the fact that $S^*$ is a feasible solution, the third follows from $c(i_2)_j \geq c(i_3)_j, \dots, c(i_m)_j$ and the last follows from our choice of $\gamma$.
\end{claimproof}

By the above claim, we have
\begin{align*}
\Pr\left[\sum_{i \in S} \digamma(c(i))_j > B_j\right] \leq \Pr\left[i_1 \in S \wedge i_2 \in S\right] = \theta^2 \leq \frac{0.25}{d},
\end{align*}
proving \eqref{eq:constraint-violation-prob}. This completes our proof.
\end{proof}
\fi

\section{Discussion and Open Questions}
\label{sec:discussion}
In this work, we prove several hardness results for the $d$-dimensional knapsack problem, via reductions from 2-CSP. In particular, our main result implies that the PTASes that have been known for decades~\cite{ChandraHW76,FRIEZE1984100,CapraraKPP00} cannot be improved up to a polylogarithmic factor (assuming Gap-ETH). We also show that the best-known exact algorithm of running time $O \left(n \cdot W^d \right)$ is the best possible up to a logarithmic factor (assuming ETH).  

An obvious open question is to close the quantitative gaps in our main theorem (\Cref{thm:main}) compared to the aforementioned PTASes. Namely, can we prove a similar hardness for $\eps$ that is an absolute constant (independent of $d$)? And can we improve the running time lower bound to $n^{\Omega(d / \eps)}$? These are closely related to similar questions for 2-CSP, which are themselves important in the quest to obtain a more complete understanding of parameterized (in)approximability.

\ifapprox
It is also intriguing to see whether the running time in our approximation algorithm (\Cref{thm:main-apx}) can be improved. As mentioned earlier, the best polynomial-time algorithm only achieves $\Omega\left(\frac{1}{d}\right)$-approximation~\cite{Srinivasan95,CapraraKPP00}. Is there a polynomial-time $\Omega\left(\frac{1}{\sqrt{d}}\right)$-approximation algorithm? An intermediate goal here would be to remove the $(\log W)^{O(d^2)}$ term in the running time of our approximation algorithm; this would yield an $\Omega\left(\frac{1}{\sqrt{d}}\right)$-approximation FPT (in $d$) algorithm for the problem.
\fi

Both $2$-CSP and R-CSP is closely related to the (parameterized) maximum clique problem.  There have been numerous developments in parameterized inapproximability of clique in recent years; it is now known that these results can be obtained under ETH  or even $\textnormal{W}[1]\neq \textnormal{FPT}$ (instead of Gap-ETH)~\cite{Lin21,LinRSW22,SK22,LinRSW23,chen2023simple}. However, the running time lower bounds here are still too weak (e.g $n^{\Omega(\log \log k)}$~\cite{LinRSW23}) to give strong lower bounds for $d$-dimensional knapsack. It remains an interesting question whether we can get near-tight running time lower bounds for approximating $d$-dimensional knapsack using these weaker assumptions.

\iffalse
\paragraph{Acknowledgment.} Pasin is grateful to Karthik C. S. for insightful discussions on \cite{karthik2023conditional}.
\fi

%In the following, we give a lower 
\bibliographystyle{alpha}
%\newpage
\bibliography{bibfile}

\newcommand{\etalchar}[1]{$^{#1}$}
\begin{thebibliography}{CMWW19}

\bibitem[ARF14]{AzadRF14}
Md. Abul~Kalam Azad, Ana Maria A.~C. Rocha, and Edite M. G.~P. Fernandes.
\newblock Improved binary artificial fish swarm algorithm for the 0-1
  multidimensional knapsack problems.
\newblock {\em Swarm Evol. Comput.}, 14:66--75, 2014.

\bibitem[ASG04]{alaya2004ant}
In{\`e}s Alaya, Christine Solnon, and Khaled GH{\'e}DIRA.
\newblock Ant algorithm for the multidimensional knapsack problem.
\newblock In {\em BIOMA}, pages 63--72, 2004.

\bibitem[BEEB09]{boyer2009heuristics}
Vincent Boyer, Moussa Elkihel, and Didier El~Baz.
\newblock Heuristics for the 0--1 multidimensional knapsack problem.
\newblock {\em European Journal of Operational Research}, 199(3):658--664,
  2009.

\bibitem[BKMN15]{bringmann2015hitting}
Karl Bringmann, L{\'a}szl{\'o} Kozma, Shay Moran, and NS~Narayanaswamy.
\newblock Hitting set for hypergraphs of low vc-dimension.
\newblock {\em arXiv preprint arXiv:1512.00481}, 2015.

\bibitem[BM20]{bonnet2020parameterized}
{\'E}douard Bonnet and Tillmann Miltzow.
\newblock Parameterized hardness of art gallery problems.
\newblock {\em {ACM} Trans. Algorithms}, 16(4):1--23, 2020.

\bibitem[BS17]{bonnet2017graph}
{\'E}douard Bonnet and Florian Sikora.
\newblock The graph motif problem parameterized by the structure of the input
  graph.
\newblock {\em Discrete Applied Mathematics}, 231:78--94, 2017.

\bibitem[BYFA08]{balev2008dynamic}
Stefan Balev, Nicola Yanev, Arnaud Fr{\'e}ville, and Rumen Andonov.
\newblock A dynamic programming based reduction procedure for the
  multidimensional 0--1 knapsack problem.
\newblock {\em European journal of operational research}, 186(1):63--76, 2008.

\bibitem[CCK{\etalchar{+}}17]{chalermsook2017gap}
Parinya Chalermsook, Marek Cygan, Guy Kortsarz, Bundit Laekhanukit, Pasin
  Manurangsi, Danupon Nanongkai, and Luca Trevisan.
\newblock From gap-eth to fpt-inapproximability: Clique, dominating set, and
  more.
\newblock In {\em FOCS}, pages 743--754, 2017.

\bibitem[CDM17]{curticapean2017homomorphisms}
Radu Curticapean, Holger Dell, and D{\'a}niel Marx.
\newblock Homomorphisms are a good basis for counting small subgraphs.
\newblock In {\em STOC}, pages 210--223, 2017.

\bibitem[CFLL23]{chen2023simple}
Yijia Chen, Yi~Feng, Bundit Laekhanukit, and Yanlin Liu.
\newblock Simple combinatorial construction of the $ k^{o (1)}$ -lower bound
  for approximating the parameterized $ k $-clique.
\newblock {\em arXiv preprint arXiv:2304.07516}, 2023.

\bibitem[CFM21]{chitnis2021parameterized}
Rajesh Chitnis, Andreas~Emil Feldmann, and Pasin Manurangsi.
\newblock Parameterized approximation algorithms for bidirected steiner network
  problems.
\newblock {\em {ACM} Trans. Algorithms}, 17(2):1--68, 2021.

\bibitem[Cha18]{Chan18a}
Timothy~M. Chan.
\newblock Approximation schemes for 0-1 knapsack.
\newblock In {\em SOSA}, pages 5:1--5:12, 2018.

\bibitem[CHW76]{ChandraHW76}
Ashok~K. Chandra, Daniel~S. Hirschberg, and C.~K. Wong.
\newblock Approximate algorithms for some generalized knapsack problems.
\newblock {\em Theor. Comput. Sci.}, 3(3):293--304, 1976.

\bibitem[CK04]{ChekuriK2004}
Chandra Chekuri and Sanjeev Khanna.
\newblock On multidimensional packing problems.
\newblock {\em SIAM journal on computing}, 33(4):837--851, 2004.

\bibitem[CKPP00]{CapraraKPP00}
Alberto Caprara, Hans Kellerer, Ulrich Pferschy, and David Pisinger.
\newblock Approximation algorithms for knapsack problems with cardinality
  constraints.
\newblock {\em Eur. J. Oper. Res.}, 123(2):333--345, 2000.

\bibitem[CLMZ23]{CLMZ23}
Lin Chen, Jiayi Lian, Yuchen Mao, and Guochuan Zhang.
\newblock A nearly quadratic-time {FPTAS} for knapsack.
\newblock {\em CoRR}, abs/2308.07821, 2023.

\bibitem[CLRS09]{CLRS}
Thomas~H. Cormen, Charles~E. Leiserson, Ronald~L. Rivest, and Clifford Stein.
\newblock {\em Introduction to Algorithms, 3rd Edition}.
\newblock {MIT} Press, 2009.

\bibitem[CMPS23]{karthik2023conditional}
Karthik {C. S.}, D{\'a}niel Marx, Marcin Pilipczuk, and U{\'e}verton Souza.
\newblock Conditional lower bounds for sparse parameterized 2-csp: A
  streamlined proof.
\newblock {\em arXiv e-prints}, pages arXiv--2311, 2023.

\bibitem[CMWW19]{CyganMWW19}
Marek Cygan, Marcin Mucha, Karol Wegrzycki, and Michal Wlodarczyk.
\newblock On problems equivalent to (min, +)-convolution.
\newblock {\em {ACM} Trans. Algorithms}, 15(1):14:1--14:25, 2019.

\bibitem[CX15]{curticapean2015parameterizing}
Radu Curticapean and Mingji Xia.
\newblock Parameterizing the permanent: Genus, apices, minors, evaluation mod
  2k.
\newblock In {\em FOCS}, pages 994--1009, 2015.

\bibitem[Din16]{dinur2016mildly}
Irit Dinur.
\newblock Mildly exponential reduction from gap-3sat to polynomial-gap
  label-cover.
\newblock In {\em ECCC}, 2016.

\bibitem[DJM23]{DengJM23}
Mingyang Deng, Ce~Jin, and Xiao Mao.
\newblock Approximating knapsack and partition via dense subset sums.
\newblock In {\em SODA}, pages 2961--2979, 2023.

\bibitem[DM18]{DinurM18}
Irit Dinur and Pasin Manurangsi.
\newblock Eth-hardness of approximating 2-csps and directed steiner network.
\newblock In {\em ITCS}, pages 36:1--36:20, 2018.

\bibitem[FC84]{FRIEZE1984100}
A.M. Frieze and M.R.B. Clarke.
\newblock Approximation algorithms for the m-dimensional 0–1 knapsack
  problem: Worst-case and probabilistic analyses.
\newblock {\em European Journal of Operational Research}, 15(1):100--109, 1984.

\bibitem[Fr{\'{e}}04]{Freville04}
Arnaud Fr{\'{e}}ville.
\newblock The multidimensional 0-1 knapsack problem: An overview.
\newblock {\em Eur. J. Oper. Res.}, 155(1):1--21, 2004.

\bibitem[GRS23]{baby-pih}
Venkatesan Guruswami, Xuandi Ren, and Sai Sandeep.
\newblock Baby {PIH:} parameterized inapproximability of min {CSP}.
\newblock {\em CoRR}, abs/2310.16344, 2023.

\bibitem[IK75]{IbarraK75}
Oscar~H. Ibarra and Chul~E. Kim.
\newblock Fast approximation algorithms for the knapsack and sum of subset
  problems.
\newblock {\em J. {ACM}}, 22(4):463--468, 1975.

\bibitem[IP01]{impagliazzo2001complexity}
Russell Impagliazzo and Ramamohan Paturi.
\newblock On the complexity of k-sat.
\newblock {\em Journal of Computer and System Sciences}, 62(2):367--375, 2001.

\bibitem[Jin19]{Jin19}
Ce~Jin.
\newblock An improved {FPTAS} for 0-1 knapsack.
\newblock In {\em ICALP}, pages 76:1--76:14, 2019.

\bibitem[JKMS13]{jansen2013bin}
Klaus Jansen, Stefan Kratsch, D{\'a}niel Marx, and Ildik{\'o} Schlotter.
\newblock Bin packing with fixed number of bins revisited.
\newblock {\em Journal of Computer and System Sciences}, 79(1):39--49, 2013.

\bibitem[JLL16]{JansenLL16}
Klaus Jansen, Felix Land, and Kati Land.
\newblock Bounding the running time of algorithms for scheduling and packing
  problems.
\newblock {\em {SIAM} J. Discret. Math.}, 30(1):343--366, 2016.

\bibitem[Joh73]{Johnson73}
David~S. Johnson.
\newblock Approximation algorithms for combinatorial problems.
\newblock In {\em STOC}, pages 38--49, 1973.

\bibitem[Kar72]{Karp72}
Richard~M. Karp.
\newblock Reducibility among combinatorial problems.
\newblock In {\em Proceedings of a symposium on the Complexity of Computer
  Computations}, pages 85--103, 1972.

\bibitem[KFRW10]{ke2010ant}
Liangjun Ke, Zuren Feng, Zhigang Ren, and Xiaoliang Wei.
\newblock An ant colony optimization approach for the multidimensional knapsack
  problem.
\newblock {\em Journal of Heuristics}, 16:65--83, 2010.

\bibitem[KK22]{SK22}
{Karthik {C. S.}} and Subhash Khot.
\newblock Almost polynomial factor inapproximability for parameterized
  k-clique.
\newblock In {\em CCC}, pages 6:1--6:21, 2022.

\bibitem[KP99]{KellererP99}
Hans Kellerer and Ulrich Pferschy.
\newblock A new fully polynomial time approximation scheme for the knapsack
  problem.
\newblock {\em J. Comb. Optim.}, 3(1):59--71, 1999.

\bibitem[KP04]{KellererP04}
Hans Kellerer and Ulrich Pferschy.
\newblock Improved dynamic programming in connection with an {FPTAS} for the
  knapsack problem.
\newblock {\em J. Comb. Optim.}, 8(1):5--11, 2004.

\bibitem[KPP04]{kellerer2004multidimensional}
Hans Kellerer, Ulrich Pferschy, and David Pisinger.
\newblock {\em Multidimensional Knapsack Problems}.
\newblock 2004.

\bibitem[KPS17]{KunnemannPS17}
Marvin K{\"{u}}nnemann, Ramamohan Paturi, and Stefan Schneider.
\newblock On the fine-grained complexity of one-dimensional dynamic
  programming.
\newblock In {\em ICALP}, pages 21:1--21:15, 2017.

\bibitem[KS81]{KORTE1981415}
Bernhard Korte and Rainer Schrader.
\newblock On the existence of fast approximation schemes.
\newblock In {\em Nonlinear Programming 4}, pages 415--437. Academic Press,
  1981.

\bibitem[KS10]{KulikS10}
Ariel Kulik and Hadas Shachnai.
\newblock There is no {EPTAS} for two-dimensional knapsack.
\newblock {\em Inf. Process. Lett.}, 110(16):707--710, 2010.

\bibitem[Law79]{Lawler79}
Eugene~L. Lawler.
\newblock Fast approximation algorithms for knapsack problems.
\newblock {\em Math. Oper. Res.}, 4(4):339--356, 1979.

\bibitem[LB12]{6297286}
Soh-Yee Lee and Yoon-Teck Bau.
\newblock An ant colony optimization approach for solving the multidimensional
  knapsack problem.
\newblock In {\em ICCIS}, pages 441--446, 2012.

\bibitem[Lin21]{Lin21}
Bingkai Lin.
\newblock Constant approximating k-clique is w[1]-hard.
\newblock In {\em STOC}, pages 1749--1756, 2021.

\bibitem[LRSW22]{LinRSW22}
Bingkai Lin, Xuandi Ren, Yican Sun, and Xiuhan Wang.
\newblock On lower bounds of approximating parameterized k-clique.
\newblock In {\em ICALP}, pages 90:1--90:18, 2022.

\bibitem[LRSW23]{LinRSW23}
Bingkai Lin, Xuandi Ren, Yican Sun, and Xiuhan Wang.
\newblock Improved hardness of approximating k-clique under {ETH}.
\newblock In {\em FOCS}, pages 285--306, 2023.

\bibitem[LRSZ20]{lokshtanov2020parameterized}
Daniel Lokshtanov, MS~Ramanujan, Saket Saurab, and Meirav Zehavi.
\newblock Parameterized complexity and approximability of directed odd cycle
  transversal.
\newblock In {\em SODA}, pages 2181--2200, 2020.

\bibitem[LS55]{lorie1955three}
James~H Lorie and Leonard~J Savage.
\newblock Three problems in rationing capital.
\newblock {\em The journal of business}, 28(4):229--239, 1955.

\bibitem[Mao23]{Mao23}
Xiao Mao.
\newblock (1-{\(\epsilon\)})-approximation of knapsack in nearly quadratic
  time.
\newblock {\em CoRR}, abs/2308.07004, 2023.

\bibitem[Mar10]{Marx10}
D{\'{a}}niel Marx.
\newblock Can you beat treewidth?
\newblock {\em Theory Comput.}, 6(1):85--112, 2010.

\bibitem[MC84]{MagazineC84}
Michael~J. Magazine and Maw{-}Sheng Chern.
\newblock A note on approximation schemes for multidimensional knapsack
  problems.
\newblock {\em Math. Oper. Res.}, 9(2):244--247, 1984.

\bibitem[MM57]{markowitz1957solution}
Harry~M Markowitz and Alan~S Manne.
\newblock On the solution of discrete programming problems.
\newblock {\em Econometrica: journal of the Econometric Society}, pages
  84--110, 1957.

\bibitem[MP22]{marx2022optimal}
D{\'a}niel Marx and Micha{\l} Pilipczuk.
\newblock Optimal parameterized algorithms for planar facility location
  problems using voronoi diagrams.
\newblock {\em {ACM} Trans. Algorithms}, 18(2):1--64, 2022.

\bibitem[MR16]{manurangsi2016birthday}
Pasin Manurangsi and Prasad Raghavendra.
\newblock A birthday repetition theorem and complexity of approximating dense
  csps.
\newblock {\em arXiv preprint arXiv:1607.02986}, 2016.

\bibitem[OM80]{OM80}
Osman Oguz and MJ~Magazine.
\newblock A polynomial time approximation algorithm for the multidimensional
  0/1 knapsack problem.
\newblock {\em Univ. Waterloo Working Paper}, 1980.

\bibitem[PW18]{pilipczuk2018directed}
Marcin Pilipczuk and Magnus Wahlstr{\"o}m.
\newblock Directed multicut is w [1]-hard, even for four terminal pairs.
\newblock {\em ACM Transactions on Computation Theory (TOCT)}, 10(3):1--18,
  2018.

\bibitem[RBEDB19]{Rezoug2019}
Abdellah Rezoug, Mohamed Bader-El-Den, and Dalila Boughaci.
\newblock {\em Hybrid Genetic Algorithms to Solve the Multidimensional Knapsack
  Problem}, pages 235--250.
\newblock Springer International Publishing, Cham, 2019.

\bibitem[Rhe15]{rhee2015faster}
Donguk Rhee.
\newblock {\em Faster fully polynomial approximation schemes for knapsack
  problems}.
\newblock PhD thesis, Massachusetts Institute of Technology, 2015.

\bibitem[RT87]{RaghavanT87}
Prabhakar Raghavan and Clark~D. Thompson.
\newblock Randomized rounding: a technique for provably good algorithms and
  algorithmic proofs.
\newblock {\em Comb.}, 7(4):365--374, 1987.

\bibitem[Sah75]{Sahni75}
Sartaj Sahni.
\newblock Approximate algorithms for the 0/1 knapsack problem.
\newblock {\em J. {ACM}}, 22(1):115--124, 1975.

\bibitem[SC14]{SabbaC14}
Sara Sabba and Salim Chikhi.
\newblock A discrete binary version of bat algorithm for multidimensional
  knapsack problem.
\newblock {\em Int. J. Bio Inspired Comput.}, 6(2):140--152, 2014.

\bibitem[Sri95]{Srinivasan95}
Aravind Srinivasan.
\newblock Improved approximations of packing and covering problems.
\newblock In {\em STOC}, pages 268--276, 1995.

\bibitem[WN67]{weingartner1967methods}
H~Martin Weingartner and David~N Ness.
\newblock Methods for the solution of the multidimensional 0/1 knapsack
  problem.
\newblock {\em Operations Research}, 15(1):83--103, 1967.

\bibitem[WS11]{WS-book}
David~P. Williamson and David~B. Shmoys.
\newblock {\em The Design of Approximation Algorithms}.
\newblock Cambridge University Press, 2011.

\bibitem[YCLZ19]{YanCLZ19}
Hong{-}Fang Yan, Ci{-}Yun Cai, De{-}Huai Liu, and Min{-}Xia Zhang.
\newblock Water wave optimization for the multidimensional knapsack problem.
\newblock In {\em ICIC}, volume 11644, pages 688--699, 2019.

\end{thebibliography}

\appendix

\newcommand{\inv}{\textsf{inv}}
\newcommand{\clause}{\textsf{clause}}
\newcommand{\cons}{\textsf{cst}}
\newcommand{\assigned}{\textsf{assigned}}
\newcommand{\relaxed}{\textsf{relaxed}}

\section{Hardness Results on CSPs}
\label{sec:SATtoRCSP}

In this section, we give a reduction from 3-SAT to 2-CSP, which is used in the proofs of \Cref{thm:RCSP} and \Cref{thm:inherent-enumerate}. Both of these proofs will use the same generic reduction, but with different instantiation of target constraint graphs and subsets associated with vertices (i.e. which specifies the ``embedding''). We stress here that both of these proofs essentially follow from previous works,~\cite{Marx10,karthik2023conditional} and~\cite{chalermsook2017gap} respectively. In the former case, we include the proof here since their proof only deals with \emph{exact} hardness and we extend it to hardness of approximation. In the latter case, we include the proof here for completeness since their statement is in a slightly different form.

We first define 2-CSP.
\begin{mdframed} \vspace{-0.1in}
\begin{definition}
	\label{def:CSP}
	{\bf $2$-Constrained Satisfaction Problem (2-CSP)}
		\begin{itemize}
			\item {\bf Input:} $\Gamma = (H,\Sigma,X)$ consisting of a constraint graph $H$, an alphabet set $\Sigma$, and constraints $X = \{X_{(u,v)}\}_{(u,v) \in E(H)}$ such that for all $(u,v) \in E(H)$ it holds that $X_{(u,v)} \subseteq \Sigma \times \Sigma$ and $X_{(u,v)} \neq \emptyset$. 
			\item {\bf Assignment:} %An {\em assignment} for $\Gamma$ is 
			A function $\lambda : V(H) \rightarrow \Sigma$. An edge $(u,v) \in E(H)$ is said to be {\em satisfied} by an assignment $\lambda$ if and
			only if $(\lambda(u), \lambda(v)) \in X_{(u,v)}$.
			\item {\bf Objective:} Find an assignment satisfying a maximum number of edges. Let $\CSP(\Gamma)$ be the maximum number of satisfied edges by an assignment for $\Gamma$.
		\end{itemize}
\end{definition} 
\end{mdframed}

%An \emph{embedding} of a set $T$ into another set $S$ is a function\footnote{Here we use $2^S$ to denote the power set of $S$.} $\psi: T \to 2^{S}$ that maps every element of $T$ to a subset of $S$. For every $u \in S$, we define $\inv^{\psi}(u) := \{v \in T \mid u \in \psi(v)\}$.
%The {\em depth} of an embedding $\psi$, denoted by $\Delta(\psi)$, is the maximum cardinality of one of the above sets:
% $\Delta(\psi) = \max_{u \in S} \left|\inv^{\psi}(u)\right|$.

To state the reduction, we also need the following notation: given a formula $\phi = (V, C)$, for every clause $c \in C$, let $\var_{\phi}(c)$ denote the set of variables appearning in $c$, and for every subset $C' \subseteq C$, let $\var_\phi(c) := \bigcup_{c \in C} \var_\phi(c)$. When $\phi$ is clear from context, we may drop it from the subscript.

We state the generic reduction, which takes in an embedding from the clause set $C$ to vertex set of the target constraint graph $H$, below in \Cref{red:3sat-2csp}. We note that this is a standard reduction used in numerous prior works on the topic (e.g.~\cite{Marx10,chalermsook2017gap,DinurM18,karthik2023conditional,baby-pih}) and can be viewed as a derandomized variant of the direct product test. The main distinction between these previous works is how the embeddings (i.e. the collection $\cC$) are chosen. Indeed, we will see later that the embeddings used in~\cite{Marx10,karthik2023conditional} and~\cite{chalermsook2017gap} are very different, leading to the different hardness results (\Cref{thm:RCSP} and \Cref{thm:inherent-enumerate}).
%In the following subsections, we will specify the embeddings and derive the final results.

\begin{mdframed} \vspace{-0.1in}
\begin{reduction}[\textbf{\textsf{3-SAT $\rightarrow $ R-CSP}} Reduction] \label{red:3sat-2csp}
Given a 3-SAT instance $\phi = (C, V)$, a graph $H$, and a collection of sets $\cC = (C_x)_{x \in V(H)}$ where $C_v \subseteq C$, the reduction produces a \textnormal{R-CSP} instance $\Pi(\phi, H, \cC) = \left(G,\Sigma,\Upsilon, \{\pi_{e, x}, \pi_{e, y}\}_{e = (x,y)\in E(G)}\right)$ as follows:
\begin{itemize}
\item $G = H$,
\item $\Sigma = \Upsilon = [2^{3 \cdot \max_{x \in V(H)} |C_x|}]$,
\item For every $x \in V(H)$, let $\Phi_x$ denote the set of partial assignments\footnote{A partial assignment on $S \subseteq V$ is simply a function $g: S \to \{0, 1\}$.} on $\var(C_x)$ that satisfy all clauses in $C_x$.
\item For every edge $e = (x, y)$, we associate $\Sigma$ with $\Phi_x$ and $\Upsilon$ with $\{0, 1\}^{\var(C_x) \cap \var(C_y)}$. Then $\pi_{e, x}$ is defined as\footnote{Note that $g|_{T}$ is the restriction of function $g$ on subset $T$.} $\pi_{e, x}(g) = g|_{\var(C_x) \cap \var(C_y)}$.
Finally, we define $\pi_{e, y}$ similarly. 
\end{itemize}
\end{reduction}
\end{mdframed}

We list a couple of useful observations. First is on the output instance size and the running time:

\begin{obs} \label{obs:3sat-to-2csp-size-time}
$|\Pi(\phi, H, \cC)| \leq 2^{O(\max_{x \in V(H)} |C_x|)} \cdot |V(H)|$, the reduction runs in time $|\Pi(\phi, H, \cC)|^{O(1)}$.
\end{obs}

Next is the completeness, which holds regardless of the graph $H$ and the subsets $C_v$'s. This can be seen by simply letting $\psi(x) = s|_{\var(C_x)}$ where $s$ denote the satisfying assignment of $\phi$.

\begin{obs} \label{obs:3sat-to-2csp-completeness}
If $\SAT(\phi) = m$, then $\MaxPar(\Pi(\phi, H, \cC)) = |V(H)|$.
\end{obs}

\subsection{Proof of \Cref{thm:RCSP}}

In this section, we prove \Cref{thm:RCSP}, which will imply the proof of \Cref{thm:RCSPeth} as a corollary. Our reduction is the same as that of \cite{karthik2023conditional}, but we need an additional (simple) argument to show that the gap between the completeness and soundness is $1 - \Omega\left(\frac{1}{\log(k)}\right)$. We summarize the main properties of the desired reduction below.

\begin{lemma} \label{thm:SATtoCSP}
For every $D \in \N$, there are constants $\mu, \theta > 0$, and a reduction \textnormal{\textsf{3-SAT$(D)$ $\rightarrow $ R-CSP}} that, given a $\textnormal{3-SAT}(D)$ instance $\phi = (V,C)$ with $n$ variables and $m$ clauses, and a parameter $k>6$ such that $k \leq \frac{n}{\log(n)}$, 
	returns %in time $ 2^{\mu \cdot \frac{n}{k}  \cdot \log(k)}$ 
	an instance $\Pi = \left(H,\Sigma,\Upsilon, \{\pi_{e, x}, \pi_{e, y}\}_{e = (x,y)\in E(H)}\right)$ of \textnormal{R-CSP}  
	which satisfies the following properties.
	\begin{enumerate}
		\item {\bf (Completeness)} If $\SAT(\phi) = m$ then $\MaxPar(\Pi) = |V(H)|$.\label{prop:c1}
		\item {\bf (Soundness)} For any $\eps \in (0, 1)$, if $\SAT(\phi) < (1 - \eps) \cdot m$, then $\MaxPar(\Pi) < \left(1-\frac{\theta \cdot \eps}{\log(k)}\right) \cdot |V(H)|$. \label{prop:c2}
		\item {\bf (Instance Size)} $|\Gamma(\phi,k)| \leq 2^{\mu \cdot \frac{n}{k}  \cdot \log k}$
		. \label{prop:c3}
		\item {\bf (Number of Variables)} $|V(H)| \leq k$ and $H$ is $3$-regular.  \label{prop:c4}
		\item {\bf (Runtime)} The reduction runs in $2^{\mu \cdot \frac{n}{k}  \cdot \log k} + n^{O(1)}$ time.
	\end{enumerate}
\end{lemma}

Before we prove \Cref{thm:SATtoCSP}, we note that it easily implies our hardness result for R-CSP. 

\begin{proof}[Proof of \Cref{thm:RCSP}]
Let $\eps, \delta$ be the constants from \Cref{GapETH} and $\mu, \theta$ be as in \Cref{thm:SATtoCSP}. We let $\zeta = \delta / \mu, \beta = \theta \cdot \eps$, and $k_0 = \max\left\{6, \lceil \mu / \delta \rceil\right\}$. Suppose for the sake of contradiction that there is an algorithm $\cA$ with guarantees as in \Cref{thm:RCSP}. We use this to solve the (gap version of) 3-SAT as follows: On input $\phi$ with $n$ variables and $m$ clauses, runs the reduction from \Cref{thm:SATtoCSP} to get an output $\Pi$, and then runs algorithm $\cA$ on $\Pi$. For any sufficiently large $n$, our choice of parameters ensure that it runs in $O(2^{\delta n})$ time. Furthermore, \Cref{thm:SATtoCSP} ensures that the algorithm can distinguish $\SAT(\phi) = m$ from $\SAT(\phi) < (1 - \eps)m$. This contradicts Gap-ETH.
\end{proof}

Assuming ETH rather than Gap-ETH and using symmetrical arguments this also gives the proof of \Cref{thm:RCSPeth}.  

To prove \Cref{thm:SATtoCSP}, we will instantiate the reduction using a {\em graph embedding} of \cite{karthik2023conditional}. We recall some definitions here for completeness. Let $H = (V(H),E(H))$ be some graph. For some $S \subseteq V(H)$ the {\em vertex-induced subgraph} of $S$ in $H$ is the graph $H[S] = (S,E[S])$ such that $$E[S] = \{(u,v) \in E(H)~|~u,v \in S\}.$$
%A {\em subgraph} $H'$ of $H$ is a graph 
We say that $H'$ is a {\em subgraph} of $H$ if there is some $S \subseteq V(H)$ such that $H' = H[S]$. 
Let $H$ be some graph and let $H_1 = (V(H_1),E(H_1)), H_2 = (V(H_2),E(H_2))$ be connected subgraphs of $H$. We say that $H_1$ and $H_2$ {\em touch} if one of the following holds.
\begin{itemize}
	\item $V(H_1) \cap V(H_2) \neq \emptyset$, or,
	\item There are $v_1 \in V(H_1)$ and $v_2 \in V(H_2)$ such that $(v_1,v_2) \in E(H)$ or $(v_2,v_1) \in E(H)$. 
\end{itemize}
%(i) $V(H_1) \cap V(H_2) \neq \emptyset$ or if (ii) there are $v_1 \in V(H_1)$ and $v_2 \in V(H_2)$ such that $(v_1,v_2) \in E(H)$. 
A {\em connected
	embedding} of a graph $G$ in a graph $H$ is a function $\psi : V(G) \rightarrow 2^{V(H)}$ that maps every
$u \in V(G)$ to a nonempty subset of vertices $\psi(u) \subseteq V(H)$ in $H$ such that $H\left[\psi(u)\right]$ is a connected subgraph of $H$ and for every edge $(u,v) \in E(G)$ the
subgraphs $H\left[\psi(u)\right]$ and $H\left[\psi(v)\right]$ touch. For every $x \in V(H)$ define $$V_x(\psi) = \{u \in V(G)~|~x \in \psi(u)\}$$ as all vertices in $V(G)$ mapped to a subgraph that contains $x$. The {\em depth} of an embedding $\psi$, denoted by $\Delta(\psi)$, is the maximum cardinality of one of the above sets: %number of vertices $v \in V(G)$ such that $\psi(v)$ contains a specific vertex from $V(H)$:
 $\Delta(\psi) = \max_{u\in V(H)} \left|V_x(\psi)\right|$. We use the following result of \cite{karthik2023conditional}. 

\begin{lemma}[\textnormal{[Theorem 3.1 in \cite{karthik2023conditional}]}]
	\label{lem:embedding}
	 There are constants $Z,L > 1$ and an algorithm \textnormal{\textsf{Embedding}} that takes as input a
	graph $G$ and an integer $k > 6$, and outputs a bipartite $3$-regular simple graph $H$ with no
	isolated vertices and a connected embedding $\psi : V(G) \rightarrow 2^{V(H)}$ such that the following
	holds.%\footnote{Another property of the algorithm not explicitly stated in \cite{karthik2023conditional} is that $|\psi(v)| = O(\log(k))$ for all $v \in V(G)$. This can slightly lower the running time of our reduction.}
	
	\begin{itemize}
		\item {\bf (Size)} $|V(H)| \leq k$.
		
		\item {\bf (Depth Guarantee)} $\Delta(\psi) \leq Z \cdot \left(1+\frac{|V(G)|+|E(G)|}{k}\right) \cdot \log(k)$.
		
		\item {\bf (Runtime)} \textnormal{\textsf{Embedding}} runs in time $\left(|V(G)|+|E(G)|\right)^{L}$. 
	\end{itemize}
\end{lemma}

With all the tools in place, we are ready to prove \Cref{thm:SATtoCSP}.

\begin{proof}[Proof of \Cref{thm:SATtoCSP}]
We use \Cref{red:3sat-2csp} with $H, \cC = (C_x)_{x \in V(H)}$ that are chosen as follows:
\begin{itemize}
\item Define the graph $G_{\clause}$ such that $V(G) = C$ and $E(G) = \{(c, c') \mid \var(c) \cap \var(c') \ne \emptyset\}$.
\item First, run the \textnormal{\textsf{Embedding}} algorithm from \Cref{lem:embedding} to produce a graph $H$ and a connected embedding $\psi: C \to 2^{V(H)}$.
\item Let $C_x = V_{x}(\psi)$ for all $x \in V(H)$.
\item Let $\Pi = \left(H,\Sigma,\Upsilon, \{\pi_{e, x}, \pi_{e, y}\}_{e = (x,y)\in E(H)}\right)$ be the output instance from \Cref{red:3sat-2csp} with the above choices.
\end{itemize}
Note that since $\phi$ is an instance of 3-SAT($D$), we have $n/3 \leq m \leq D \cdot n$.
The completeness of the reduction follows immediately from~\Cref{obs:3sat-to-2csp-completeness}, as does the size from \Cref{obs:3sat-to-2csp-size-time}. The runtime follows from \Cref{obs:3sat-to-2csp-size-time} and \Cref{lem:embedding}. 

Finally, we will prove the soundness. Suppose contrapositively that $\MaxPar(\Pi) \geq \left(1-\frac{\theta \cdot \eps}{\log(k)}\right) \cdot |V(H)|$ for $\theta = 0.01/Z$ (where $Z$ is from \Cref{lem:embedding}) and any $\eps \in (0, 1)$. Let $\varphi: V(H) \to \Sigma \cup \{\perp\}$ denote the consistent partial assignment such that $|\varphi| = \MaxPar(\Pi)$. 

Let $V(H)_{\assigned} := \{x \in V(H) \mid \varphi(x) \ne \perp\}$ and $C_{\assigned} := \{c \in C \mid \psi(c) \subseteq V(H)_{\assigned}\}$. %, and finally $V_{\assigned} := \{v \in V \mid \forall c \in C \text{ such that } v \in \var(c), c \in C_{\assigned}\}$.
% Furthermore, for every $v \in V$, let $\psi(v) $
%
We define an assignment $s: V \to \{0, 1\}$ by assigning each $s(v)$ as follows:
\begin{itemize}
\item If there exists $c \in C_{\assigned}$ such that $v \in \var(c)$, then pick one such $c_v$ and $x_v \in \psi(c_v)$ (arbitrarily) and let\footnote{Recall that $\varphi(x_v)$ can be viewed as a partial assignment on $\var(C_{x_v})$} $s(v) = (\varphi(x_v))(v)$.
\item Otherwise, assign $s(v)$ arbitrarily.
\end{itemize}

%We claim that $s$ satisfies all clauses in $C_{\assigned}$. To see that this is the case, 
Consider any clause $c \in C_{\assigned}$ and let $v_1, v_2, v_3$ denote its variable. Pick any $x \in \psi(c)$ arbitrarily; by definition of the alphabet of $x$, we have that $\varphi(x)$ satisfies $c$. We claim that $(\varphi(x))(v_j) = s(v_j)$ for all $j \in [3]$. This is true because, by our definition of $s$, we have $s(v_j) = (\varphi(x_{v_j}))(v_j)$. Meanwhile, $(c, c_{v_j})$ is an edge in $G_{\clause}$; thus, the \textnormal{\textsf{Embedding}} algorithm ensures that $E[\psi(c) \cup \psi(c_{v_j})]$ is connected. Therefore, since $\varphi$ is consistent, we must have $(\varphi(x))(v_j) = (\varphi(x_{v_j}))(v_j) = s(v_j)$. Thus, the claim holds. This implies that the clause $c$ is satisfied by $s$. As a result, we have
\begin{align*}
\SAT(\phi) \geq |C_{\assigned}|
&= m - |\{c \in C \mid \psi(c) \nsubseteq V(H)_{\assigned}\}| \\
&\geq m - \sum_{x \in V(H) \setminus V(H)_{\assigned}} |\{c \in C \mid x \in \psi(c)\}| \\
&= m - \sum_{x \in V(H) \setminus V(H)_{\assigned}} |V_x(\psi))| \\
&\geq m - (|V(H)| - |\varphi|) \cdot \Delta(\psi) \\
&\geq m - \frac{\theta \cdot \eps}{\log(k)} \cdot |V(H)| \cdot \Delta(\psi) \\
&\geq m - \frac{\theta \cdot \eps}{\log(k)} \cdot k \cdot \left(Z \cdot \left(1+\frac{n+m}{k}\right) \cdot \log(k)\right) &\geq m - \eps \cdot m,
\end{align*}
where the last inequality follows from our choices of parameters and from $n \leq m/3$. This implies that the soundness holds, which completes our proof.
\end{proof}

\subsubsection{From R-CSP to 2-CSP}
\label{sec:RCSPto2CSP}

As stated in the introduction, our result can be easily adapted to the standard Max-2-CSP version. We obtain the following result for 2-CSP. %(\Cref{thm:CSP}).

\begin{theorem}
	\label{thm:CSP}
	Assuming \textnormal{Gap-ETH},  there exist constants $\zeta,\chi > 0$ and $k_0 \in \N$, such that, for any constant $k \geq k_0$, there is no algorithm that takes in an \textnormal{2-CSP} instance $\Gamma$ with a 3-regular constraint graph $H$ such that $|V(H)| \leq k$ variables, runs in time $O\left(|\Gamma|^{\zeta \cdot \frac{k}{\log(k)}}\right)$ and distinguish between:
	\begin{itemize}
		\item {\bf (Completeness)} $\CSP(\Gamma) = |E(H)|$, and,
		\item {\bf (Soundness)} $\CSP(\Gamma) < \left(1 - \frac{\chi}{\log(k)}\right) \cdot |E(H)|$.
	\end{itemize}
\end{theorem} 

\begin{proof}[Proof of \Cref{thm:CSP}]
Let $\zeta, \beta, k_0$ be as in \Cref{thm:RCSP}. We let $\chi = 2 \beta / 3$. Suppose for the sake of contradiction that there is an algorithm $\cA$ with guarantees as in \Cref{thm:CSP}. We claim that $\cA$ can solve the problem in \Cref{thm:RCSP} as well. To see that it is correct, note that the completeness is obvious (i.e. $\MaxPar(\Pi) = |V(H)|$ iff $\CSP(\Gamma) = |E(H)|$). As for the soundness, suppose contrapositively that $\CSP(\Gamma) \geq \left(1 - \frac{\chi}{\log(k)}\right) \cdot |E(H)|$, i.e. there is an assignment $\lambda: V(H) \to \Sigma$ that violates at most $\frac{\chi}{\log(k)} 
\cdot |E(H)|$ edges. Define $\varphi: V(H) \to \Sigma \cup \{\perp\}$ by $\varphi(v) = \lambda(v)$ iff all edges adjacent to $v$ are satisfied by $\lambda$. Otherwise, we let $\varphi(v) = \perp$. It is clear from the definition that $\varphi$ is consistent. Furthermore, we have $|\varphi| \geq |V(H)| - \frac{\chi}{\log(k)} \cdot |E(H)| = \left(1 - \frac{\beta}{\log(k)}\right) \cdot |V(H)|$ where the last equality is due to the fact that $H$ is 3-regular. As such, the algorithm is correct. From \Cref{thm:RCSP}, this violates Gap-ETH.
\end{proof}

\subsection{Proof of \Cref{thm:inherent-enumerate}}
\label{sec:inherent-enumerate}

To prove \Cref{thm:inherent-enumerate}, we use the reduction from \cite{chalermsook2017gap}, restated slightly to fit in our terminologies. The guarantees of the reduction are stated below.

\begin{lemma} \label{thm:SATtoCSP-inherent-enumerate}
For every $D \in \N$ and $\eps > 0$, there are constants $\mu > 0$ and $r_0 \in \N$ such that the following holds: For any constants $k \geq r \geq r_0$, there is a reduction \textnormal{\textsf{3-SAT$(D)$ $\rightarrow $ R-CSP}} that, given a $\textnormal{3-SAT}(D)$ instance $\phi = (V,C)$ with $n$ variables and $m$, 
	returns %in time $ 2^{\mu \cdot \frac{n}{k}  \cdot \log(k)}$ 
	an instance
	$\Pi = \left(H,\Sigma,\Upsilon, \{\pi_{e, x}, \pi_{e, y}\}_{e = (x,y)\in E(H)}\right)$ of \textnormal{R-CSP}  
	which satisfies the following properties.
	\begin{enumerate}
		\item {\bf (Completeness)} If $\SAT(\phi) = m$ then $\MaxPar(\Pi) = k$.
		\item {\bf (Soundness)} If $\SAT(\phi) < (1 - \eps) \cdot m$, then $\MaxPar(\Pi) < r$.
		\item {\bf (Instance Size)} $|\Gamma(\phi,k)| \leq 2^{\mu \cdot \frac{n}{r}}$. 
		\item {\bf (Constraint Graph)} $H$ is a complete graph on $k$ vertices.
		\item {\bf (Runtime)} The reduction runs in $2^{\mu \cdot \frac{n}{k}  \cdot \log k} + n^{O(1)}$ time.
	\end{enumerate}
\end{lemma}

Before we prove \Cref{thm:SATtoCSP-inherent-enumerate}, we note that it easily implies \Cref{thm:inherent-enumerate}. 

\begin{proof}[Proof of \Cref{thm:inherent-enumerate}]
Let $\eps, \delta$ be the constants from \Cref{GapETH} and $\mu, r_0$ be as in \Cref{thm:SATtoCSP-inherent-enumerate}. We let $\zeta = \delta / \mu$. Suppose for the sake of contradiction that there is an algorithm $\cA$ with guarantees as in \Cref{thm:inherent-enumerate}. We use this to solve the (gap version of) 3-SAT as follows: On input $\phi$ with $n$ variables and $m$ clauses, runs the reduction from \Cref{thm:SATtoCSP-inherent-enumerate} to get an output $\Pi$, and then runs algorithm $\cA$ on $\Pi$. For any sufficiently large $n$, our choice of parameters ensure that it runs in $O(2^{\delta n})$ time. Furthermore, \Cref{thm:SATtoCSP-inherent-enumerate} ensures that the algorithm can distinguish $\SAT(\phi) = m$ from $\SAT(\phi) < (1 - \eps)m$. This contradicts Gap-ETH.
\end{proof}

To describe the reduction of \cite{chalermsook2017gap}, we need the notion of a \emph{disperser}. An \emph{$(m, k, \ell, r)$-disperser} (w.r.t. universe $U$ of size $m$) is a collection of $\ell$-size subsets $I_1, \dots, I_k \subseteq m$ such that, for any distinct $i_1, \dots, i_r \in [k]$, we have 
$\left|I_{i_1} \cup \cdots \cup I_{i_r}\right| \geq (1 - \eps) m$.

\begin{lemma}[\cite{chalermsook2017gap}] \label{lem:disperser}
There exist constants $C_1, C_2 > 0$ such that the following holds: For any $m \geq C_1 r \ln k$, a $(m, k, \ell, r)$-disperser exist for $\ell \leq \lceil 3m / (\eps r) \rceil$. Moreover, such a disperser can be computed in time $(2^{rk \log k} + m)^{C_2}$.
\end{lemma}

We are now ready to prove \Cref{thm:SATtoCSP-inherent-enumerate}.

\begin{proof}[Proof of \Cref{thm:SATtoCSP-inherent-enumerate}]
We use \Cref{red:3sat-2csp} with $H, \cC = (C_x)_{x \in V(H)}$ that are chosen as follows:
\begin{itemize}
\item Let $H$ be the complete graph on $k$ vertices. 
\item Let $(C_x)_{x \in V(H)}$ be an $(m, k, \ell, r)$-disperser (w.r.t. universe $C$) computed using \Cref{lem:disperser}.
\item Let $\Pi = \left(H,\Sigma,\Upsilon, \{\pi_{e, x}, \pi_{e, y}\}_{e = (x,y)\in E(H)}\right)$ be the output instance from \Cref{red:3sat-2csp} with the above choices.
\end{itemize}
Note that since $\phi$ is an instance of 3-SAT($D$), we have $n/3 \leq m \leq D \cdot n$.
The completeness of the reduction follows immediately from~\Cref{obs:3sat-to-2csp-completeness}, as does the size from \Cref{obs:3sat-to-2csp-size-time}. The runtime follows from \Cref{obs:3sat-to-2csp-size-time} and \Cref{lem:disperser} (assuming that $n$ is sufficiently larger than $k$). 

Finally, we will prove the soundness. Suppose contrapositively that $\MaxPar(\Pi) \geq r$. Let $\varphi: V(H) \to \Sigma \cup \{\perp\}$ denote the consistent partial assignment such that $|\varphi| = \MaxPar(\Pi) \geq r$.

Let $V(H)_{\assigned} := \{x \in V(H) \mid \varphi(x) \ne \perp\}$ and\footnote{Note that this is different from $C_{\assigned}$ defined in the proof of \Cref{thm:SATtoCSP}. In particular, $C^{\relaxed}_{\assigned}$ contains all clauses $c$ that belongs to $C_x$ for \emph{some} $x \in V(H)_{\assigned}$. Meanwhile $C_{\assigned}$ contains only the clauses $c$ where $x \in V(H)_{\assigned}$ for \emph{all} $x \in V(H)$ such that $c \in C_x$.} $C^{\relaxed}_{\assigned} := \bigcup_{x \in V(H)_{\assigned}} C_x$.
We then define an assignment $s: V \to \{0, 1\}$ by assigning each $s(v)$ as follows:
\begin{itemize}
\item If there exists $c \in C^{\relaxed}_{\assigned}$ such that $v \in \var(c)$, then pick one such $c_v$ and $x_v \in (\psi(c_v) \cap V(H)_{\assigned})$ (arbitrarily) and let\footnote{Recall that $\varphi(x_v)$ can be viewed as a partial assignment on $\var(C_{x_v})$} $s(v) = (\varphi(x_v))(v)$.
\item Otherwise, assign $s(v)$ arbitrarily.
\end{itemize}

%We claim that $s$ satisfies all clauses in $C_{\assigned}$. To see that this is the case, 
Consider any clause $c \in C^{\relaxed}_{\assigned}$ and let $v_1, v_2, v_3$ denote its variable. Pick any $x \in \psi(c) \cap V(H)_{\assigned}$ arbitrarily; by definition of the alphabet of $x$, we have that $\varphi(x)$ satisfies $c$. We claim that $(\varphi(x))(v_j) = s(v_j)$ for all $j \in [3]$. This is true because, by our definition of $s$, we have $s(v_j) = (\varphi(x_{v_j}))(v_j)$. Meanwhile, since $H$ is a complete graph and $\varphi$ is consistent, we must have $(\varphi(x))(v_j) = (\varphi(x_{v_j}))(v_j) = s(v_j)$. Thus, the claim holds. This implies that the clause $c$ is satisfied by $s$. As a result, we have
\begin{align*}
\SAT(\phi) \geq |C^{\relaxed}_{\assigned}|
&= \left|\bigcup_{x \in V(H)_{\assigned}} C_x\right| \geq (1 - \eps) \cdot m,
\end{align*}
where the last inequality follows from $|V(H)_{\assigned}| \geq r$ and that $(C_x)_{x \in V(H)}$ is an $(m, k, \ell, r)$-disperser. This implies that the soundness holds, which completes our proof.
\end{proof}

\end{document}